\documentclass[sigconf]{acmart}

\usepackage{balance}

\renewcommand\footnotetextcopyrightpermission[1]{} 
\usepackage{bm} 

\usepackage{graphicx}
\usepackage{tikz}
\usetikzlibrary{fit,shapes, arrows, shadows, calc, automata, topaths, positioning,matrix,backgrounds,decorations.pathmorphing,backgrounds,petri,decorations.markings,decorations.pathreplacing}
\usepackage{pgfplots}
\pgfplotsset{compat=1.16}

\usepackage{amsmath,amsthm,amsfonts}
\usepackage{multicol,enumitem}
\usepackage{caption,subcaption}
\usepackage{subfiles}
\usepackage{mathtools}
\usepackage{pdfpages}
\usepackage[toc,page]{appendix}
\usepackage{cancel, xcolor}
\usepackage{algorithm}
\usepackage[noend]{algpseudocode}
\usepackage{tabularx}

\definecolor{Purple}{HTML}{911146}

\makeatletter
\DeclareRobustCommand*\cal{\@fontswitch\relax\mathcal}
\makeatother

\newtheorem{theorem}{Theorem}[section]
\newtheorem{lemma}[theorem]{Lemma}
\newtheorem{definition}[theorem]{Definition}

\newtheorem{corollary}[theorem]{Corollary}
\newtheorem{example}[theorem]{Example}
\newtheorem{property}[theorem]{Property}
\newtheorem{remark}[theorem]{Remark}
\newtheorem{claim}[theorem]{Claim}

\newcommand{\suchthat}{\ | \ }
\newcommand{\R}{{\mathbb R}}
\newcommand{\calI}{{\cal I}}
\newcommand{\calH}{{\cal H}}
\newcommand{\calV}{{\cal V}}
\newcommand{\calE}{{\cal E}}
\newcommand{\calT}{{\cal T}}
\newcommand{\calC}{{\cal C}}
\newcommand{\vars}{\text{\sf vars}}
\newcommand{\poly}{\text{\sf poly}}
\newcommand{\polylog}{\text{\sf polylog}}
\newcommand{\fhtw}{\text{\sf fhtw}}
\newcommand{\subw}{\text{\sf subw}}
\newcommand{\ijw}{\text{\sf ijw}}

\newcommand{\faqai}{\text{\sf FAQ-AI}}
\newcommand{\ed}{\text{\sf ED}}
\newcommand{\td}{\text{\sf TD}}
\newcommand{\bcqij}{\text{\sf IJ} }
\newcommand{\bcq}{\text{\sf EJ} }
\newcommand{\query}{\text{\sf EIJ} }

\newcommand{\fpt}{\textsc{FPT}}
\newcommand{\panda}{\text{\sf PANDA}}
\newcommand{\dom}{\text{\sf Dom}}
\newcommand{\D}{\textbf{D}}
\newcommand{\Q}{\textbf{Q}}
\newcommand{\onestring}{\text{``1''}}
\newcommand{\zerostring}{\text{``0''}}

\newcommand{\nop}[1]{}

\newcommand{\punto}{\hfill\ $\Box$}

\newcommand{\arxiv}[1]{}

\definecolor{c82b366}{RGB}{130,179,102} %
\definecolor{cd5e8d4}{RGB}{213,232,212}
\definecolor{cb85450}{RGB}{184,84,80} %
\definecolor{cf8cecc}{RGB}{248,206,204}
\definecolor{c9673a6}{RGB}{150,115,166} %
\definecolor{ce1d5e7}{RGB}{225,213,231}
\definecolor{cd6b656}{RGB}{214,182,86} %
\definecolor{cfff2cc}{RGB}{255,242,204}
\definecolor{cffffff}{RGB}{255,255,255}

\definecolor{cd79b00}{RGB}{215,155,0}
\definecolor{cffe6cc}{RGB}{255,230,204}

\begin{document}

\pagestyle{empty} 
\fancyhead{}

\title{The Complexity of Boolean Conjunctive Queries with Intersection Joins}

\author{Mahmoud Abo Khamis}
\affiliation{%
    \institution{RelationalAI}
    \city{Berkeley}
    \country{United States}
}%
\author{George Chichirim}
\affiliation{%
    \institution{University of Oxford}
    \city{Oxford}
    \country{United Kingdom}
}%
\author{Antonia Kormpa}
\affiliation{%
    \institution{University of Oxford}
    \city{Oxford}
    \country{United Kingdom}
}%
\author{Dan Olteanu}
\affiliation{%
    \institution{University of Zurich}
    \city{Zurich}
    \country{Switzerland}
}


\begin{abstract}
Intersection joins over interval data are relevant in spatial and temporal data settings. A set of intervals join if  their intersection is non-empty. In case of point intervals, the intersection join becomes the standard equality join.

We establish the complexity of Boolean conjunctive queries with intersection joins by a many-one equivalence to disjunctions of Boolean conjunctive queries with equality joins. The complexity of any query with intersection joins is that of the hardest query with equality joins in the disjunction exhibited by our equivalence. This is captured by a new width measure called the ij-width.

We also introduce a new syntactic notion of acyclicity called iota-acyclicity to characterise the class of Boolean queries with intersection joins that admit linear time computation modulo a poly-logarithmic factor in the data  size. Iota-acyclicity is for intersection joins  what alpha-acyclicity is for equality joins. It strictly sits between gamma-acyclicity and Berge-acyclicity. The intersection join queries that are not iota-acyclic are at least as hard as the Boolean triangle query with equality joins, which is widely considered not computable in linear time.
\end{abstract}



\maketitle

\section{Introduction}
\label{section:introduction}


Interval data is common in spatial and temporal databases. One important type of joins on intervals is the intersection join: A set of intervals join if their intersection is non-empty. In the case of point intervals, the intersection join becomes the classical equality join.

This paper establishes the complexity of Boolean conjunctive queries with \textbf{I}ntersection \textbf{J}oins (denoted by $\bcqij$). Whereas the complexity of Boolean conjunctive queries with \textbf{E}quality \textbf{J}oins (denoted by $\bcq$) has been extensively investigated in the literature by, e.g.,~\citet{jacm/Marx13} and \citet{pods/Khamis0S17}, the complexity of the more general \bcqij queries remained open for decades. 

The key tool aiding our investigation is a many-one equivalence of any \bcqij query to a  disjunction of \bcq queries. It uses a forward reduction ($\bcqij$-to-$\bcq$) and a backward reduction ($\bcq$-to-$\bcqij$), cf\@. Figure~\ref{fig:reductions}.

\begin{figure}
\begin{tikzpicture}
\tikzstyle{node2} = [
    draw, rectangle, rounded corners = 3 pt, align=center, inner sep = 6 pt
]
\tikzstyle{path2} = [
    ->, double, line width = 1 pt,rounded corners=3pt
]
\node [node2] at (0, 0) (QD) {$(Q, \D)$};
\node [node2] at (7, 0) (tQtD) {
    $(\tilde Q_1, \tilde\D)$\\
    $\vdots$\\
    $(\tilde Q_i, \tilde\D)$\\
    $\vdots$\\
    $(\tilde Q_\ell, \tilde\D)$
};
\draw [path2] (QD)--(tQtD)
    node [at start, pos = .5, above]{Forward reduction (Thm.~\ref{theorem:reduction-correctness})}
    node [at start, pos = .5, below]{$|\tilde \D| = O(|\D|\cdot\polylog(|\D|))$};
\node [node2] at (7, -2.5) (tQtD2) {$(\tilde Q_i, \tilde \D_2)$};
\node [node2] at (0, -2.5) (QD2) {$(Q, \D_2)$};
\draw[path2] (tQtD2)--(QD2)
    node [at start, pos = .5, above]{Backward reduction (Thm.~\ref{theorem:lower-bound})}
    node [at start, pos = .5, below]{$|\D_2| = O(|\tilde\D_2|)$};
\end{tikzpicture}
\caption{Forward and backward reductions from Sections~\ref{section:ij-to-ej}
and~\ref{section:reverse-reduction} respectively. $Q$ is a query with intersection joins, $\tilde Q_1,\ldots,\tilde Q_\ell$ are queries with equality joins, $\D$ and $\D_2$ are databases of intervals, $\tilde\D$ and $\tilde\D_2$ are databases of numbers.}
\Description[Forward and backward reductions.]{Forward and backward reductions from Sections~\ref{section:ij-to-ej}
and~\ref{section:reverse-reduction} respectively. $Q$ is a Boolean query with intersection joins, $\tilde Q_1,\ldots,\tilde Q_\ell$ are Boolean queries with equality joins, $\D$ and $\D_2$ are databases of intervals, $\tilde\D$ and $\tilde\D_2$ are databases of numbers.}
\label{fig:reductions}
\end{figure}
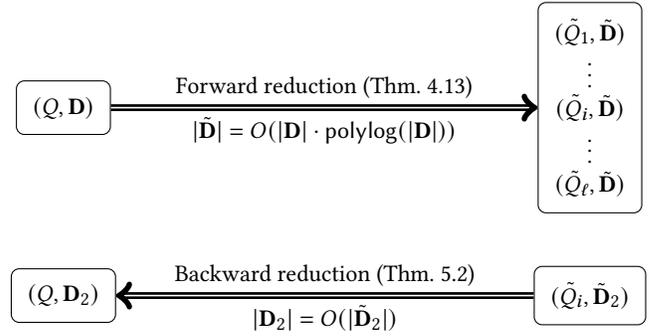

The {\em forward} reduction takes an \bcqij query $Q$ and a database $\D$ of intervals. It reduces $Q$ to a disjunction of \bcq queries, all over the same database $\tilde\D$ of numbers represented as bitstrings.  The number and size of the \bcq queries only depend on the structure of $Q$, and the size of $\tilde\D$ is within a poly-logarithmic factor from the size of $\D$. This forward reduction enables us to use any algorithms and associated runtime {\em upper bounds} for the \bcq queries as upper bounds for the \bcqij query $Q$ as well. Specifically, $Q$'s runtime is upper bounded by the maximum runtime upper bound among the generated \bcq queries.

The {\em backward} reduction takes an $\bcq$ query $\tilde Q_i$, whose structure matches that of one of the queries obtained by the forward reduction of an $\bcqij$ query $Q$ without self-joins, and an arbitrary database $\tilde\D_2$ of numbers chosen independently from $\tilde \D$ and $\D$ in the forward reduction. It then reduces $\tilde Q_i$ to an $\bcqij$ query $Q$ whose structure matches that of the original \bcqij query $Q$, and reduces $\tilde\D_2$ to  some database  $\D_2$ of intervals and size $O(|\tilde \D_2|)$. The backward reduction shows that we can use any {\em lower bounds} (i.e. hardness results) on any one of the \bcq queries $\tilde Q_i$ constructed by the forward reduction as  lower bounds on the \bcqij query $Q$. Together with the upper bounds, this implies that the \bcqij query $Q$ is {\em precisely} at the same hardness level as the hardest $\bcq$ query $\tilde Q_i$. Our forward reduction thus produces an optimal solution to $Q$ given optimal solutions to the queries $\tilde Q_i$.

The quest for the optimality of \bcq computation has a long history. The submodular width has been recently established as an optimality yardstick~\cite{jacm/Marx13,pods/Khamis0S17}. Let $\mathit{BCQ}({\cal C})$ be the decision problem: Given an \bcq query $Q\in{\cal C}$ and a database $\D$, check whether $Q(\D)$ is true. This problem is fixed-parameter tractable (\textsc{FPT}; with parameter the query size $|Q|$) if there is an algorithm solving every ${\cal C}$-instance in time $f(|Q|)\cdot |\D|^d$ for some fixed constant $d$ and any computable function $f$. The problem $\mathit{BCQ}({\cal C})$ is \textsc{FPT} if and only if every $Q\in{\cal C}$ has bounded submodular width~\cite{jacm/Marx13}.

A natural question is then what would be an optimality yardstick for \bcqij computation. We settle this question with a new width notion called the ij-width. This is the maximum submodular width~\cite{jacm/Marx13} of the \bcq queries obtained by our forward reduction. Any \bcqij query $Q$ can then be computed in time $O(N^{w} \polylog\ N)$, where $N$ is the size of the input database and  $w$ is the ij-width of $Q$. Each \bcq query created by our forward reduction is computed over databases of size $O(N \polylog\ N)$ and in time $O(N^{w} \polylog\ N)$.

The relationship between the complexities of \bcqij and \bcq is non-trivial. In general, \bcqij queries are, as expected, more expensive than their \bcq counterparts. There are simple $\alpha$-acyclic \bcqij queries whose complexity is on par with that of cyclic \bcq queries. Consider the \bcq counterparts of \bcqij queries, where we replace the intersection joins by equality joins. The triangle \bcqij query has ij-width 3/2 (Section~\ref{sec:example}), which matches the submodular width of the \bcq triangle query. Yet for the \bcqij Loomis-Whitney query with four variables, the ij-width is 5/3 while the \bcq counterpart has submodular width 4/3. 
 
This raises the question of what is the natural counterpart of $\alpha$-acyclic \bcq queries, which are known to be {\em the} linear-time computable $\bcq$ queries~\cite{vldb/Yannakakis81}. We introduce an acyclicity notion called iota ($\iota$) that syntactically characterises the class of \bcqij queries computable in linear time (modulo a polylog factor): An  \bcqij query is $\iota$-acyclic if and only if the incidence graph of its hypergraph does not have Berge cycles of length strictly greater than two. Our forward reduction maps $\iota$-acyclic \bcqij queries to an equivalent disjunction of $\alpha$-acyclic \bcq queries, so by definition the ij-width of any $\iota$-acyclic \bcqij query is one. $\iota$-acyclicity implies $\gamma$-acyclicity and is implied by Berge-acyclicity. The \bcqij queries that are not $\iota$-acyclic are at least as hard as the \bcq triangle query, which is widely considered not computable in linear time (based on the \textsc{3SUM} conjecture)~\cite{DBLP:conf/stoc/Patrascu10}. $\iota$-acyclicity is thus for \bcqij queries what $\alpha$-acyclicity is for \bcq queries.

An intersection join can also be expressed as a disjunction of inequality joins, which can be evaluated using FAQ-AI~\cite{FAQAI:TODS:2020}. However, FAQ-AI has a higher complexity than our approach: The exponents in the time complexity of FAQ-AI for the \bcqij triangle, Loomis-Whitney four, and the clique-four queries are: 2, 2, and respectively 3; whereas the ij-widths are: 3/2, 5/3, and respectively 2.

\subsection{Example: The Triangle Query}
\label{sec:example}

We introduce our approach using the Boolean triangle query, where each join is an intersection join:
\[
	Q_{\triangle} = R([A],[B]) \wedge S([B], [C]) \wedge T([A], [C])
\]

Brackets denote interval variables ranging over intervals with real-valued endpoints. The two occurrences of $[A]$ in the query denote an intersection join on $A$. An equality join on $A$ is expressed  using the variable $A$ without brackets.

A common approach first computes the join of two of the three relations and then joins with the third relation. The first join can take $O(N^2)$, where $N$ is the size of the relations.  
An equivalent encoding of $Q_{\triangle}$ using inequality joins can be computed in time $O(N^2\log^3 N)$ using FAQ-AI~\cite{FAQAI:TODS:2020} (Appendix~\ref{app:ex:triangle}). Our approach takes time $O(N^{3/2}\log^3 N)$, which matches the complexity of the \bcq triangle query (modulo polylog factor).

Our approach is based on a decomposition of the tensor representing an intersection join. In our example, we decompose the three joins as follows. We construct three segment trees: one for the intervals from $R$ and $T$ for the  interval variable $[A]$, another for the intervals from $R$ and $S$ for the interval variable $[B]$, and the third one for the intervals from $S$ and $T$ for the interval variable $[C]$. The segment tree for $O(N)$ intervals can be constructed in $O(N \log N)$ time and has depth $O(\log N)$. Its nodes represent intervals called segments. A segment includes its descendant segments and is partitioned by its child segments. A property of a segment tree is that each input interval can be expressed as the disjoint union of at most $O(\log N)$ segments. The problem of checking whether two input intervals intersect now becomes the problem of finding two segments, one per interval, that lie along the same root-to-leaf path in the segment tree. There are three ways this can happen: The two segments are the same, one segment is an ancestor of the other, or the other way around; this corresponds to the possible permutations of the two segments along a path, with one permutation also including the case where the segments are the same. A further property is that each permutation can be expressed using equality joins, as explained next. We encode each node in the segment tree as a bitstring: the empty string represents the root, the strings "0" and "1" represent the left and right child respectively, the strings "00" and "01" represent the left and right child of the "0" respectively, and so on. Given two nodes $n_1$ and $n_2$, where $n_1$ is an ancestor of $n_2$, the bitstring for $n_1$ is then a prefix of that for $n_2$. We can capture this relationship in a query with equality joins: We use one variable $A_1$ to stand for the bitstring for $n_1$, which is also a prefix of the bitstring for $n_2$, and $A_2$ to stand for the remaining bitstring for $n_2$. 

There are eight possible configurations for the ancestor-descen\-dant relationship between the segment tree nodes for the two $A$-intervals in $R$ and $S$ and similarly for $B$ and $C$. Each such case can be expressed using an \bcq query. Our query $Q_{\triangle}$ is then equivalent to the disjunction $\tilde{Q}_{\triangle} = \bigvee_{i\in[8]}\tilde{Q}_i$ of the following \bcq queries: 
{\small
\begin{align*}
	\tilde{Q}_1 &= R_{2;2}(A_1,A_2,B_1,B_2) \hspace*{-1em}&\wedge &S_{1;2}(B_1,C_1,C_2) \hspace*{-1em}&\wedge &T_{1;1}(A_1,C_1) &\hfill \\
	\tilde{Q}_2 &= R_{2;2}(A_1,A_2,B_1,B_2) \hspace*{-1em}&\wedge &S_{1;1}(B_1,C_1) \hspace*{-1em}&\wedge &T_{1;2}(A_1,C_1,C_2) &\hfill \\
	\tilde{Q}_3 &= R_{2;1}(A_1,A_2,B_1) \hspace*{-1em}&\wedge &S_{2;2}(B_1,B_2,C_1,C_2) \hspace*{-1em}&\wedge &T_{1;1}(A_1,C_1) &\hfill \\
	\tilde{Q}_4 &= R_{2;1}(A_1,A_2,B_1) \hspace*{-1em}&\wedge &S_{2;1}(B_1,B_2,C_1) \hspace*{-1em}&\wedge &T_{1;2}(A_1,C_1,C_2) &\hfill \\
	\tilde{Q}_5 &= R_{1;2}(A_1,B_1,B_2) \hspace*{-1em}&\wedge &S_{1;2}(B_1,C_1,C_2) \hspace*{-1em}&\wedge &T_{2;1}(A_1,A_2,C_1) &\hfill \\
	\tilde{Q}_6 &= R_{1;2}(A_1,B_1,B_2) \hspace*{-1em}&\wedge &S_{1;1}(B_1,C_1) \hspace*{-1em}&\wedge &T_{2;2}(A_1,A_2,C_1,C_2) &\hfill \\
	\tilde{Q}_7 &= R_{1;1}(A_1,B_1) \hspace*{-1em}&\wedge &S_{2;2}(B_1,B_2,C_1,C_2) \hspace*{-1em}&\wedge &T_{2;1}(A_1,A_2,C_1) &\hfill \\
	\tilde{Q}_8 &= R_{1;1}(A_1,B_1) \hspace*{-1em}&\wedge &S_{2;1}(B_1,B_2,C_1) \hspace*{-1em}&\wedge &T_{2;2}(A_1,A_2,C_1,C_2)
\end{align*}
}
The query is defined over new relations $R_{{\cal A};{\cal B}}$, $S_{{\cal B};{\cal C}}$, and $T_{{\cal A};{\cal C}}$ that are transformations of the original relations $R$, $S$, and $T$ to hold the bitstrings in place of the original intervals. The subscript ${\cal A}=i$ stands for the indices $1,\ldots,i$ of the variables $A_1,\ldots,A_i$ ranging over bitstrings; similarly for ${\cal B}$ and ${\cal C}$. To avoid clutter, in the remainder of this paper we will denote all such constructed new relations $R_{{\cal A};{\cal B}}$ by $\tilde{R}$ and use their schema to identify them uniquely.
For input relations of size $N$, the new relations have size $O(N\log^2 N)$, with one logarithmic factor per join interval variable. We next explain the purpose of these bitstring relations.
	A tuple $(a_{1}, a_{2}, b_{1}, b_{2})$ is in $R_{2;2}$ if and only if there is a pair of intervals $(a_R,b_R)$ in $R$ and the concatenations $a_{1}\circ a_{2}$ and $b_{1}\circ b_{2}$ reconstruct bitstrings of  nodes covered by intervals $a_R$ and respectively $b_R$. Similarly, a tuple $(b_1,c_1,c_2)$ in $S_{1;2}$ reconstructs the bitstrings $b_1$ and $c_1\circ c_2$ of nodes covered by intervals $b_S$ and $c_S$. A tuple $(a_1,c_1)$ in $T_{1;1}$ specifies the bitstrings $a_1$ and $c_1$ of nodes covered by $a_T$ and $c_T$. Therefore, $\tilde{Q}_{1}$ holds in case:
		(1) there exists a node covered by an interval $a_R$ in $R$ and a node covered by an interval $a_T$ in $T$ such that the bitstring of the latter is a prefix of the former; 
		(2) there exists a node covered by an interval $b_R$ in $R$ and a node covered by an interval $b_S$ in $S$ such that the bitstring of the latter is a prefix of the former; and
		(3) there exists a node covered by an interval $c_S$ in $S$ and a node covered by an interval $c_T$ in $T$ such that the bitstring of the latter is a prefix of the former. 
		Equivalently, $a_R$ intersects with $a_T$, $b_R$ intersects with $b_S$, and $c_S$ intersects with $c_T$, i.e., $R(a_R, b_R) \wedge S(b_S, c_S) \wedge T(a_T, c_T)$ holds.

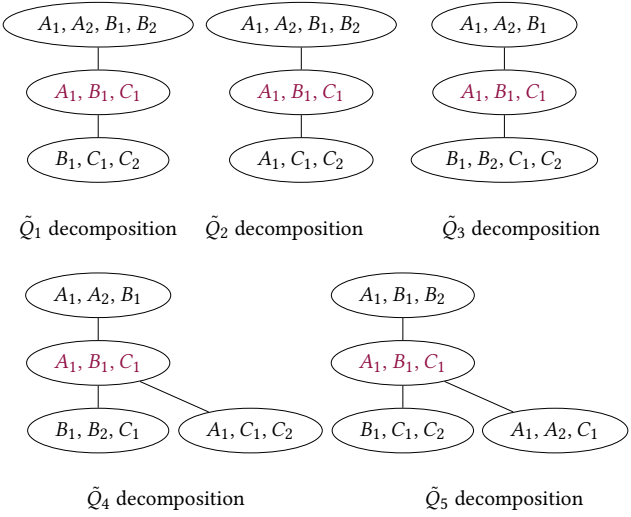
\begin{figure}[t]
\scalebox{0.9}{
\begin{tikzpicture}
	\node[draw, ellipse] (d11)  at(0, 1){$A_1$, $A_2$, $B_1$, $B_2$};
	\node[draw, ellipse] (d12)  at(0, 0){\color{Purple}$A_1$, $B_1$, $C_1$};
	\node[draw, ellipse] (d13)  at(0, -1){$B_1$, $C_1$, $C_2$};
	\draw (d11)--(d12)--(d13);
	\node at (0,-2) {$\tilde{Q}_1$ decomposition};

	\node[draw, ellipse] (d21)  at(3, 1){$A_1$, $A_2$, $B_1$, $B_2$};
	\node[draw, ellipse] (d22)  at(3, 0){\color{Purple}$A_1$, $B_1$, $C_1$};
	\node[draw, ellipse] (d23)  at(3, -1){$A_1$, $C_1$, $C_2$};
	\draw (d21)--(d22)--(d23);
	\node at (2.75,-2) {$\tilde{Q}_2$ decomposition};
	
	\node[draw, ellipse] (d31)  at(6, 1){$A_1$, $A_2$, $B_1$};
	\node[draw, ellipse] (d32)  at(6, 0){\color{Purple}$A_1$, $B_1$, $C_1$};
	\node[draw, ellipse] (d33)  at(6, -1){$B_1$, $B_2$, $C_1$, $C_2$};
	\draw (d31)--(d32)--(d33);
	\node at (6.25,-2) {$\tilde{Q}_3$ decomposition};

	\node[draw, ellipse] (d41)  at(0, -3){$A_1$, $A_2$, $B_1$};
	\node[draw, ellipse] (d42)  at(0, -4){\color{Purple}$A_1$, $B_1$, $C_1$};
	\node[draw, ellipse] (d43)  at(0, -5){$B_1$, $B_2$, $C_1$};
	\node[draw, ellipse] (d44)  at(2.25, -5){$A_1$, $C_1$, $C_2$};
	\draw (d41)--(d42)--(d43);
	\draw (d42)--(d44);
	\node at (1,-6) {$\tilde{Q}_4$ decomposition};

	\node[draw, ellipse] (d51)  at(4.5, -3){$A_1$, $B_1$, $B_2$};
	\node[draw, ellipse] (d52)  at(4.5, -4){\color{Purple}$A_1$, $B_1$, $C_1$};
	\node[draw, ellipse] (d53)  at(4.5, -5){$B_1$, $C_1$, $C_2$};
	\node[draw, ellipse] (d54)  at(6.75, -5){$A_1$, $A_2$, $C_1$};
	\draw (d51)--(d52)--(d53);
	\draw (d52)--(d54);
	\node at (6,-6) {$\tilde{Q}_5$ decomposition};

\end{tikzpicture}
}

\caption{Hypertree decompositions for five of the eight \bcq queries in $\tilde{Q}_{\triangle}$. All decompositions have a bag {\color{Purple}$\{A_1$, $B_1$, $C_1\}$}, whose materialisation requires the computation of a triangle query with equality joins. The remaining three \bcq queries admit a  decomposition similar with the first one.}
\Description[Hypertree decompositions for five of the eight \bcq queries in $\tilde{Q}_{\triangle}$.]{Hypertree decompositions for five of the eight \bcq queries in $\tilde{Q}_{\triangle}$. All decompositions have a bag {\color{Purple}$\{A_1$, $B_1$, $C_1\}$}, whose materialisation requires the computation of a triangle query with equality joins. The remaining three \bcq queries admit a  decomposition similar with the first one.}
\label{fig:triangle-decomposition}
\vspace*{-0.75em}
\end{figure}

The hypergraph of each of the eight \bcq queries admits a hypertree decomposition in the form of a star with the central bag {\color{Purple}$\{A_1,B_1,C_1\}$}, cf.\@ Figure~\ref{fig:triangle-decomposition}. In each of these decompositions, the materialisation of this bag requires solving the triangle join $R'(A_1,B_1)\wedge S'(B_1,C_1) \wedge T'(A_1,C_1)$, where $R'$ is a projection of $R_{{\cal A}; {\cal B}}$ to $A_1,B_1$ and similarly for $S'$ and $T'$. The new relations and their projections have size $O(N\log^2 N)$. The materialisation of the join takes time $O((N\log^2 N)^{3/2}) = O(N^{3/2}\log^3 N)$ using existing worst-case optimal join algorithms~\cite{jacm/NgoPRR18}. Checking whether any of the eight \bcq queries is true takes time linear in the maximum size of the bags of its decomposition. This gives an overall computation time $O(N^{3/2}\log^3 N)$ for $\tilde{Q}_{\triangle}$ and also for $Q_{\triangle}$. 

We close the example with a discussion on an alternative encoding: Instead of $\tilde{R}(A_1,A_2,B_1,B_2)$, we can use its lossless decomposition into $\tilde{R}_A(Id, A_1,A_2)$ and $\tilde{R}_B(Id, B_1,B_2)$. Here, a tuple $(i,a_1,a_2)\in\tilde{R}_A$ encodes that the $[A]$-interval in the tuple of $R$ with identifier $i$ is mapped to the bitstring $a_1\circ a_2$ in the segment tree for $[A]$. We thus avoid the  explicit materialisation of all combinations of encodings for $[A]$ and $[B]$ (and the same for the other two pairs of variables). This decomposition applies systematically to all constructed relations. Each of these new relations has size $O(N\log N)$, which is less than $O(N\log^2 N)$ in our default encoding. In general, for each $m$-way join interval variable, this encoding creates $m$ new relations regardless of whether such variables occur in the same relational atom in the query. In contrast, our default encoding creates $m^k$ new relations for each atom that contains $k$ such $m$-way join variables. Although more space efficient, this encoding comes with the same data complexity (modulo log factors) as the one used in the paper.

\section{Related Work}
\label{section:related-work}

Algorithms for intersection joins have been developed in the context of temporal~\cite{vldb/GaoJSS05} and spatial databases \cite{Mamoulis:SpatialDM:2011, tods/JacoxS07}. 
In temporal databases, tuples can be associated with intervals that represent the valid time periods. Temporal or interval joins are used to match tuples that are valid at the same time. 
In spatial databases, tuples can be associated with 2D objects that are approximated by two intervals defining minimum bounding rectangles~\cite{tods/MamoulisP01}. Spatial joins are used to find tuples with overlapping bounding rectangles. Temporal and spatial joins are thus intersection joins.
 Similarity joins under different distance metrics can be reduced to geometric containment, which is expressible using intersection joins~\cite{tods/HuYT19}.

{\noindent\bf Intersection joins.} There is a wealth of work on algorithms for intersection joins, mostly binary joins computed one at a time and over relations with 1D or 2D intervals~\cite{Mamoulis:SpatialDM:2011}. These algorithms use indices or partitioning and are typically disk-based with the objective of minimising I/O accesses. Examples of index-based algorithms include: the slot index spatial join~\cite{tkde/MamoulisP03}, the seeded tree join~\cite{sigmod/LoR94}, the R-tree join~\cite{sigmod/BrinkhoffKS93}, and relational interval tree join \cite{sigmod/EnderleHS04}. Extensions of binary joins~\cite{sigmod/BrinkhoffKS93} to multi-way joins have also been considered~\cite{tods/MamoulisP01}. Partition-based algorithms include: the partition based spati\-al-merge join~\cite{sigmod/PatelD96}, the spatial hash join~\cite{sigmod/LoR96}, the size separation spatial join~\cite{sigmod/KoudasS97}, the sweeping-based spatial join~\cite{vldb/ArgePRSV98}, and the plane-sweep method~\cite{preparata2012computational}. The partition-based algorithms can be naturally parallelised and distributed~\cite{gis/TsitsigkosBMT19,icde/PiatovHD16, pvldb/BourosM17}. These algorithms can compute two-way intersection joins in $O(N \log N + \text{OUT})$, where $\text{OUT}$ is the output size and $N$ is the input size. To sum up, there is no development on optimal algorithms for queries with intersection joins. Furthermore, most existing approaches focus on one join at a time, which can be suboptimal since they can produce intermediate results  asymptotically larger than the final result, as in the case of equality joins~\cite{Ngo:SIGREC:2013}. Our approach escapes the limitation of existing intersection join algorithms and benefits from worst-case optimal algorithms for equality joins.

{\noindent\bf Inequality joins.} 
An intersection join can be expressed as a disjunction of inequality joins: Given two intervals $[l_1,r_1]$ and $[l_2,r_2]$ defined by their starting and ending points, checking whether they intersect can be expressed as $(l_1\leq l_2 \leq r_1)\vee(l_2 < l_1 \leq r_2)$. \bcqij queries can thus be reformulated in the framework of Functional Aggregate Queries with Additive Inequalities (FAQ-AI)~\cite{FAQAI:TODS:2020}. 
The hypergraph of an FAQ-AI has two types of hyperedges: normal hyperedges, one per relation in the query and that covers the nodes representing the variables of that relation, and relaxed hyperedges, one per inequality join and that covers the variables in the inequality. This hypergraph is subject to relaxed hypertree decompositions, which are fractional hypertree decompositions~\cite{jacm/Marx13} where each normal hyperedge is covered by one bag and each relaxed hyperedge is covered by two adjacent bags. 
For a database of size $N$, an FAQ-AI can be solved in time $O(N^{\textsf{subw}_\ell}\polylog\ N)$, where $\textsf{subw}_\ell$ is the relaxed submodular width of the FAQ-AI and corresponds to the submodular width of the FAQ-AI hypergraph computed over its possible relaxed hypertree decompositions~\cite{FAQAI:TODS:2020}. 
For the triangle \bcqij query $Q_\triangle$ in Section~\ref{sec:example}, the ij-width is $\ijw(Q_\triangle)=3/2$ yet $\subw_\ell(Q_\triangle)=2$ (Appendix~\ref{app:ex:triangle}). 
Furthermore, it can be shown that $\ijw$ is lower than $\subw_\ell$ for the 
Loomis-Whitney 4, and the 4-clique \bcqij queries (see Table~\ref{tab:examples} 
for a full analysis).

\section{Preliminaries}
\label{section:preliminaries}

This section introduces notation used in the main body of the paper. 
For lack of space, further preliminaries and proofs are deferred 
to Appendix.

\paragraph{Segment Tree}

Let $\calI$ be a set of $n$ intervals. Let $p_{1}, \dots, p_{m}$ be the sequence of the distinct endpoints of the intervals in ascending order (so, $m \leq 2n$).
Consider the following disjoint intervals called elementary segments that form a partition of the real line:
$(-\infty, p_1),$ $ [p_1, p_1],$ $ (p_1, p_2),$ $ [p_2, p_2],$ $\dots $ $(p_{m-1},$ $p_{m}),$ $[p_m, p_m],$ $(p_m, +\infty)$. The {\em segment tree} $\mathfrak{T}_{\calI}$ for $\calI$ is a complete binary tree\footnote{In a complete binary tree, every level, except possibly the last, is completely filled and the nodes in the last level are positioned as far left as possible. Every node of the segment tree is thus either a leaf or an internal node with exactly two children.}, where:

\begin{itemize}

\item The leaves of $\mathfrak{T}_{\calI}$ correspond to the elementary segments induced by an order of the endpoints of the intervals in $\calI$: the leftmost leaf corresponds to the leftmost elementary segment, and so on. The elementary segment corresponding to a leaf $v$ is denoted by $\text{seg}(v)$.

\item The internal nodes of $\mathfrak{T}_{\calI}$ correspond to segments that are the union of elementary segments at the leaves of their subtrees: the segment $\text{seg}(u)$ corresponding to an internal node $u$ is the union of the elementary segments $\text{seg}(v)$ at the leaves $v$ in the subtree rooted at $u$; $\text{seg}(u)$ is thus the union of the segments at its two children.

\item Each node $v$ is associated with the \textit{canonical subset} of $v$ defined by
$\calI_{v} := \{i \in \calI \mid \text{seg}(v) \subseteq i \land \text{seg}(\text{parent}(v)) \not\subseteq i\}$.
That is, each interval $i \in \calI$ is stored into all the maximal segment tree nodes $v$ with respect to the $\text{seg}(v)$ inclusion order (or, equivalently, the nodes as high as possible in the tree) such that $\text{seg}(v) \subseteq i$.

\item Each node of a segment tree is uniquely identified by a bitstring. The root is the empty bitstring, its left child is the bistring '0', its right child has the bitstring '1', and so on. Throughout the paper, we are using a node and its corresponding bitstring interchangeably.
\end{itemize}

Let $V(\mathfrak{T}_{\calI})$ denote the set of nodes in the segment tree $\mathfrak{T}_{\calI}$. For a node $u\in V(\mathfrak{T}_{\calI})$, let $\text{anc}(u)$ be the set of ancestors of $u$ including $u$.
Given an interval (or point) $x$, $\text{leaf}(x)$ denotes the leaf that includes the left endpoint of the interval $x$ (or includes $x$, in case $x$ is a point).

\begin{definition}[Canonical Partition]
	\label{definition:canonical-partition}
	Let $\calI$ be a set of intervals and $x \in \calI$. The Canonical Partition of $x$ with respect to $\calI$
	is  $\text{CP}_{\calI}(x) := \{u \in V(\mathfrak{T}_{\calI}) \mid x \in \calI_{u}\}$.
\end{definition}

We use the following properties of a segment tree.
\begin{property}[Segment Tree] \label{property:segment-tree}
Let $\calI$ be any set of intervals and $\mathfrak{T}_{\calI}$ be the segment tree for it.
	\begin{enumerate}
		\item Let $u$ and $v$ be nodes in the segment tree. Then $u \in \text{anc}(v)$ if and only if $\text{seg}(u) \supseteq \text{seg}(v)$. Equivalently, $u$ is a prefix of $v$.

		\item For any interval $x\in \calI$, there cannot be two nodes in $\text{CP}_\mathcal{I}(x)$ such that one of them is an ancestor of the other.

		\item For any interval $x\in \calI$, $\text{CP}_{\calI}(x)$ has size and can be computed in time $O(\log \rvert \calI \lvert)$.

	\end{enumerate}
\end{property}

Since $\mathfrak{T}_{\mathcal{I}}$ is a complete binary tree with $O(\lvert \calI \rvert)$ leaves, the size of the tree is $O(\lvert \calI \rvert)$, while its height is $O(\log \lvert \calI \rvert)$. By Property~\ref{property:segment-tree}(3), it follows that $\mathfrak{T}_{\mathcal{I}}$ has size $O(\lvert \calI \rvert \cdot \log \lvert \calI \rvert)$~\cite{preparata2012computational}. 
Figure~\ref{fig:segment-tree-main} gives an example segment tree for a set of two intervals.

\tikzset{every node/.style={rectangle, align=center, font=, inner sep=3pt}}
\definecolor{green-2}{RGB}{130,179,102} %
\definecolor{red-2}{RGB}{184,84,80} %

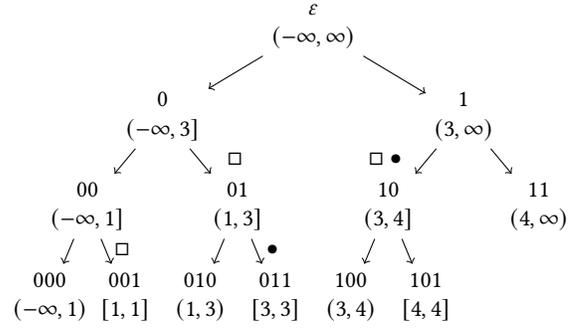
\begin{figure}
\centering
    
\begin{tikzpicture}[->,
 level 1/.style={sibling distance=40mm},
 level 2/.style={sibling distance=20mm},
 level 3/.style={sibling distance=10mm},
  level distance=12mm
]
\node [] {$\varepsilon$\\$(-\infty, \infty)$}
  	child {node [] {$0$\\$(-\infty, 3]$}   
    	child {node [] {$00$\\$(-\infty, 1]$}
      		child {node [] {$000$\\$(-\infty, 1)$}
    	}
      		child {node [, label=above:{$\square$ \color{black}}] {$001$\\$[1, 1]$}
    	}
    	}
    	child {node [, label=above:{$\square$ \color{black}}] {$01$\\$(1, 3]$}
      		child {node [] {$010$\\$(1, 3)$}
    	}
      		child {node [, label=above:{$\bullet$ \color{black}}] {$011$\\$[3, 3]$}
    	}
    	}
    }
  	child {node [] {$1$\\$(3, \infty)$}
    	child {node [, label=above:{$\square$ $\bullet$ \color{black}}] {$10$\\$(3, 4]$}
      		child {node [] {$100$\\$(3, 4)$}
    	}
      		child {node [] {$101$\\$[4, 4]$}
    	}
    	}
  		child {node [] {$11$\\$(4, \infty)$}
  		}
	};
    
\end{tikzpicture}
    \caption{Segment tree on the set of intervals $\mathcal{I} = \{\square \color{black} = [1, 4], \bullet \color{black} = [3, 4]\}$. 
	The interval $[1, 4]$ is contained in the canonical subsets of the nodes $001$, $01$, and $10$. 
	The interval $[3, 4]$ is contained in the canonical subsets of the nodes $011$ and $10$.}
    \Description[Segment tree example.]{Segment tree on the set of intervals $\mathcal{I} = \{ \square \color{black} = [1, 4], \bullet \color{black} = [3, 4]\}$. 
	The interval $[1, 4]$ is contained in the canonical subsets of the nodes $001$, $01$, and $10$. 
	The interval $[3, 4]$ is contained in the canonical subsets of the nodes $011$ and $10$.}
	\label{fig:segment-tree-main}
\end{figure}

\paragraph{Queries}

We consider queries with intersection joins and equality joins.
An intersection join is expressed using an {\em interval variable}, denoted by $[X]$, which takes as values intervals with real-valued endpoints from a finite domain $\dom([X])$.
An equality join is expressed using a {\em point variable} or variable for short, denoted by $X$, whose values are real numbers from a finite domain $\dom(X)$.
\nop{We denote variable names by upper case letters and their values by lower case letters.}
For an interval $x$, we use $x.l$ and $x.r$ to denote its left and respectively right endpoints.
Given a set $e$ of variables, $R_e \subseteq\prod_{X\in e}\dom(X)$ is a relation consisting of tuples of $|e|$ real values;
relations over intervals are defined similarly by replacing $X$ with $[X]$. A tuple $t$ with schema $e$ is a mapping of the variables in $e$ to values in their domains.
We denote by $t(X)$ the value for variable $X$ (or $[X]$) in $t$ and by $t(e')$ the set mapping variables in $e'\subseteq e$ to their values in $t$.

A (multi-)hypergraph $\mathcal{H}=(\mathcal{V}, \mathcal{E})$ has a set $\mathcal{V}$ of vertices and a multiset $\mathcal{E}\subseteq 2^{\mathcal{V}}$ of hyperedges. We label the hyperedges to distinguish between those representing the same set of vertices.
For a vertex $X \in \calV$, $\calE_X$ denotes the subset of $\calE$ that contains $X$.

\begin{definition}[Queries]
	\label{definition:bcqij}
	Given a hypergraph $\mathcal{H}=(\mathcal{V}, \mathcal{E})$, where $\mathcal{V}$ is a set of variables, a query over $\mathcal{H}$ has the form $Q = \bigwedge_{e \in \mathcal{E}} R_{e}(e)$.
	If the vertices in $\mathcal{H}$ are interval variables, then $Q$ is a {\em Boolean conjunctive query with intersection joins}, or \bcqij for short.
	If the vertices in $\mathcal{H}$ are point variables, then $Q$ is a {\em Boolean conjunctive query with equality joins}, or \bcq for short.
	If the vertices in $\mathcal{H}$ are point and interval variables, then $Q$ is a {\em Boolean conjunctive query with intersection and equality joins}, or \query for short.

    An \query $Q = \bigwedge_{e \in \mathcal{E}} R_{e}(e)$ evaluates to true if and only if there exist tuples
    $(t_e)_{e\in \calE}\in \prod_{e \in \calE} R_e$ that satisfy the following:
    \begin{itemize}
        \item $\forall [X] \in \calV$, we have $\big(\bigcap_{e \in \calE_{[X]}} t_{e}([X])) \neq \emptyset$.
        \item $\forall X \in \calV$, $t_e(X)$ is the same for all $e \in \calE_X$.
    \end{itemize}

\end{definition}

For the evaluation of an $\bcqij$ query $Q$ with hypergraph $\calH = (\calV, \calE)$ over a database $\D$, we assume without loss of generality that the schema of $\D$ is given by $\calH$; any database can be brought into this form by appropriately ensuring a bijection between the vertices in $\calV$ and attributes in $\D$ and a bijection between the hyperedges $e\in\calE$ and the relations $R_{e}$ over schema $e$ in $\D$.

Given a set $S$, a permutation of $S$ is an ordered sequence of the elements in $S$. We denote by $\pi(S)$ the set of all permutations of the elements in $S$. For a sequence $s$, $s_{i}$ denotes its $i$-th element. The concatenation of sequences $s_{1}, \dots, s_{k}$ is denoted by $s_{1} \circ \dots \circ s_{k}$.

\section{From Intersections to Equalities}
\label{section:ij-to-ej}

In this section, we show that the \bcqij evaluation problem can be reduced to the \bcq evaluation problem. This forward reduction is used to give an upper bound on the time complexity for the former problem using the complexity of the latter problem. Section~\ref{section:reverse-reduction} then presents a backward reduction to give a corresponding lower bound on the time complexity of the \bcqij evaluation problem 
(recall Figure~\ref{fig:reductions}).

\subsection{Rewriting the Intersection Predicate}
\label{section:intersection-predicate-rewriting}

At the core of \bcqij evaluation lies the non-emptiness check of the intersection of $k$ intervals $x_{1}, \dots, x_{k}\in\mathcal{I}$:
$\left(\bigcap_{i \in[k]} x_{i} \right)\neq \emptyset$. We call this check the {\em intersection predicate}. 
In this section, we show how to rewrite this predicate into an equivalent form that uses the canonical partitions of the intervals in a segment tree $\mathfrak{T}_{\mathcal{I}}$.

Since the elementary segments that correspond to the leaves of $\mathfrak{T}_\mathcal{I}$ 
form a partition of $\mathbb{R}$, for any point $p \in \mathbb{R}$ there is precisely one leaf node $\text{leaf}(p)$ such that $p \in \text{seg}(\text{leaf}(p))$. 
By Property~\ref{property:segment-tree}(1),
$\text{anc}(\text{leaf}(p)) = \{v \in V(\mathfrak{T}_\mathcal{I}) \mid p \in \text{seg}(v)\}$.
That is, the nodes whose segments contain the point $p$ are precisely the ancestors of $\text{leaf}(p)$.

\begin{lemma}[Intersection Predicate Rewriting 1]
\label{lemma:intersection-predicate-not-ordered}
For any set of intervals $S = \{x_{1}, \dots, x_{k}\} \subseteq \mathcal{I}$, 
the predicate $\left(\bigcap_{i \in[k]} x_{i} \right)\neq \emptyset$ 
is equivalent to: 
\begin{align*}
	\bigvee_{i \in[k]} 
	\left[
		 \bigvee_{(v_1, \dots, v_k) \in \text{anc}(\text{leaf}(x_i))^k} 
		 \left( 
			 \bigwedge_{\substack{j \in [k] \\ j \neq i}} v_j \in \text{CP}_\mathcal{I}(x_j) 
		 \right) 
	\right]
\end{align*}
\end{lemma}

Lemma~\ref{lemma:intersection-predicate-not-ordered} states the following. 
The intervals in $S$ intersect if and only if there is an interval $x_i \in S$ 
such that the canonical partitions of each other interval in $S$ contain an 
ancestor of $\text{leaf}(x_i)$. By construction, this leaf contains the left endpoint of $x_{i}$. 

\begin{property}
    \label{property:unique-tuple-of-nodes}
    Consider a set of intervals $S = \{x_{1}, \dots, x_{k}\} \subseteq \calI$ and a segment tree $\mathfrak{T}_{\calI}$. 
    For any $x_i \in S$, there can be at most one tuple of nodes $v_j \in \text{anc}(\text{leaf}(x_i))$  for $j\in[k], j\neq i$ 
    that satisfy the conjunction of Lemma~\ref{lemma:intersection-predicate-not-ordered}.
\end{property}

The conjunction in Lemma~\ref{lemma:intersection-predicate-not-ordered} can be satisfied by several $i$-values when there are several intervals in $S$ that have the same left endpoint. 
The database can be transformed such that any two intervals from different relations have distinct left endpoints without affecting query evaluation (Appendix~\ref{appendix:disjoint}).
If the intervals in $S$ have distinct left endpoints, then the conjunction in Lemma~\ref{lemma:intersection-predicate-not-ordered} can be satisfied by at most one $i \in [k]$, namely the one with the maximum left endpoint of the intervals in $S$; this is also the left endpoint of the interval representing the intersection of all intervals in $S$.

Consider a path from the root of the segment tree $\mathfrak{T}_{\mathcal{I}}$ down to
$\text{leaf}(x_i)$ from Lemma~\ref{lemma:intersection-predicate-not-ordered}.
The nodes $v_1, \dots, v_k$ from Lemma~\ref{lemma:intersection-predicate-not-ordered} all lie on this path. These nodes satisfy $v_j \in \text{CP}_\mathcal{I}(x_j)$ for all $j \in [k]-\{i\}$.
WLOG let $v_i = \text{leaf}(x_i)$.
Let $u_1, \ldots, u_k$ be a permutation of $v_1, \dots, v_k$ listing them in order along the path
where $u_j \in \text{anc}(u_{j+1})$ for all $j \in [k-1]$ and $u_k = v_i = \text{leaf}(x_i)$.
Let $\sigma_1, \ldots, \sigma_k$ be the corresponding permutation of the line segments $x_1, \ldots, x_k$ where
$u_j \in \text{CP}_\mathcal{I}(\sigma_j)$ for all $j \in [k-1]$ and $\sigma_k = x_i$.
Such a permutation always exists. Hence it is possible to reformulate
Lemma~\ref{lemma:intersection-predicate-not-ordered} where on the right-hand side we consider
(a disjunction over) all such permutations $\sigma_1, \ldots, \sigma_k$ (subsuming the disjunction over
$i \in [k]$) and in return we get to assume that the inner disjunction is now over {\em ordered}
$(u_1, \dots, u_k) \in \text{anc}(\text{leaf}(\sigma_k))^k$, meaning that $u_{j} \in \text{anc}(u_{j+1})$
for all $j \in [k-1]$ and $u_k = \text{leaf}(\sigma_k)$. This leads to the following variant of
Lemma~\ref{lemma:intersection-predicate-not-ordered}, which will be easier to utilize later in our reduction.

\begin{lemma} [Intersection Predicate Rewriting 2]
\label{lemma:intersection-predicate-ordered}
For any set of intervals $S = \{x_{1}, \dots, x_{k}\} \subseteq \mathcal{I}$, 
the predicate $\left(\bigcap_{i \in [k]} x_{i} \right)\neq \emptyset$
is equivalent to:
\begin{align*}
	 \bigvee_{\sigma \in \pi(S)} 
		\bigvee_{\substack{
            (u_1, \dots, u_k) \in \text{anc}(\text{leaf}(\sigma_k))^k\\
            \forall j\in[k-1]: u_{j} \in \text{anc}(u_{j+1})\\
            u_k = \text{leaf}(\sigma_k)}} 
		\left( 
			\bigwedge_{j\in[k-1]} u_j \in \text{CP}_\mathcal{I}(\sigma_j) 
		\right)
\end{align*}
\end{lemma}

By Property~\ref{property:unique-tuple-of-nodes}, for a permutation $\sigma \in \pi(S)$, 
there can be at most one tuple $(u_{j})_{j \in [k-1]}$ that satisfies the conjunction. 
However, even if the intervals in $S$ have distinct left endpoints, 
the predicate of Lemma~\ref{lemma:intersection-predicate-ordered} may be satisfied by multiple permutations. 
To see this, suppose that $\sigma$ and $(u_{j})_{j \in [k-1]}$ satisfy the predicate. 
If there is $j \in [k-1]$ such that $u_j = u_{j+1}$, then the permutation $\sigma'$, 
obtained by swapping $\sigma_j$ and $\sigma_{j+1}$ in $\sigma$, 
together with the tuple $(u_{j})_{j \in [k-1]}$, also satisfy the predicate. 
It is possible to further restrict the permutations 
such that each tuple of segment tree nodes that satisfies the conjunction corresponds to exactly one permutation
(Appendix~\ref{appendix:disjoint}). 
The equivalence in Lemma~\ref{lemma:intersection-predicate-ordered} can be alternatively expressed using the bitstrings of the nodes in the segment tree. 
By Property~\ref{property:segment-tree} (1) the expression $u_{j} \in \text{anc}(u_{j+1})$  can be equivalently stated as $u_j$ being a prefix of $u_{j+1}$. 
In other words, there exists a tuple of bitstrings $(b_{1}, \dots, b_{k})$ such that $u_{j} = b_{1} \circ \dots \circ b_{j}$ for $j\in[k]$. 
This observation leads to the following rewrite of Lemma~\ref{lemma:intersection-predicate-ordered}.

\begin{lemma}[Intersection Predicate Rewriting 3]
    \label{lemma:from-intersections-to-equalities}
    Consider a set of intervals $S = \{x_{1}, \dots, x_{k}\} \subseteq \mathcal{I}$. 
    The predicate $\left(\bigcap_{x \in S} x\right) \neq \emptyset$ is true if and only 
    if there exists a permutation $\sigma \in \pi(S)$ and a tuple of bitstrings
    $(b_{1}, \dots, b_{k})$ such that:
    \begin{itemize}
    \item $(b_{1} \circ \dots \circ b_{j})  \in \text{CP}_{\calI}(\sigma_{j})$ for $j\in[k-1]$, and
    \item $(b_{1} \circ \dots \circ b_{j}) = \text{leaf}(\sigma_{j})$ for $j = k$.
    \end{itemize}
\end{lemma}

Section~\ref{section:reduction-onestep} lifts the rewriting in Lemma~\ref{lemma:from-intersections-to-equalities} to the level of queries.

\subsection{One-Step Forward Reduction}
\label{section:reduction-onestep}

For an \bcqij query $Q$ with hypergraph $\calH=(\calV,\calE)$ and a database $\D$, the forward reduction proceeds iteratively on $Q$ and $\D$ and resolves one join interval variable at a time. Let this variable be $[X]$. The reduction yields a new query $Q_{[X]}$ and a new database $\D_{[X]}$ such that $Q_{[X]}(\D_{[X]})$ is true if and only if $Q(\D)$ is true.

The core computation needed to evaluate $Q$ over $\D$ is the intersection predicate $\left(\bigcap_{x \in S} x\right) \neq \emptyset$, where $S$ consists of one input interval per relation involved in the intersection join on $[X]$. Lemma~\ref{lemma:from-intersections-to-equalities} explains how to express this computation for any subset $S$ of an input set of intervals $\mathcal{I}$ for $[X]$ using the segment tree $\mathfrak{T}_{I}$ for $\mathcal{I}$. 

Let $k=\lvert \calE_{[X]} \rvert$ be the number of hyperedges in $\calH$ containing $[X]$. The reduction maps $[X]$ to fresh point variables $X_{1}, \dots, X_{k}$ that range over the possible bitstrings of the segment tree nodes from the canonical partitions of the intervals of $[X]$.

Given a permutation $\sigma=(\sigma_1,\ldots,\sigma_k)\in\pi(\calE_{[X]})$ of the hyperedges $\calE_{[X]}$ containing $[X]$, each hyperedge $\sigma_i$ induces a fresh hyperedge $\tilde{\sigma_i}$ that has the fresh point variables $X_1,\ldots,X_i$ in place of the original interval variable $[X]$: $\tilde{\sigma_i} = \sigma_i\setminus\{[X]\}\cup\{X_1,\ldots,X_i\}$.

\begin{definition}[One-step Hypergraph Transformation]\label{definition:hypergraph-onestepreduction}
	Gi-ven a hypergraph $\calH=(\calV,\calE)$, an interval variable $[X]$, and any permutation $\sigma\in\pi(\calE_{[X]})$, the hypergraph $\tilde{\calH}_{([X],\sigma)}$ has the set $\calV_{[X]}$ of vertices and the set $\calE_{([X],\sigma)}$ of hyperedges, where 
    $\calV_{[X]}=\calV\setminus\{[X]\}\cup\{X_{1}, \dots, X_{k}\}$ and 
	$\calE_{([X],\sigma)} = \calE\setminus\{\sigma_i\mid i\in [k]\}\cup \{\tilde{\sigma_i} \mid i\in [k]\}$.
    The set $\tilde{\calH}_{[X]}=\{\tilde{\calH}_{([X],\sigma)} \mid \sigma\in\pi(\calE_{[X]})\}$ consists of all hypergraphs created from $\calH$ by resolving the interval variable $[X]$.
\end{definition}

For a given permutation $\sigma$, there is a one-to-one correspondence between the hyperedges in $\calH$ and those in $\tilde{\calH}_{([X],\sigma)}$. We obtain as many new hypergraphs as the number of permutations of $\calE_{[X]}$. 

\begin{example}\label{ex:onestep}\em
	Let $\calH$ be a hypergraph with vertices $\{[A],[B],[C]\}$ and edges $e_1=e_2=\{[A],[B],[C]\}$ and $e_3=\{[A]\}$.

	We reduce $\calH$ by resolving the interval variable $[A]$. Since $[A]$ occurs in three edges, we create three point variables $A_1,A_2,A_3$ and consider six permutations. The new edges created for the permutation $\sigma=(e_1,e_2,e_3)$ are: $\tilde{\sigma}_1=\{A_1,[B],[C]\}$, $\tilde{\sigma}_2=\{A_1,A_2,[B],[C]\}$, and $\tilde{\sigma}_3=\{A_1,A_2,A_3\}$. For the permutation $(e_3,e_2,e_1)$, the new edges are: $\{A_1\}$, $\{A_1,A_2, [B], [C]\}$, and $\{A_1,A_2, A_3, [B], [C]\}$.\punto
\end{example}

Each hypergraph $\tilde{\calH}_{([X],\sigma)}$ defines a new query $Q_{([X],\sigma)}$ and the corresponding database $\D_{([X],\sigma)}$. We next explain how to construct the new query and the new database.

The query $Q$ is rewritten according to the new hypergraphs: For each permutation $\sigma$, we create an \query query $\tilde{Q}_{([X],\sigma)}$, which has equality joins and possibly remaining intersection joins, whose hypergraph is $\tilde{\calH}_{([X],\sigma)}$. By taking all permutations, we thus create a query $\tilde{Q}_{[X]}$ that is a disjunction of  \query queries such that each such query has one join interval variable less, namely $[X]$. Note that $Q$ need not be an \bcqij query: It may have both intersection and equality joins, for instance if it is the result of a previous rewriting step.

\begin{definition}[One-Step Query Rewriting]\label{definition:query-onestepreduction}
    \label{definition:one-step-rewriting}
	For any \query query $Q$ with hypergraph $\calH=(\calV,\calE)$, the \query query   
	$\tilde{Q}_{([X],\sigma)}$ with hypergraph $\tilde{\calH}_{([X],\sigma)} = (\calV_{[X]},\calE_{([X],\sigma)})$ is defined by:

    \[
        \tilde{Q}_{([X], \sigma)} := \bigwedge_{\tilde{e} \in \tilde\calE_{([X],\sigma)}} \tilde{R}(\tilde{e})
    \]

	The query $\tilde{Q}_{[X]}$ is the disjunction of the \query queries $\tilde{Q}_{([X],\sigma)}$ over all possible permutations $\sigma\in\pi(\calE_{[X]})$:

    \[
        \tilde{Q}_{[X]} := \bigvee_{\sigma \in \pi(\calE_{[X]})} \tilde{Q}_{([X], \sigma)}.
    \]
\end{definition}

\begin{example}\em 
	The \bcqij query $Q=$ $R([A],[B],[C])$ $\wedge S([A],[B],[C])$ $\wedge T([A])$ has the hypergraph $\calH$ in Example~\ref{ex:onestep}.
	The permutation $(e_1,e_2,e_3)$ yields $\tilde{Q}_1$ $= \tilde{R}(A_1,[B],[C]),$ $\tilde{S}(A_1,A_2,[B],[C]),$ $\tilde{T}(A_1,A_2,$ $A_3).$  
	The permutation $(e_3,e_2,e_1)$ yields $\tilde{Q}_2$ $= \tilde{R}(A_1,A_2,A_3,[B],[C]),$ $\tilde{S}(A_1,A_2,[B],[C]),$ $\tilde{T}(A_1).$
	The final query $\tilde{Q}_{[A]}$ is a disjunction of six \query queries, including $\tilde{Q}_1$ and $\tilde{Q}_2$.\punto
\end{example}

For each new hyperedge $\tilde{\sigma_i}$ in $\tilde{\calH}_{([X],\sigma)}$, there is a new relation $\tilde{R}_{\tilde{\sigma_i}}(\tilde{\sigma_i})$. To avoid clutter, we denote it by $\tilde{R}(\tilde{\sigma_i})$; its schema $\tilde{\sigma_i}$ uniquely identifies the transformation of the original relation $R$. 

\begin{definition}[One-Step Database Transformation]
    \label{definition:one-step-transformation}
    The database $\tilde{\D}_{([X],\sigma)}$ is constructed from the database $\D$ as follows.
    For each tuple $t\in R(\sigma_i)$, we construct tuples $\tilde{t}\in\tilde{R}(\tilde{\sigma_i})$ such that:
    \begin{itemize}
    	\item $t(\sigma_i\setminus\{[X]\}) = \tilde{t}(\tilde{\sigma_i}\setminus\{X_1,\ldots,X_i\})$
    	\item $(\tilde{t}(X_1)\circ\cdots\circ \tilde{t}(X_i))\in \text{CP}_{\calI}(t([X]))$ if $i\in [k-1]$
    	\item $(\tilde{t}(X_1)\circ\cdots\circ \tilde{t}(X_i)) = \text{leaf}(t([X]))$ if $i = k$
    \end{itemize}
    The relations whose schemas do not contain $[X]$ are copied from $\D$ to $\tilde{\D}_{([X],\sigma)}$.
    The new database $\tilde{\D}_{[X]}$ is the set of all relations in the databases $\tilde{\D}_{([X],\sigma)}$.  
\end{definition} 

The number of tuples $\tilde{t}$ constructed for a tuple $t$ in Definition~\ref{definition:one-step-transformation} depends on the size of the canonical partition $\text{CP}_{\calI}(t([X]))$ of $t([X])$ and on the number of ways we can partition  the bitstring of a node in the canonical partition into $i$ substrings. 
Overall, this number is poly-logarithmic in the number of input intervals $\mathcal{I}$. This is made more precise in the next lemma.

\begin{lemma} 
    \label{lemma:db-x-sigma-size}
    Each new relation $\tilde{R}(\tilde{\sigma_i})$ in database $\tilde{\D}_{([X],\sigma)}$ constructed from the database $\D$ following Definition~\ref{definition:one-step-transformation} has the size:
    $O(\lvert R(\sigma_{i}) \rvert \cdot \log^{i} \lvert \mathcal{I}\rvert)$ if $i\in[k - 1]$; and
    $O(\lvert R(\sigma_{i})\rvert \cdot \log^{i - 1} \lvert \mathcal{I} \rvert)$ if $i = k$.
    It can be constructed in time proportional to its size.
\end{lemma}

The transformations in Definitions~\ref{definition:one-step-transformation} and \ref{definition:one-step-rewriting} preserve the equivalence to the original evaluation problem: 
The result of $Q$ over $\D$ is the same as the result of $\tilde{Q}_{[X]}$ over $\tilde{\D}_{[X]}$.

\begin{lemma}
    \label{lemma:correctness-one-step}
    Given any \query query $Q$, interval variable $[X]$ in $Q$, and any database $\D$, let the \query query $\tilde{Q}_{[X]}$ and database $\tilde{\D}_{[X]}$ be constructed as per Definitions~\ref{definition:one-step-transformation} and \ref{definition:one-step-rewriting}. Then, $Q(\D)$ is true if and only if $\tilde{Q}_{[X]}(\tilde{\D}_{[X]})$ is true.
\end{lemma}

\begin{example}\label{ex:forward-reduction}
\em
We demonstrate the database transformation for the query $Q_{\Delta}$
in Section~\ref{sec:example}. Consider the interval variable $[A]$.
We convert $Q_{\Delta}$ into the disjunction of two queries with 
two new point  variables $A_{1}$ and $A_{2}$ in lieu of 
the interval variable $[A]$: 
\begin{eqnarray*}
    \tilde Q' &=& \tilde R_{1}(A_{1}, [B]) \wedge S([B], [C]) \wedge \tilde T_{1}(A_{1}, A_{2}, [C]),\\
    \tilde Q'' &=& \tilde R_{2}(A_{1}, A_{2}, [B]) \wedge S([B], [C]) \wedge \tilde T_{2}(A_{1}, [C]).
\end{eqnarray*}
Let $\D$ be the input database with relations $R, S$ and $T$ and let $N$ be the size of $\D$. By Definition~\ref{definition:one-step-transformation}, we construct a new database instance $\tilde \D_{[A]}$ that consists of $S$ and four new relations $\tilde R_{1}, \tilde T_{1}, \tilde R_{2}$ and $\tilde T_{2}$. By Lemma~\ref{lemma:correctness-one-step}, $Q_{\Delta}(\D)$ is true
if and only if $\tilde Q'(\tilde \D_{[A]})$ or $\tilde Q'' (\tilde \D_{[A]})$ are true.

Let $\calI$ be the set of all $[A]$-intervals in $R$ and $T$. We have $\lvert \calI \rvert \leq  2\cdot N$.
Let $\mathfrak{E}$ be the set of endpoints of these intervals. We have $\lvert \mathfrak{E} \rvert \leq 2 \cdot \lvert \calI \rvert$.
Assume that $\mathfrak{E} = \{0, 1, 2, \dots, k\}$ for some integer $k \leq 2\cdot N$. 
This is without loss of generality since the intersection problem does not depend on the absolute values of the end points but only on their relative positioning. 
Assume also that $k$ is a power of $2$; otherwise, replace $k$ with the smallest power of $2$ that is $\geq k$.
Let $\mathfrak{T}_{\calI}$ be the segment tree whose root corresponds to the interval $[0, k)$, and
the left and right children of each node correspond to the left and right halves of the corresponding interval respectively.
Assume without loss of generality that each interval $i \in \calI$ corresponds to a node in $\mathfrak{T}_{\calI}$; 
otherwise, $i$ can be broken down into $O(\log_2 \lvert k \rvert) = O(\log_2 N)$ intervals that correspond to nodes in $\mathfrak{T}_{\calI}$ (Property~\ref{property:segment-tree} (3)).
Each node $n$ in $\mathfrak{T}_{\calI}$ can be encoded as a binary string, cf.\@ Section~\ref{section:preliminaries}: If $n$ corresponds to the string $b$, then its left and right children correspond to the strings $``b0"$ and $``b1"$ respectively.
For an interval $i$ that corresponds to $n$ in $\mathfrak{T}_{\calI}$, let $\text{bin}(i)$ denote the binary string that encodes node $n$.
Let $i_{1}$ and $i_{2}$ be two intervals that correspond to the nodes $n_{1}$ and $n_{2}$. They intersect if and only if one of the two nodes is an ancestor of the other (or the nodes are the same). That is, one of them  contains the other (Property 3.3, part (1)). Equivalently, they intersect if and only if one of the two binary strings $\text{bin}(i_{1})$ and $\text{bin}(i_{2})$ is a prefix of the other (Property~\ref{property:segment-tree} (1)).

Define the new relations as follows (Definition~\ref{definition:one-step-transformation}):
\begin{eqnarray*}
    \tilde R_{1} &=& \{(\text{bin}([a]), [b])   \mid   ([a], [b]) \in R\}\\
    \tilde T_{2} &=& \{(\text{bin}([a]), [c])   \mid   ([a], [c]) \in T\}\\
   \tilde  T_{1} &=& \{(a_{1}, a_{2}, [c])   \mid   ([a], [c]) \in T   \wedge   a_{1}\circ a_{2} = \text{bin}([a])\}\\
   \tilde R_{2} &=& \{(a_{1}, a_{2}, [b])   \mid   ([a], [b]) \in R   \wedge   a_{1}\circ a_{2} = \text{bin}([a])\}.
\end{eqnarray*}
By Lemma~\ref{lemma:db-x-sigma-size} and the assumption that each $[A]$-interval corresponds to exactly one node of $\mathfrak{T}_{\calI}$, we have $\lvert \tilde R_{1} \rvert = \lvert R \rvert$, $\lvert \tilde T_{2} \rvert  = \lvert T \rvert$, and both $\tilde R_{1}$ and $\tilde T_{2}$ can be constructed in linear time. Furthermore,  $\lvert \tilde T_{1}\rvert = O(\lvert T \rvert \cdot \log N)$ because there are $O(\log N)$ ways to break a binary string $\text{bin}([a])$ of length $O(\log N)$ in two. $\tilde T_{1}$ can be constructed in time $O(\rvert T\lvert  \cdot\log N)$. Similarly for $\tilde R_{2}$ (Lemma~\ref{lemma:db-x-sigma-size}).
Then, $\tilde Q'(\tilde \D_{[A]})$ holds if and only if $Q(\D)$ has a satisfying assignment where the $[A]$-interval from $R$ contains the $[A]$-interval from $T$. Also, $\tilde Q''(\tilde \D_{[A]})$ holds if and only if $Q(\D)$ has a satisfying assignment where the $[A]$-interval from $R$ is contained in the $[A]$-interval from $T$.
Consequently $Q(D)$ holds if and only if $\tilde Q'(\tilde \D_{[A]})$ or $\tilde Q''(\tilde \D_{[A]})$ holds (Lemma~\ref{lemma:correctness-one-step}). \punto

\end{example}

\subsection{Full Forward Reduction}
\label{subsection:reduction}

In this section we show how to completely reduce (1) any \bcqij query to a disjunction of \bcq queries and (2) any database to a database with bitstrings in place of intervals for join interval variables. In particular, the full reduction is obtained by iteratively applying the reduction step from Section~\ref{section:reduction-onestep} for each interval variable to the result of the previous reduction step or to the input query and database in case of the first reduction step.

\begin{algorithm}
    \caption{\bcqij to \bcq Reduction\\ \textbf{Input}: \bcqij query $Q$ with hypergraph $\calH$, database $\mathbf{D}$}
    \label{algorithm:reduction}
    \begin{algorithmic}[1]
    \Procedure{Reduce}{$\tilde{\mathbf{H}}=\{\calH\}, \mathbf{Q} = \{Q\}, \tilde{\mathbf{D}} = \mathbf{D}$}
    \For{each interval join variable $[X]$ in $Q$}
		\State{$\tilde{\mathbf{H}}_0 := \tilde{\mathbf{H}}$; $\tilde{\mathbf{H}} := \emptyset$}
		\For{each $\calH\in\tilde{\mathbf{H}}_0$}
			\State{create $\tilde{\calH}_{[X]}$ from $\calH$ following Definition~\ref{definition:hypergraph-onestepreduction}}
			\State{$\tilde{\mathbf{H}} := \tilde{\mathbf{H}} \cup \tilde{\calH}_{[X]}$}
		\EndFor
		\State{$\tilde{\mathbf{Q}}_0 := \tilde{\mathbf{Q}}$; $\tilde{\mathbf{Q}} := \emptyset$}
		\For{each $Q\in\tilde{\mathbf{Q}}_0$}
			\For{each $\sigma\in\pi(\calE_{[X]})$}
				\State{create $\tilde{Q}_{([X],\sigma)}$ from $Q$ following Definition~\ref{definition:query-onestepreduction}}
				\State{$\tilde{\mathbf{Q}} := \tilde{\mathbf{Q}} \cup \{\tilde Q_{([X],\sigma)}\}$}
			\EndFor
		\EndFor
		\State{create $\tilde{\mathbf{D}}_{[X]}$ from $\tilde{\mathbf{D}}$ following Definition~\ref{definition:one-step-transformation}}
		\State{$\tilde{\mathbf{D}} := \tilde{\mathbf{D}}_{[X]}$}
    \EndFor
    \State{\textbf{return} $(\tilde{\mathbf{H}},\tilde{\mathbf{Q}},\tilde{\mathbf{D}})$}
    \EndProcedure
    \end{algorithmic}
\end{algorithm}

Algorithm~\ref{algorithm:reduction} details the reduction. The result is a triple consisting of: the set $\tilde{\mathbf{H}}$ of hypergraphs constructed by iteratively resolving the join interval variables in the input \bcqij query $Q$; the set $\tilde{\mathbf{Q}}$ of \bcq queries, with one such query per hypergraph in $\tilde{\mathbf{H}}$ ; and the database $\tilde{\mathbf{D}}$. The final query is the disjunction of the \bcq queries in $\tilde{\mathbf{Q}}$.

We further define the transformation function $\tau$ that takes any hypergraph $\calH$ to the set of hypergraphs $\tilde{\mathbf{H}}$: $\tau(\calH) = \tilde{\mathbf{H}}$, where $(\tilde{\mathbf{H}},\tilde{\mathbf{Q}},\tilde{\mathbf{D}}) = \textsc{Reduce}(\{\calH\}, \{Q\},\D)$. This is used in the following sections to define the complexity of \bcqij queries.

The next theorem states that our reduction is correct.

\begin{theorem}[Correctness]
\label{theorem:reduction-correctness}
For any $\bcqij$ query $Q$ with hypergraph $\calH$ and any database $\D$, it holds that $Q(\D)$ is true if and only if $\bigvee_{\tilde{Q}\in\tilde{\mathbf{Q}}} \tilde{Q}(\tilde{\mathbf{D}})$ is true, where $(\tilde{\mathbf{H}},\tilde{\mathbf{Q}},\tilde{\mathbf{D}}) = \textsc{Reduce}(\{\calH\}, \{Q\},\D)$.
\end{theorem} 

\subsection{Complexity of \bcqij Queries}
\label{subsection:complexity}

We give the data complexity of \bcqij queries using the reduction to \bcq queries from Section~\ref{subsection:reduction}. We next define a new width measure for \bcqij queries called the ij-width using the submodular width of the \bcq queries obtained in the full reduction  (Definition~\ref{definition:subw}).
\begin{definition}[ij-width]
\label{definition:ij-width}
For any hypergraph $\mathcal{H}$, the \textit{ij-width} of $\mathcal{H}$ is defined as follows:
\begin{align*}
\ijw(\mathcal{H}) 
:= 
\max_{\tilde{\mathcal{H}} \in \tau(\mathcal{H})} 
\subw(\tilde{\mathcal{H}})
\end{align*}
\end{definition}

The complexity of a given \bcqij query is that of the most expensive \bcq query constructed by the full reduction. This justifies taking the maximum in the definition of the ij-width. The optimality yardstick for the evaluation of \bcq queries is given by the submodular width~\cite{jacm/Marx13} (see also discussion in Section~\ref{section:introduction}), which justifies the use of this width measure in the definition of the ij-width. 


Let $\calV'\subseteq\calV$ be the set of join interval variables in the query and $k_{[X]}=|\calE_{[X]}|$ be the number of hyperedges containing the interval variable $[X]$.
Our full reduction constructs up to 
$\prod_{ [X]\in\calV'} k_{[X]}$ new relations and up to 
$\prod_{[X]\in\calV'} k_{[X]}!$
\bcq queries. By Lemma~\ref{lemma:db-x-sigma-size}, each new relation has size $O(N \polylog\ N)$, where $N$ is the size of the input relations. The number of constructed \bcq queries only depends on the structure of $Q$.

\begin{theorem}\label{theorem:complexity}
    Given any \bcqij query $Q$ with hypergraph $\calH$ and database $\D$, $Q(\D)$ can be computed in time $O(|\D|^{\ijw(\calH)}\cdot\polylog |\D|)$.
\end{theorem}

Table~\ref{tab:example-queries} gives the time complexities for $\faqai$~\cite{FAQAI:TODS:2020} and our approach for the cyclic \bcqij queries: triangle, Loomis Whitney with 4 variables, and 4-clique. Appendix~\ref{appendix:examples} details the $\faqai$ evaluation of a reformulation of the triangle \bcqij query using inequality joins.

\begin{table}[t]
    \begin{tabularx}{\columnwidth}{lll}
        \toprule
        \bcqij Query &
        $\faqai$~\cite{FAQAI:TODS:2020} &
        Our approach\\
        \midrule
        Triangle &
        $O(N^2 \log^3 N)$ &
        $O(N^{3/2} \log^3 N$)\\
        Loomis-Whitney-4 &
        $O(N^2 \log^k N)$, for $k \geq 9$ &
        $O(N^{5/3} \log^8 N)$ \\
        4-clique &
        $O(N^3 \log^k N)$, for $k \geq 5$ &
        $O(N^2 \log^8 N)$\\
        \bottomrule
    \end{tabularx}
    \caption{Our approach versus $\faqai$ for three \bcqij queries.}
    \label{tab:example-queries}
\end{table}

\section{From Equalities to Intersections: Reduction optimality}
\label{section:reverse-reduction}

In the previous section, we showed how to reduce the evaluation of an \bcqij query $Q$ over a database $\D$ of intervals to a disjunction $\tilde\Q$ of \bcq queries over a database $\tilde\D$ of numbers, where $|\tilde \D| = O(|\D|\cdot \polylog |\D|)$.
This proves that the runtime on $Q$ is upper bounded by the maximum upper bound over all
queries in $\tilde\Q$ (within a polylog factor).
In this section, we do the opposite. We show that the runtime on $Q$ is also lower bounded
by the maximum lower bound over all queries in $\tilde\Q$.
We start with an $\bcq$ query $\tilde Q$,
whose query structure matches that of one of the queries in $\tilde \Q$, and 
with an {\em arbitrary} database $\tilde\D_2$ over the schema of $\tilde Q$. The values in $\tilde \D_2$ are numbers and can be chosen independently from $\tilde \D$ and $\D$ in the forward reduction.
We show how to reduce $\tilde Q(\tilde\D_2)$ to $Q(\D_2)$, where $Q$ is an $\bcqij$ query, whose hypergraph matches that of the original $Q$, and $\D_2$ is some database with intervals and $O(|\tilde \D_2|)$  size. See Figure~\ref{fig:reductions}.

\begin{example}
\em
Consider the \bcqij query $Q_\triangle$ from Example~\ref{sec:example}.
WLOG let's take the \bcq query $\tilde{Q}_3$ that results from the reduction:
\[
	\tilde{Q}_3 = R_{2;1}(A_1,A_2,B_1) \wedge S_{2;2}(B_1,B_2,C_1,C_2) \wedge T_{1;1}(A_1,C_1)
\]
Consider an {\em arbitrary} database 
$\tilde\D = \{\tilde\D_{R_{2;1}}, \tilde\D_{S_{2;2}}, \tilde\D_{T_{1;1}}\}$
over the schema
$\{R_{2;1}, S_{2;2}, T_{1;1}\}$.
We can reduce solving $\tilde Q_3(\tilde\D)$ to solving query $Q_\triangle$ over another database 
$\D=\{\D_{R}, \D_S, \D_T\}$ (whose values are intervals) constructed as follows.
Let $F$ be a function that maps binary strings $\{0, 1\}^*$ into intervals $[x, y)$
for $0 \leq x \leq y \leq 1$, that is defined recursively:
$F(\varepsilon) = [0, 1), F(\text{``0''}) = [0, 1/2), F(\text{``1''}) = [1/2, 1), F(\text{``00''}) = [0, 1/4)$ and so on.
Namely for any given binary string $b$, $F(b\circ \text{``0''})$ and $F(b \circ \text{``1''})$ correspond
to the first and second half of $F(b)$ respectively.
WLOG we can assume that the domain of $\tilde \D$ is $\{0, 1\}^d$, i.e., the set of binary strings
of length $d$ for some fixed constant $d$. Construct $\D_R, \D_S,$ and $\D_T$ as follows:
\begin{eqnarray*}
    \D_R &:=& \{(F(a_1 \circ a_2), F(b_1)) \suchthat (a_1, a_2, b_1) \in \tilde \D_{R_{2;1}}\},\\
    \D_S &:=& \{(F(b_1 \circ b_2), F(c_1 \circ c_2)) \suchthat (b_1, b_2, c_1, c_2) \in \tilde \D_{S_{2;2}}\},\\
    \D_T &:=& \{(F(a_1), F(c_1)) \suchthat (a_1, c_1) \in \tilde \D_{T_{1;1}}\}.
\end{eqnarray*}
We can show that $Q_\triangle(\D)$ holds if and only if $\tilde Q_3(\tilde\D)$.

Moreover $|\tilde \D| = |\D|$.
This basically proves that solving $Q_\triangle$ is {\em at least} as hard as solving $\tilde Q_3$.
The same holds for all queries $\{\tilde Q_1, \ldots, \tilde Q_8\}$.
In contrast, the forward reduction shows that $Q_\triangle$ is {\em at most} as hard as solving the
hardest query among $\{\tilde Q_1, \ldots, \tilde Q_8\}$. Together, this implies that $Q_\triangle$ is
{\em exactly} as hard as the hardest query among $\{\tilde Q_1, \ldots, \tilde Q_8\}$, meaning that our
forward reduction is actually tight.
See Appendix~\ref{appendix:reverse:reduction} for more details.
\punto
\end{example}

\begin{theorem}
    \label{theorem:lower-bound}
	Let $Q$ be any self-join-free $\bcqij$ query with hypergraph $\calH$. Let $\tilde{Q}$ be any $\bcq$ query whose hypergraph is in $\tau(\calH)$. For any database $\tilde\D$, let $\Omega(T(|\tilde\D|))$ be a lower bound on the time complexity for computing $\tilde{Q}$, where $T$ is a function of the size of the database $\tilde\D$.
	There cannot be an algorithm $\mathcal{A}_Q$ that computes $Q(\D)$ in time $o(T(|\D|))$ (i.e., asymptotically strictly smaller), for any database $\D$.
\end{theorem}

\section{Iota-Acyclicity}
\label{section:iota-acyclicity}

In this section, we answer the following question: Which \bcqij queries can be computed in linear time (modulo a polylog factor)? To answer this question, we introduce a new notion of acyclicity, called $\iota$-acylicity, which captures precisely the linear-time computable \bcqij queries. In other words, $\iota$-acyclicity is for \bcqij queries what $\alpha$-acyclicity is for \bcq queries.

\begin{definition}[Iota Acyclic Hypergraph] 
\label{definition:iota-acyclicity}
A hypergraph $\mathcal{H}$ is $\iota$-acyclic if and only if each hypergraph in 
$\tau(\mathcal{H})$ is $\alpha$-acyclic.
\end{definition}

It is immediate to see why Definition~\ref{definition:iota-acyclicity} defines the hypergraphs of {\em some} linear-time computable \bcqij queries, namely those computed via our reduction. Theorem~\ref{theorem:hardness-of-non-iota-acyclic} later shows that Definition~\ref{definition:iota-acyclicity} defines in fact {\em all} linear-time computable \bcqij queries.

Since all hypergraphs in $\tau(\mathcal{H})$ are $\alpha$-acyclic, they correspond to \bcq queries that can be computed in linear time~\cite{vldb/Yannakakis81}. Furthermore, the size of $\tau(\mathcal{H})$ is independent of the input database and only depends on $\mathcal{H}$. Definition~\ref{definition:iota-acyclicity} defines $\iota$-acyclicity indirectly using our reduction. We next show that this is equivalent to a simple syntactic characterisation of the hypergraph of the given \bcqij query.

\begin{definition}[Berge Cycle~\cite{jacm/Fagin83}]
    \label{definition:berge-cycle-main}
    A Berge cycle in $\mathcal{H}$ is a sequence $(e^1, v^{1}$, $e^2, v^{2}, \dots$, $e^n, v^{n}, e^{n+1})$ such that:
    $v^{1}, \dots, v^{n}$ are distinct vertices in $\mathcal{V}$;
    $e^1, \dots, e^n$ are distinct hyperedges in $\mathcal{E}$ and $e^{n+1} = e^1$;
    $n \geq 2$; and $v^{i}$ is in $e^i$ and $e^{i+1}$ for each $1 \leq i \leq n$.
\end{definition}

\begin{theorem}[Iota Acyclicity Characterisation]
    \label{theorem:iota-acyclicity}
    A hypergraph is $\iota$-acyclic if and only if it has no Berge cycle of length strictly greater than two.
\end{theorem}

The smallest Berge cycle that makes a hypergraph ${\cal H}$ not $\iota$-acyclic has length three. This is a sequence $(e^1,v^1,e^2,v^2,e^3,v^3,e^4)$, where $v^1,v^2,v^3$ are distinct vertices in the hypergraph, $e^1,e^2,e^3$ are distinct hyperedges in the hypergraph, $e^4=e^1$, and $v^i\in e^i\cap e^{i+1}$ for $1\leq i\leq 3$ (Definition~\ref{definition:berge-cycle-main}).

As a corollary of Theorem~\ref{theorem:iota-acyclicity}, $\iota$-acyclicity strictly sits between Berge-acyclicity and $\gamma$-acyclicity~\cite{jacm/Fagin83, csur/Brault-Baron16}. A further corollary is that each $\iota$-acyclic hypergraph is also $\alpha$-acyclic.

\begin{corollary}
    \label{corollary:iota-between-berge-and-gamma}
    The class of $\iota$-acyclic hypergraphs is a strict superset of the class of 
    Berge-acyclic hypergraphs and it is a strict subset of the class of $\gamma$-acyclic hypergraphs.
\end{corollary}

\begin{figure}
    \centering 
    \begin{subfigure}{0.2\textwidth}
      \includegraphics[width=\linewidth]{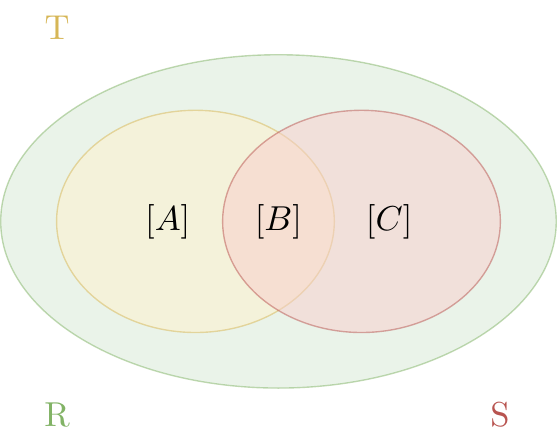}
      \caption{}
      \label{fig:non-iota-hyp-3-main}
    \end{subfigure}
    \begin{subfigure}{0.25\textwidth}
      \includegraphics[width=\linewidth]{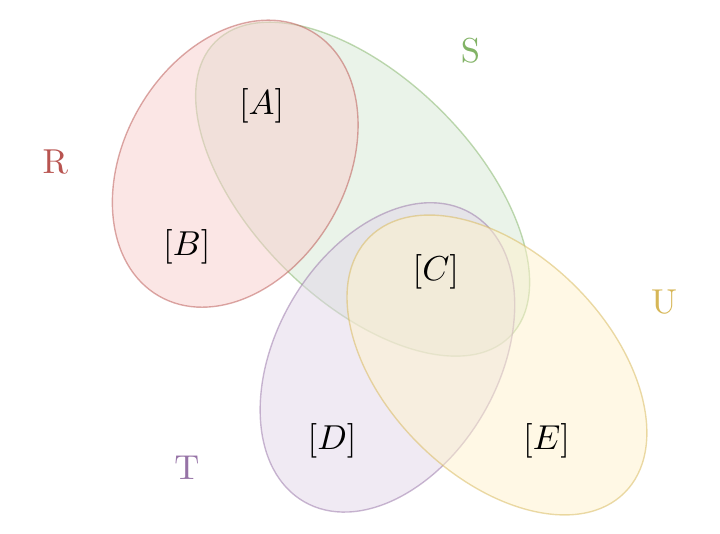}
      \caption{}
      \label{fig:iota-hyp-2-main}
    \end{subfigure}\hfil
    \caption{(a) Hypergraph with Berge cycle of length three and thus not $\iota$-acyclic. (b) Hypergraph without a Berge cycle and thus $\iota$-acyclic. } 
    \Description[Example hypergraphs.]{Example hypergraphs.}
    \label{fig:iota-example-hypergraphs-main}
    \end{figure}

\begin{example}\em
	The hypergraph of the query
    	$Q = R([A], [B], [C])$ $\wedge S([A], [B], [C]) \wedge T([A], [B])$  is not $\iota$-acyclic since it has the following Berge cycle of length three: $R-[C]-S-[B]-T-[A]-R$. 
    It becomes $\iota$-acyclic by removing any of its vertices or hyperedges.
	The hypergraph of the query  $Q = R([A], [B], [C]) \wedge S([A], [B], [C]) \wedge T([A])$ is $\iota$-acyclic since it has no Berge cycle of length strictly greater than two. 
    It only has three Berge cycles of length two. Those are the cycles $R - [A] - S - [B] - R$, $R - [B] - S - [C] - R$, $R - [A] - S - [C] - R$.

        We now turn to the two hypergraphs in Figure~\ref{fig:iota-example-hypergraphs-main}. The hypergraph in Figure~\ref{fig:non-iota-hyp-3-main} has a Berge cycle of length $3$: $R-[A]-T-[B]-S-[C]-R$.
Applying the translation from intersection joins to equality joins produces $2! \cdot 3! \cdot 2! = 24$ hypergraphs. We next analyse the reduced versions of the hypergraphs, where we drop the vertices that appear in one hyperedge only. This reduction is justified in our analysis of the $\iota$-acyclicity, since such vertices cannot contribute to a Berge cycle. There are only three distinct reduced hypergraphs $\calH_1$, $\calH_2$ and $\calH_3$. We next analyse their widths.

The hypergraph $\calH_1$ has hyperedges
$\tilde R(A_1, B_1, C_1)$, 
$\tilde S(B_1, C_1, B_2)$ and
$\tilde T(A_1, B_1, B_2)$.
Its fractional hypertree width is $1.5$. This is obtained using a hypertree decomposition consisting of the bag 
$\{A_1, B_1,$ $C_1, B_2\}$. This bag covers all hyperedges 
and has a fractional edge cover number of $1.5$. This is obtained by assigning the weights $[0.5, 0.5, 0.5]$  to the hyperedges. 
	
$\calH_2$ has hyperedges
$\tilde R(A_1, B_1, C_1, B_2)$, 
$\tilde S(B_1, C_1, B_2)$ and 
$\tilde T(A_1, B_1)$.
Its fractional hypertree width is $1.0$. This is obtained using a hypertree tree decomposition consisting of the bag $\{A_1, B_1, C_1, B_2\}$. This bag covers all hyperedges 
and has a fractional edge cover number of $1.0$. This is obtained by assigning the weights $[1.0, 0.0, 0.0]$ to the hyperedges. 

$\calH_3$ has hyperedges 
$\tilde R(A_1, B_1, C_1, B_2)$,
$\tilde S(B_1, C_1)$ and 
$\tilde T(A_1, B_1, B_2)$.
Its fractional hypertree width is $1.0$, witnessed by the same hypertree decomposition as for $\calH_2$. 

The submodular width is the same as the fractional hypertree width for all three hypergraphs. We conclude that the $\ijw$ is $3/2$, wich is the maximum of the above three fractional hypertree widths.
Therefore, our approach takes time $O(N^{3/2} \cdot \polylog N)$ for the $\bcqij$ query with the hypergraph in Figure~\ref{fig:non-iota-hyp-3}.

The hypergraph in Figure~\ref{fig:iota-hyp-2-main} has no Berge cycle. Applying the reduction produces 
	$2! \cdot 1! \cdot 3! \cdot 1! \cdot 1! = 12$ hypergraphs that are all $\alpha$-acyclic. Our approach thus take time $O(N \cdot \polylog N)$.
    \punto
\end{example}

\bcqij~queries whose hypergraphs are not $\iota$-acyclic cannot be computed in linear time (unless the 3SUM conjecture fails). This is shown by a reduction from the problem of computing the triangle \bcq query, which takes super-linear time~\cite{tods/KhamisNRR16} unless the \textsc{3SUM} conjecture fails~\cite{DBLP:conf/stoc/Patrascu10}. The \textsc{3SUM} problem asks, given a set $S$ of $n$ numbers, to find distinct $x,y,z \in S$ such that $x + y = z$. The problem can be solved in $O(n^2)$ time, and it is a long-standing conjecture that this quadratic complexity is essentially the best possible. 

\begin{theorem}[Iota Acyclicity Dichotomy]~\label{theorem:hardness-of-non-iota-acyclic}
    Let $Q$ be any \bcqij~query with hypergraph $\mathcal{H}$ and let $\D$ be any database.

    If $\mathcal{H}$ is $\iota$-acyclic, then $Q$ can be computed in time $O(\lvert \D \rvert\cdot \polylog \lvert \D \rvert)$.
    
    If $\mathcal{H}$ is not $\iota$-acyclic, then there is no algorithm 
    that can compute $Q$ in time $O(\lvert \D \rvert^{4/3-\epsilon})$ for $\epsilon>0$, unless the \textsc{3SUM} conjecture fails.     
\end{theorem}

\section{Conclusion and Future Work}

This paper pinpoints the complexity of Boolean conjunctive queries with intersection joins and characterises syntactically the class of such queries that can be computed in linear time modulo a poly-logarithmic factor. Core to our approach is a reduction of the evaluation problem for such queries to Boolean queries with equality joins. This reduction is robust: It also works for non-Boolean queries with both intersection and equality joins.

A natural extension of this work is to refine the acyclicity notion in the presence of both intersection joins and equality joins. This notion necessarily lies between $\alpha$-acyclicity and $\iota$-acyclicity: It is the former when all joins are equality joins, as in the literature, and it is the latter when all joins are intersection joins, as in this paper. A further type of join that is naturally supported by the development in this paper is the {\em membership join}: This can be expressed by using a join variable to range over both intervals and points. Our reduction can be optimised to accommodate membership joins, in addition to intersection and equality joins. Characterising the linear-time computable Boolean queries with all three types of joins is an exciting venue of future research.

\begin{acks}
This project has received funding from the European Union's Horizon 2020 research and innovation programme under grant agreement No 682588.
AK gratefully acknowledges support from EPSRC via a CASE grant supported by Ordnance Survey.
\end{acks}

\bibliographystyle{ACM-Reference-Format}
\bibliography{amain}

\newpage
\appendix
\onecolumn

\section{Background}
\label{section:background}

A \textit{hypergraph} is a generalisation of a graph in which an edge can connect any number of vertices. 

\begin{definition}[(Multi-)Hypergraph]
\label{definition:hypergraph}
A hypergraph $\mathcal{H}$ is a pair $(\mathcal{V}, \mathcal{E})$, 
where $\mathcal{V}$ is a finite set of vertices and $\mathcal{E}$ is a set of non-empty subsets of $\mathcal{V}$ called hyperedges, i.e., $\mathcal{E} \subseteq 2^\mathcal{V} \setminus \{\emptyset\}$, where $2^\mathcal{V}$ is the power set of $\mathcal{V}$.
A multi-hypergraph is a hypergraph where several hyperedges may be the same set of vertices, i.e., $\calE$ is a multiset.
\end{definition}

Queries and database schemas are associated with a hypergraph in the following way: Each attribute in the schema is associated with a vertex of the hypergraph and each relation is associated with a hyperedge of the hypergraph. 
Properties of queries and database schemas can be studied on their associated hypergraphs~\cite{jacm/Fagin83, jacm/BeeriFMY83}.

\begin{figure}[t]
	\centering 
	\includegraphics[width=.5\linewidth]{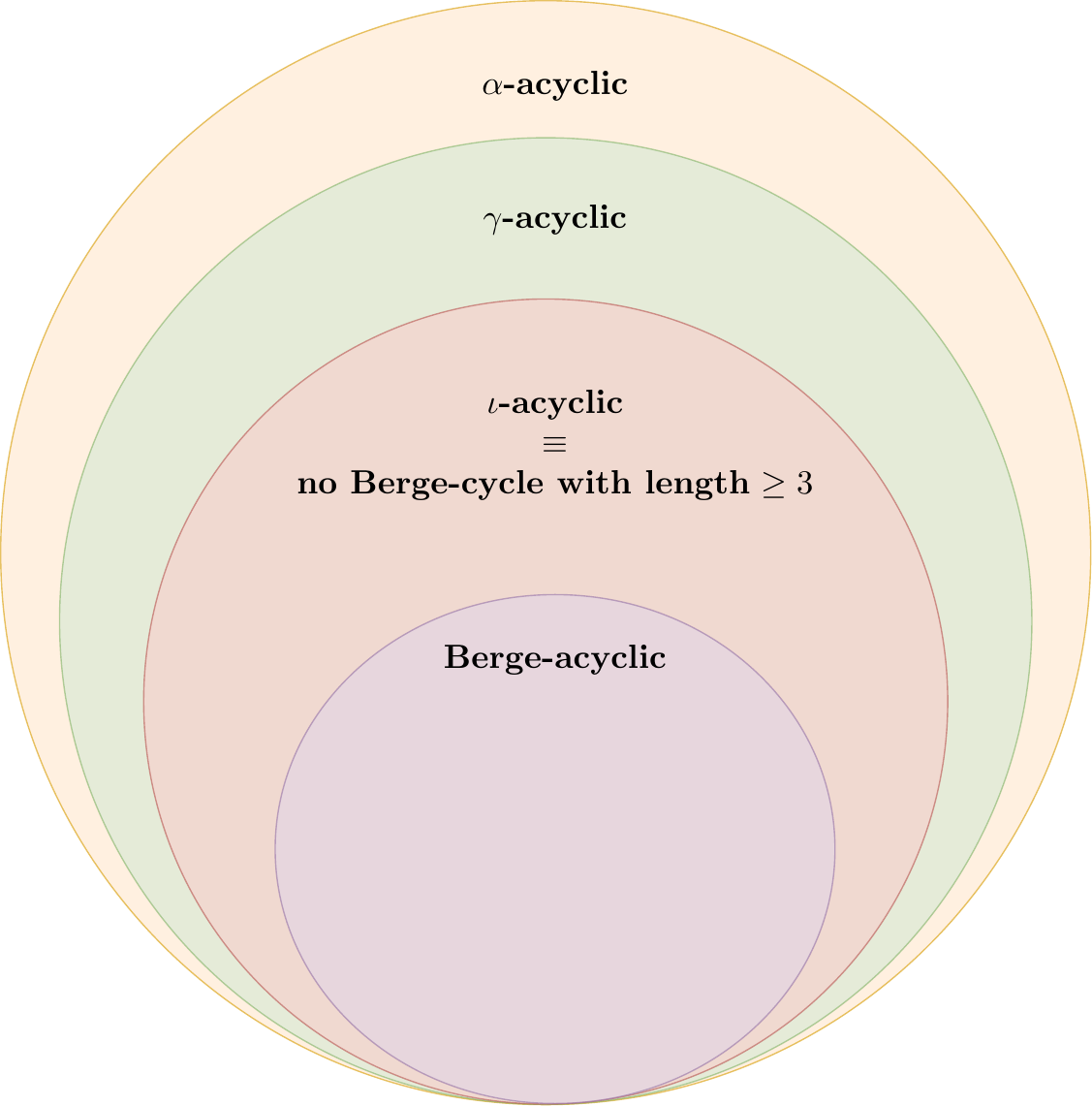}
	\caption{Venn diagram for various notions of acyclicity. Iota acyclicity is the new notion introduced in this work.}
	\label{figure:acyclicity-classes}
	\Description[]{}
\end{figure}

\subsection{Hypergraph Acyclicity}
\label{subsection:hypergraph-acyclicity}

A database schema is \textit{acyclic} if its hypergraph is acyclic. There are several notions of acyclicity: \textit{alpha, beta, gamma} and \textit{Berge}. The relationship between them is depicted by the Venn diagram in Figure~\ref{figure:acyclicity-classes}.
In the following, we describe in detail the notions of acyclicity that are relevant to this work.

\subsubsection{Berge acyclicity} 
A hypergraph $\mathcal{H} = (\mathcal{V}, \mathcal{E})$ can be represented by its \textit{incidence graph}. This is the bipartite graph $(\mathcal{V}, \mathcal{E}, \mathcal{F})$, where $\mathcal{V}$ and $\mathcal{E}$ are the partitions of the vertices of the graph and $\mathcal{F}$ is a set of edges such that $v$ and $e$ are connected with an edge in case the vertex $v$ is contained in a hyperedge $e$ in $\mathcal{E}$. 
A hypergraph is Berge-acyclic if its incidence graph is acyclic~\cite{bergegraphs}. Berge-acyclicity was subsequently expressed using the notion of Berge cycle instead of the notion of cycle in the incidence graph of $\mathcal{H}$~\cite{jacm/Fagin83}.

In the following we use the latter definition.

\begin{definition}[Berge Cycle~\cite{jacm/Fagin83}]
\label{definition:berge-cycle}
A Berge cycle in $\mathcal{H}$ is a sequence $(e^1, v^{1}$, $e^2, v^{2}, \dots$, $e^n, v^{n}, e^{n+1})$ such that:
\begin{itemize}
	\item $v^{1}, \dots, v^{n}$ are distinct vertices in $\mathcal{V}$;
	\item $e^1, \dots, e^n$ are distinct hyperedges in $\mathcal{E}$ and $e^{n+1} = e^1$;
	\item $n \geq 2$, that is, there are at least $2$ hyperedges involved;  and
	\item $v^{i}$ is in $e^i$ and $e^{i+1}(1 \leq i \leq n)$.
\end{itemize} 
\end{definition}

\begin{definition}[Berge Acyclic Hypergraph~\cite{jacm/Fagin83}]
A hypergraph is Berge-acyclic if it has no Berge cycle. 
\end{definition}

The equivalence between a Berge cycle in $\mathcal{H}$ and a cycle in the incidence graph of $\mathcal{H}$ is natural: the traversal of a cycle in the incidence graph of $\mathcal{H}$ is a sequence of alternating vertices from the partitions $\mathcal{V}$ and $\mathcal{E}$ such that the only repeated vertices are the first and last. Since the incidence graph is a bipartite graph, the minimum length of a cycle is 4, so there are at least two hyperedges from the partition $\mathcal{E}$ involved.

\subsubsection{Alpha acyclicity} 
A weaker notion of acyclicity is $\alpha$-acyclicity. 
The class of $\alpha$-acyclic hypergraphs is a superset of the class of Berge-acyclic hypergraphs. 

\begin{definition}[Join Tree of a Query]
\label{definition:join-tree}
A join tree of a conjunctive query $Q$ with hypergraph $\mathcal{H} = (\mathcal{V}, \mathcal{E})$ is a tuple $(\mathcal{T}$, $\chi)$ where $\mathcal{T}$ is a tree and $\chi$ is a bijection of the form $\chi:V(\mathcal{T})\rightarrow\mathcal{E}$ where 
for every vertex $v \in \mathcal{V}$, the set $\{t \mid v \in \chi(t)\}$ is a non-empty connected subtree of $\mathcal{T}$ (connectivity).
\end{definition}

There exist several characterisations of alpha acyclicity~\cite{AbiteboulHV95, csur/Brault-Baron16, jacm/Fagin83}: 
\begin{itemize}
\item A conjunctive query $Q$ is $\alpha$-acyclic iff $Q$ has a join tree (see Definition~\ref{definition:join-tree}), and 
\item A conjunctive query $Q$ is $\alpha$-acyclic iff its hypergraph is GYO reducible to the empty hypergraph. That is, by
repeated application of one of the following two rules:
\begin{enumerate}
	\item if a vertex $v$ occurs in only one edge $e$, then remove $v$ from $e$; and
	\item if two distinct edges $e$ and $f$ satisfy $e \subseteq f$, then remove $e$,
\end{enumerate}
the set of hyperedges of the hypergraph is reduced to the empty set.
\end{itemize}

\begin{definition}[Induced Set~\cite{csur/Brault-Baron16}]
\label{definition:induced-set}
Let $\mathcal{E}$ be a family of sets. The set $\mathcal{E}[S] = \{e \cap S \mid e \in \mathcal{E}\} \setminus \{\emptyset\}$ is the induced set of $\mathcal{E}$ on a set $S \subseteq \bigcup \mathcal{E}$.
\end{definition}
    
\begin{definition}[Minimisation of a Familiy of Sets~\cite{csur/Brault-Baron16}]
\label{definition:minimisation}
Let $\mathcal{E}$ be a family of sets. The set $\mathcal{M}(\mathcal{E}) = \{e \in \mathcal{E} \mid \nexists f \in \mathcal{E}, e \subset f\}$ is the \textit{minimization} of $\mathcal{E}$. That is, the subset of hyperedges that are maximal with respect to the inclusion order. Therefore, we trivially have that $\mathcal{M(E)} \subseteq \mathcal{E}$.
\end{definition}
    
\begin{definition}[Conformal Hypergraph~\cite{csur/Brault-Baron16}]
\label{definition:hypergraph-properties}
A hypergraph $\mathcal{H} = (\mathcal{V}, \mathcal{E})$ is conformal if there is no $S \subseteq \mathcal{V}$, with cardinality $\geq 3$ such that
$\mathcal{M}(\mathcal{E}[S]) = \{S\setminus \{x\} \mid x \in S\}$.
\end{definition}
    
\begin{definition}[Cycle-Free Hypergraph~\cite{csur/Brault-Baron16}]
A hypergraph $\mathcal{H} = (\mathcal{V}, \mathcal{E})$ is cycle-free if there is no tuple $(v_{1}, \dots, v_{n})$ with $n \geq 3$ of pairwise distinct vertices such that
$
\mathcal{M}(\mathcal{E}[\{v^{i} \mid 1 \leq i \leq n\}]) = \{\{v^{i}, v^{i+1}\} \mid 1 \leq i < n\} \cup \{\{v^{n}, v^{1}\}\}.
$ 
\end{definition}
    
\begin{definition}[Alpha Acyclic Hypergraph~\cite{csur/Brault-Baron16}]
\label{definition:alpha-conformal-cycle-free}
A hypergraph $\mathcal{H}$ is $\alpha$-acyclic iff it is conformal and cycle-free.
\end{definition}

Boolean conjunctive queries that are $\alpha$-acyclic can be evaluated in time linear in the size of the input database. 
Moreover, a full conjunctive query on this schema can be evaluated in time linear in the size of the input plus the size of the output. Both these results are achieved using Yannakakis's algorithm~\cite{vldb/Yannakakis81}. 
If we treat the size of the query as part of the problem input, then the time complexity of Yannakakis' algorithm becomes polynomial in the size of the query, the input, and the output respectively~\cite{AbiteboulHV95, vldb/Yannakakis81}.

\subsubsection{Further acyclicity notions} 
There are notions of acyclicity stricter than $\alpha$-acyclicity and weaker than Berge acyclicity: beta acyclicity and gamma acyclicity. 
That is, the class of $\alpha$-acyclic hypergraphs is a superset of the class of $\beta$-acyclic hypergraphs, 
which in turn is a superset of the class of $\gamma$-acyclic hypergraphs, 
which in turn is a superset of the class of Berge-acyclic hypergraphs~\cite{jacm/Fagin83,csur/Brault-Baron16}.

\begin{definition}[Gamma Acyclic Hypergraph~\cite{csur/Brault-Baron16}]
    \label{definition:gamma-acyclicity}
 A hypergraph $\calH = (\calV, \calE)$ is $\gamma$-acyclic if $\calH$ is cycle-free and we cannot 
 find $x, y, z \in \calV$ such that $\{\{x, y\},\{x, z\},\{x, y, z\}\} \subseteq \calE[\{x, y, z\}]$.
\end{definition}

\subsection{Width Measures}
\label{subsection:widths}

\begin{definition}[Fractional Edge Cover Number]
\label{definition:rho-star}
Let $\calH = (\calV, \calE)$ be a hypergraph. The {\em fractional edge covers} of $S \subseteq \calV$ are precisely the feasible solutions $(x_e): e \in \calE$ for the following linear program:
\[
\begin{tabular}{crll}
$L_{Q}$: & $\text{minimise}$ & $\sum_{e \in \calE} x_{e}$ & \\ 
& $\text{subsect to}$ & $\sum_{e: v \in e} x_e \geq 1$ & $\text{for all } v \in S$\\
& & $x_{e} \geq 0$ & $\text{for all } e \in \calE$.
\end{tabular}
\]
and the {\em fractional edge cover number} $\rho_{\calE}^{\ast}(S)$ is the cost of the optimal solution. The minimum exists and it is rational. 
\end{definition}

Let $Q$ be a full conjunctive query with equality joins whose hypergraph is $\calH$. 
The fractional edge cover number of $\calH$ provides a tight bound to the worst-case answer size of $Q$~\cite{siamcomp/AtseriasGM13, talg/GroheM14} for any database.
This means that for any database $D$, the size of $Q(\D)$ is $O(|\D|^{{\rho}_{\calE}^{\ast}(\calV)})$. Moreover, there exist arbitrarily large database instances $\D$ for which the size of $Q(\D)$ is at least $\Omega(|\D|^{{\rho}_{\calE}^{\ast}(\calV)})$. 
There are query evaluation algorithms matching this bound up to a log factor~\cite{jacm/NgoPRR18, Veldhuizen14}. 

\begin{definition}[Hypertree Decomposition]
The {\em (hyper)tree decomposition} of a hypergraph $\calH$ is a pair $(\calT, \chi)$, where $\calT$ is a tree whose vertices are $V(\calT) $and $\chi: V(\calT) \rightarrow 2^{\calV}$ maps each node $t$ of the tree $\calT$ to a subset $\chi(t)$ of vertices such that the following properties hold:
\begin{enumerate}
	\item every hyperedge $e \in \calE$ is a subset of a set $\chi(t)$ for some $t \in V(\calT)$, and
	\item for every vertex $v \in \calV$, the set $\{t \mid v \in \chi(t)\}$ is a non-empty connected subtree of $\calT$.
The sets $\chi(t)$ are called the {\em bags} of the tree decomposition.
\end{enumerate}
We use $\td(\calH)$ to denote the set of tree decompositions of a given hypergraph $\calH$.
\end{definition}

\begin{definition}[Polymatroid~\cite{pods/Khamis0S17}]
Consider the vertex set $\calV$. A function $f: 2^{\calV} \rightarrow \R^+$ is a (non-negative) {\em set function} on $\calV$. A set function $f$ on $\calV$ is: 
\begin{itemize}
	\item {\em modular} if $f(S) = \sum_{v \in S} f(\{v\})$ for all $S \subseteq \calV$;
	\item {\em monotone} if $f(X) \leq f(Y)$ whenever $X \subseteq Y$; and 
	\item {\em submodular} if $f(X \cup Y) + f(X \cap Y) \leq f(X) + f(Y)$ for all $X, Y \subseteq \calV$.	
\end{itemize}

A monotone, submodular set function $h : 2^{\calV} \rightarrow \R^+$ with $h(\emptyset) = 0$ is a {\em polymatroid.}
We use $\Gamma_\calV$ to denote the set of all polymatroids $f: 2^{\calV} \rightarrow \R^+$ over the set $\calV$.
\end{definition}

\begin{definition}[Edge Dominated Set Functions~\cite{jacm/Marx13, pods/Khamis0S17}]
Let $\calH = (\calV, \calE)$ be a hypergraph. The set of edge dominated set functions is defined as follows:
\[
\ed(\calH) := \{h \mid h: 2^{\calV} \rightarrow \R^+, h(S) \leq 1, \forall S \in \calE\}. 
\]
\end{definition}

\subsubsection{Fractional Hypertree Width} 
\label{subsubsection:fhtw}

\begin{definition}[Fractional Hypertree Width]
\label{definition:fhtw}
Consider a hypergraph $\calH = (\calV, \calE)$. Recall that $\td(\calH)$ denote the set of all tree decompositions of $\calH$. The fractional hypertree width of $\calH$ is defined by:
\begin{equation}
\fhtw(\calH) := \min_{(\calT, \chi) \in \td(\calH)} \max_{t \in V(\calT)} \rho^{\ast}_{\calE}(\chi(t)).
\label{eq:fhtw}
\end{equation}
\end{definition}

The following is an alternative characterization for $\fhtw(\calH)$~\cite{pods/Khamis0S17}:
(Recall that $\Gamma_{\calV}$ denotes the set of polymatroids over $\calV$ and $\ed(\calH)$
denotes the set of edge dominated set functions over $\calH$.)
\begin{equation}
\fhtw(\calH) :=
\min_{(\calT, \chi) \in \td(\calH)}
\max_{h \in \ed(\calH) \cap \Gamma_{\calV}}
\max_{t \in V(\calT)}
h(\chi(t)).
\label{eq:fhtw:alt}
\end{equation}
The equivalence of the two characterizations is shown~\cite{pods/Khamis0S17} by proving that
for a fixed tree decomposition $(\calT, \chi)$ and a fixed node $t \in V(\calT)$, the following holds:
\begin{equation} \label{eq:max-rho}
\rho^{\ast}_{\calE}(\chi(t)) =
\max_{h \in \ed(\calH) \cap \Gamma_{\calV}} 
h(\chi(t)).
\end{equation}


It is clear from the definition that $\fhtw(\calH) \leq {\rho}_{\calE}^{\ast}(\calV)$. A Boolean conjunctive query $Q$ with equality joins whose hypergraph is $\calH$ can be computed using its tree decomposition as follows:
\begin{enumerate}
	\item  materialize each bag of the tree decomposition by computing the full conjunctive query associated to it using the Leapfrog Triejoin algorithm, and 
	\item run Yannakakis' algorithm~\cite{vldb/Yannakakis81} on the $\alpha$-acyclic Boolean conjunctive query (i.e., the tree decomposition) which arises from the materialization of the bags. 
\end{enumerate} 
For a database $\D$ of size $N$, the first step takes $O(N^{\fhtw(\calH)}\log N)$ time, and the second step takes time linear in the size of the materialization of the bags. The second step takes time linear in the size of the bags. 
By using this algorithm, $Q(D)$ can be thus computed in time $O(N^{\fhtw(\calH)}\log N)$.

\subsubsection{Submodular Width}
\label{subsubsection:subw}

\begin{definition}[Submodular Width~\cite{jacm/Marx13}]
\label{definition:subw}
Given a hypergraph $\calH = (\calV, \calE)$, the submodular width of $\calH$ is defined by:
\begin{equation}
\subw(\calH) :=
\max_{h \in \ed(\calH) \cap \Gamma_{\calV}}
\min_{(\calT, \chi) \in \td(\calH)}
\max_{t \in V(\calT)} h(\chi(t)),
\label{eq:subw}
\end{equation}
where $\ed(\calH)$ denotes the set of edge dominated set functions over $\calH$, and
$\Gamma_\calV$ denotes the set of polymatroids over $\calV$.
\end{definition}

By comparing~\eqref{eq:subw} to~\eqref{eq:fhtw:alt} and using the minimax inequality, it is easy
to see that $\subw(\calH) \leq \fhtw(\calH)$ for any hypergraph $\calH$~\cite{jacm/Marx13,pods/Khamis0S17}.
Moreover, there are classes of queries with bounded submodular width and unbounded fractional hypertree width~\cite{jacm/Marx13}.

Marx showed that a class $\calC$ of Boolean conjunctive queries with equality joins is fixed-parameter tractable $\fpt$
(with the parameter being the query size) if and only if $\calC$ has a bounded submodular width~\cite{jacm/Marx13}. His result suggests the use of submodular width as a yardstick for optimality of algorithms
solving Boolean conjunctive queries with equality joins. Marx gave an algorithm that can solve a query $Q$ in time
$O(\poly(N^{\subw(\calH)}))$ where $N$ is the input database size and
$\calH$ is the hypergraph of $Q$.
His algorithm decomposes the given input database into a union of ``uniform'' databases
and then uses a different tree decomposition to solve the original query over each database separately.
\citet{pods/Khamis0S17} gave another algorithm, called $\panda$, that can answer such a query
in time $O(N^{\subw(\calH)}\cdot \polylog(N))$.
The $\panda$ algorithm works by writing a sequential proof for the upper bound on $\subw(\calH)$
and then interpreting each proof step as an algorithmic operation. It has recently been extended
to handle count queries as well as queries with inequalities~\cite{FAQAI:TODS:2020}.


\section{Missing Details from Section~\ref{section:preliminaries}}
\label{appendix:preliminaries}

\subsection{Segment Tree}
\label{appendix:subsection:segment-tree}

\tikzset{every node/.style={rectangle, align=center, font=, inner sep=3pt}}
\definecolor{green-2}{RGB}{130,179,102} %
\definecolor{red-2}{RGB}{184,84,80} %

\begin{figure}
\centering
    
\begin{tikzpicture}[->,
 level 1/.style={sibling distance=50mm},
 level 2/.style={sibling distance=23mm},
 level 3/.style={sibling distance=14mm},
  level distance=18mm
]
\node [] {$\varepsilon$\\$(-\infty, \infty)$}
  	child {node [] {$0$\\$(-\infty, 3]$}   
    	child {node [] {$00$\\$(-\infty, 1]$}
      		child {node [] {$000$\\$(-\infty, 1)$}
    	}
      		child {node [, label=above:{\color{red-2} $\square$ \color{black}}] {$001$\\$[1, 1]$}
    	}
    	}
    	child {node [, label=above:{\color{red-2} $\square$ \color{black}}] {$01$\\$(1, 3]$}
      		child {node [] {$010$\\$(1, 3)$}
    	}
      		child {node [, label=above:{\color{green-2} $\bullet$ \color{black}}] {$011$\\$[3, 3]$}
    	}
    	}
    }
  	child {node [] {$1$\\$(3, \infty)$}
    	child {node [, label=above:{\color{red-2} $\square$ \color{green-2} $\bullet$ \color{black}}] {$10$\\$(3, 4]$}
      		child {node [] {$100$\\$(3, 4)$}
    	}
      		child {node [] {$101$\\$[4, 4]$}
    	}
    	}
  		child {node [] {$11$\\$(4, \infty)$}
  		}
	};
	
    \draw[mark=*, red-2, -] (-3,-6.5)-- (2.4,-6.5);
     \foreach \x in {-3,2.4} {
         \draw[red-2, -] (\x,-6.25)-- (\x,-6.75);
    }
    
    \draw[green-2, -] (-0.7,-7.25)-- (2.4,-7.25);
     \foreach \x in {-0.7,2.4} {
         \draw[green-2, -] (\x,-7)-- (\x,-7.5);
    }
\end{tikzpicture}
    \caption{Segment tree on the set of intervals $\mathcal{I} = \{ \color{red-2} \square \color{black} = [1, 4], \color{green-2} \bullet \color{black} = [3, 4]\}$. 
	The interval $[1, 4]$ is contained in the canonical subsets of the nodes $001$, $01$, and $10$. 
	The interval $[3, 4]$ is contained in the canonical subsets of the nodes $011$ and $10$.}
    \label{fig:segment-tree}
	\Description[]{}
\end{figure}

\begin{algorithm}
\caption{Segment Tree Insertion Algorithm\\
\textbf{Input}: A node $v \in V(\mathfrak{T}_{\calI})$ and an interval $i \in \calI$.}
\label{alg:segment-tree-insert}
\begin{algorithmic}[1]
\Procedure{Insert}{$v$, $i$}
    \If{$\text{seg}(v) \subseteq i$}
        \State Insert $i$ into the canonical subset $\calI_v$.
    \Else
        \State \textbf{if} $\text{seg}(\text{left-child}(v)) \cap i \neq \emptyset$ \textbf{then} \Call{Insert}{$\text{left-child}(v)$, $i$}
        \State \textbf{if} $\text{seg}(\text{right-child}(v)) \cap i \neq \emptyset$ \textbf{then} \Call{Insert}{$\text{right-child}(v)$, $i$}
    \EndIf
\EndProcedure
\end{algorithmic}
\end{algorithm}

\begin{algorithm}
	\caption{Segment Tree Query Algorithm\\
	A node $v \in V(\mathfrak{T}_{\calI})$ and a query point $p$}
	\label{alg:segment-tree-query}
	\begin{algorithmic}[1]
	\Procedure{Query}{$v$, $p$}
		\State Report all the intervals in $\calI_v$.
		\If{$v$ is not a leaf}
			\State \textbf{if} $p \in~\text{seg}(\text{left-child}(v))$ \textbf{then} \Call{Query}{$\text{left-child}(v)$, $p$}
			\State \textbf{else} \Call{Query}{$\text{right-child}(v)$, $p$}
		\EndIf
	\EndProcedure
	\end{algorithmic}
	\end{algorithm}

Our formal definition of a segment tree $\mathfrak{T}_{\mathcal{I}}$ on an input set of intervals $\calI$ is given in Section~\ref{section:preliminaries}.
Figure~\ref{fig:segment-tree} shows an example segment tree. 
To construct the segment tree, first, we sort the endpoints of the intervals in $\mathcal{I}$ in time $O(\lvert \calI \rvert \cdot \log \lvert \calI \rvert)$ 
to obtain the elementary segments, and then construct a balanced binary tree such that 
the elementary segments sorted from left to right correspond to the leaves of the tree from the left to right respectively.
Then, we compute the corresponding segments of the nodes in a bottom-up fashion in time $O(N)$. 
To compute the canonical subset $\calI_v$ for each node $v$, we use procedure $\textsc{Insert}$ of Algorithm~\ref{alg:segment-tree-insert}, 
called with $v = \text{root}(\mathfrak{T}_{\mathcal{I}})$.
This procedure inserts each interval $i \in \mathcal{I}$ into the canonical subsets of its corresponding maximal segment tree nodes.
For any interval $i \in \mathcal{I}$, the recursive procedure $\textsc{Insert}(\text{root}(\mathfrak{T}_{\calI}, i))$
visits at most $4$ nodes per level of the tree~\cite{preparata2012computational}. 
Hence, the time complexity to insert a single interval is $O(\log \vert \calI \rvert)$.
Therefore, the total time to construct the segment tree is $O(\lvert \calI \rvert \cdot \log \lvert \calI \rvert)$.

Given a query point $p$ and a segment tree $\mathfrak{T}_{\calI}$ the procedure $\textsc{Query}$ of Algorithm~\ref{alg:segment-tree-query}
reports all the intervals in $\calI$ that contain this query point $p$. 
For any query point $p$, the procedure $\textsc{Query}(\text{root}(\mathfrak{T}_{\calI}), p)$ is called at the root of the segment tree
visits one node per level of the tree, so $O(\log \lvert \calI \rvert)$ nodes in total. 
Therefore, the total time complexity of the query is $O(\log \lvert \calI \rvert + k)$, where $k$ is the number of reported intervals~\cite{preparata2012computational}.

\begin{remark}\label{remark:segtree-modification}
We assume wlog that all input intervals are closed intervals. Since there are finitely many input intervals, there exists a sufficiently small\footnote{$\epsilon $ is less than the distance between any two distinct endpoints of all the intervals.} $\epsilon > 0$ such that any open interval $(x, y)$ can be replaced with the closed interval $[x + \epsilon, y - \epsilon]$.	
\end{remark}

\subsection{Proof of Property~\ref{property:segment-tree}}
{
       \textsc{Property~\ref{property:segment-tree}} (Segment Tree)
	   Let $\calI$ be any set of intervals and $\mathfrak{T}_{\calI}$ be the segment tree for it.
	   \begin{enumerate}
		\item Let $u$ and $v$ be nodes in the segment tree. Then $u \in \text{anc}(v)$ if and only if $\text{seg}(u) \supseteq \text{seg}(v)$. Equivalently, $u$ is a prefix of $v$.

		\item For any interval $x\in \calI$, there cannot be two nodes in $\text{CP}_\mathcal{I}(x)$ such that one of them is an ancestor of the other.

		\item For any interval $x\in \calI$, $\text{CP}_{\calI}(x)$ has size and can be computed in time $O(\log \rvert \calI \lvert)$.
	   \end{enumerate}
}
\begin{proof}
	We prove the three statements separately.
	\begin{enumerate}
		\item This statement holds by the construction of the segment tree.
		\item Assume that there exist two nodes $u, v \in \text{CP}_{I}(x)$ such that $u$ is ancestor of $v$.
		From Property~\ref{property:segment-tree}(1) we have $\text{seg}(u) \supseteq \text{seg}(v)$. This is a contradiction
		since the set of segments $\{\text{seg}(v) \mid v \in \text{CP}_{\mathcal{I}}(x)\}$ forms a partition of $x$.
		\item We claim that given an interval $x \in \calI$, there are no three nodes in $\text{CP}_{\calI}(x)$ 
		that are at the same depth of the tree. Therefore, the size of $\text{CP}_{\calI}(x)$ has size 
		at most $O(\log \lvert \calI \rvert)$.  To see why this is true, let
		$v_{1}, v_{2}, v_{3}$ be three nodes at the same depth, numbered from left to right. Suppose
		$v_1$, $v_3 \in \text{CP}_{\calI}(x)$. This means that $x$ spans the whole interval
		from the left endpoint of $\text{seg}(v_{1})$ to the right endpoint of $\text{seg}(v_{3})$. 
		Because $v_{2}$ lies between $v_{1}$ and $v_{3}$, $\text{seg}(\text{parent}(v_{2}))$ must be contained in 
		$x$. Hence, $v_{2} \notin \text{CP}_{\calI}(x)$. 

		The canonical partition of $x$ can be computed using the procedure $\textsc{Insert}(\text{root}(\mathfrak{T}_{\calI}), x)$
		in Algorithm~\ref{alg:segment-tree-insert}. The recursive procedure visits at most $4$ nodes per level of the tree~\cite{preparata2012computational}. 
		Hence, the time complexity to insert a single interval is $O(\log \vert \calI \rvert)$.
	\end{enumerate}
\end{proof}

\section{Missing Details from Section~\ref{section:ij-to-ej}}
\label{appendex:ij-to-ej}

\subsection{Proof of Lemma~\ref{lemma:intersection-predicate-not-ordered}}
\label{proof:intersection-predicate-not-ordered}
{
	\textsc{Lemma 
	\ref{lemma:intersection-predicate-not-ordered}}
	For any set of intervals $S = \{x_{1}, \dots, x_{k}\} \subseteq \mathcal{I}$, the following equivalence holds
	\begin{align*}
		\left(\bigcap_{i \in[k]} x_{i} \right)\neq \emptyset \equiv 
		\bigvee_{i \in[k]} 
		\left[
			 \bigvee_{(v_1, \dots, v_k) \in \text{anc}(\text{leaf}(x_i))^k} 
			 \left( 
				 \bigwedge_{\substack{j \in [k] \\ j \neq i}} v_j \in \text{CP}_\mathcal{I}(x_j) 
			 \right) 
		\right]
	\end{align*}
}
	
\begin{proof}
	The intersection of the intervals in $S$ is equal to the interval $[l, r]$ if $l \leq r$, and equal to $\emptyset$ otherwise, 
	where $l := \max_{1 \leq i \leq k} \, x_i.l$ and $r := \min_{1 \leq i \leq k} \, x_i.r$. 
	Therefore, the intervals in $S$ intersect if and only if there is an interval $x_i \in S$ (i.e., $i$ is equal to $\text{argmax}_{1 \leq i \leq k} \, x_i.l$) 
	such that the point $x_i.l$ is contained in all the other intervals in $S$. 
	Hence, we have:
	\begin{equation}
		\left(
			\bigcap_{i \in [k]} x_i 
		\right)
		\neq \emptyset \equiv \bigvee_{i \in [k]} 
		\left(
			\bigwedge_{\substack{j \in [k] \\ j \neq i}} x_i.l \in x_j 
		\right)
	\end{equation}
	
	Since the set of segments $\{\text{seg}(v) \mid v \in \text{CP}_{\mathcal{I}}(x_j)\}$ forms a partition of the interval $x_j$, we have:
	\begin{equation}
		\left(
		\bigcap_{i \in [k]} x_i 
		\right)
		\neq \emptyset \equiv 
		\bigvee_{i \in [k]} 
		\left[ 
			\bigwedge_{\substack{j \in [k] \\ j \neq i}} 
			\left( 
				\bigvee_{v \in \text{CP}_{\mathcal{I}}(x_j)} x_i.l \in \text{seg}(v) 
			\right) 
		\right]
	\end{equation}
	
	By the construction of the segment tree, a point $p$ is included in the segment $\text{seg}(v)$ if and only if the node $v$ is an ancestor of $\text{leaf}(p)$. 
	Hence, we have:
	\begin{equation}
		\begin{split}
		\left(\bigcap_{i \in [k]} x_i \right) \neq \emptyset 
		&\equiv \bigvee_{i \in [k]} \left[ \bigwedge_{\substack{j \in [k] \\ j \neq i}} \left( \bigvee_{v \in \text{CP}_{\mathcal{I}}(x_j)} v \in \text{anc}(\text{leaf}(x_i)) \right) \right] \\
		&\equiv \bigvee_{i \in [k]} \left[ \bigwedge_{\substack{j \in [k] \\ j \neq i}} \left( \bigvee_{v_j \in \text{anc}(\text{leaf}(x_i))} v_j \in \text{CP}_{\mathcal{I}}(x_j) \right) \right] \\
		&\equiv \bigvee_{i \in [k]} \left[ \bigvee_{(v_1, \dots, v_k) \in (\text{anc}(\text{leaf}(x_i)))^{k}} \left( \bigwedge_{\substack{j \in [k] \\ j \neq i}} v_j \in \text{CP}_\mathcal{I}(x_j) \right) \right],
		\end{split}
	\end{equation}
	where the second equivalence is due to the fact that $\bigvee_{a \in A} a \in B \equiv \bigvee_{a \in B} a \in A \equiv A \cap B \neq \emptyset$, for any two sets $A$ and $B$, and the third equivalence is due to the distributivity of the conjunction over the disjunction. 
	\nop{
		Note that the tuple $(v_1, \dots, v_k)$ in Equation~\eqref{eq:intersection-predicate-not-ordered} does not include the variable $v_i$.
	}
\end{proof}

\subsection{Proof of Property~\ref{property:unique-tuple-of-nodes}}
\label{appendix:property:unique-tuple-of-nodes}

{
	\textsc{Property 
	\ref{property:unique-tuple-of-nodes}}
	Consider a set of intervals $S = \{x_{1}, \dots, x_{n}\} \subseteq \calI$ and a segment tree $\mathfrak{T}_{\calI}$. 
    For any $x_i \in S$, there can be at most one tuple of nodes $v_j \in (\text{anc}(\text{leaf}(x_i)))$  for $j\in[k], j\neq i$ 
    that satisfy the conjunction of Lemma~\ref{lemma:intersection-predicate-not-ordered}.
}

\begin{proof}
	Assume that there exist two such tuples: $(v_j)_{j\in[k], j\neq i}$ and $(v_j^\prime)_{j\in[k], j\neq i}$ 
	that satisfy the conjunction of Lemma~\ref{lemma:intersection-predicate-not-ordered}. Since $x_{i}$ is fixed, the nodes from the two tuples correspond to the same 
	root-to-leaf path whose leaf is $\text{leaf}(x_{i})$. Therefore, there exists $j \in [k], j \neq i$ such that $v_{j} \neq v_{j}^{\prime}$ 
	and $v_{j}, v_{j}^{\prime} \in \text{CP}_{\calI}(x_{j})$. 
	This is not possible due to Property~\ref{property:segment-tree}(2), 
	as there cannot be distinct nodes in $\text{CP}_{\calI}(x_{j})$ that belong to the same root-to-leaf path .
\end{proof}

\subsection{Proof of Lemma~\ref{lemma:from-intersections-to-equalities}}

{
	\textsc{Lemma 
	\ref{lemma:from-intersections-to-equalities}}
    Consider a subset of intervals $S = \{x_{1}, \dots, x_{k}\} \subseteq \mathcal{I}$. 
    The predicate $\left(\bigcap_{x \in S} x\right) \neq \emptyset$ is true if and only 
    if there exists a permutation $\sigma \in \pi(\{x_{1}, \dots, x_{k}\})$ and a tuple of bit-strings
    $(y_{1}, \dots, y_{k})$ such that:
    \begin{itemize}
    \item for each $1 \leq i < k$ we have $y_{1} \circ \dots \circ y_{i}  \in \text{CP}_{\calI}(\sigma_{i})$, and
    \item for $i = k$ we have $y_{1} \circ \dots \circ y_{i} = \text{leaf}(\sigma_{i})$.
    \end{itemize}
}

\begin{proof}
Let $\mathfrak{T}_{\calI}$ be a segment tree on $\calI$.

$\Rightarrow$: Assume that the predicate $\left(\bigcap_{x \in S} x\right) \neq \emptyset$ is true.
By Lemma~\ref{lemma:intersection-predicate-ordered} there exists a permutation $\sigma \in \pi(\{x_{1}, \dots, x_{k}\})$ and a tuple of nodes
$(u_{1}, \dots, u_{k})$ such that $u_{1} \in \text{anc}(u_{2}),
\dots, u_{k-1} \in \text{anc}(u_{k}), u_{k} = \text{leaf}(\sigma_{k})$ and $u_{i} \in \text{CP}_{\calI}(\sigma_{i})$ for each $1 \leq i < k$.
By Property~\ref{property:segment-tree} (1), 
we have that $u_{1}$ is a prefix of $u_{2}$ is a prefix of $u_{3}$ and so on.
Hence, there exists a tuple of bit-strings $(y_{1}, \dots, y_{k})$ such that $u_{i} = y_{1} \circ \dots \circ y_{i}$ for each $1 \leq i \leq k$.
Since, $u_{i} \in \text{CP}_{\calI}(\sigma_{i})$ for each $1 \leq i < k$ and $u_{k} = \text{leaf}(\sigma_{k})$
we have $y_{1} \circ \dots \circ y_{i} \in \text{CP}_{\calI}(\sigma_{i})$ for each $1 \leq i < k$ 
and $y_{1} \circ \dots \circ y_{k} = \text{leaf}(\sigma_{k})$.
Hence, the statement of Lemma~\ref{lemma:from-intersections-to-equalities} is true.

$\Leftarrow$:
Assume that the statement of Lemma~\ref{lemma:from-intersections-to-equalities} is true. 
That means that there exists a permutation $\sigma \in \pi(\{x_{1}, \dots, x_{k}\})$ and a tuple 
of bit-strings $(y_{1}, \dots, y_{k})$ such that 
for each $1 \leq i < k$ we have $y_{1} \circ \dots \circ y_{i}  \in \text{CP}_{\calI}(\sigma_{i})$
and $y_{1} \circ \dots \circ y_{k} = \text{leaf}(\sigma_{k })$. Let 
$u_{i} = y_{1} \circ \dots, y_{i}$ for each $1 \leq i \leq k$. We have that $u_{i} \in \text{CP}_{\calI}(\sigma_{i})$ for each $1 \leq i < k$
and $u_{k} = \text{leaf}(\sigma_{k})$. Furthermore, $u_{1}$ is prefix of $u_{2}$,
$u_{2}$ is prefix of $u_{3}$ and so on. Hence, by Property~\ref{property:segment-tree} (1), we have $u_{1} \in \text{anc}(u_{2})$, $u_{2} \in \text{anc}(u_3)$, 
$\dots$, $u_{k-1} \in \text{anc}(u_{k})$ and so on. Therefore, the predicate of Lemma~\ref{lemma:intersection-predicate-ordered}
is true. So, $\left(\bigcap_{x \in S} x\right) \neq \emptyset$ is true.
\end{proof}

\subsection{Proof of Lemma~\ref{lemma:db-x-sigma-size}}
\label{appendix:lemma:db-x-sigma-size}

{
	\textsc{Lemma 
	\ref{lemma:db-x-sigma-size}}
	Each new relation $\tilde{R}(\tilde{\sigma_i})$ in database $\tilde{\D}_{([X],\sigma)}$ constructed from the database $\D$ following Definition~\ref{definition:one-step-transformation} has the size:
    \begin{itemize}
        \item $O(\lvert R(\sigma_{i}) \rvert \cdot \log^{i} \lvert \mathcal{I}\rvert)$ if $i\in[k - 1]$ and
        \item  $O(\lvert R(\sigma_{i})\rvert \cdot \log^{i - 1} \lvert \mathcal{I} \rvert)$ if $i = k$,
    \end{itemize}
    and can be constructed in time proportional to its size.
}

\begin{proof}
	Given a node $u$ from $V(\mathfrak{T}_{\calI})$ and an integer $i$, let $\mathfrak{F}(u, i)$ denote the set 
	that contains all the tuples $(x_1, \dots, x_{i})$ such that $x_{1} \circ \dots \circ x_{i} = u$.
	
	\begin{claim}
		\label{claim:db-x-sigma-size}
		Let $u \in V(\mathfrak{T}_{\calI})$ and $i$ be an integer.
		The size of $\mathfrak{F}(u, i)$ is $O(\log^{i-1} \lvert \calI\rvert)  = O(\log^{i-1} \lvert \D\rvert)$.
	\end{claim}

	We prove each statement separately:
	
	\begin{itemize}
        \item Let $i$ be an integer such that $1 \leq i < k$. 
		The relation $\tilde{R}(\tilde{\sigma}_{i})$
        can be constructed using the following procedure: 
		for each tuple $t \in R(\sigma_{i})$,
		for each $u \in \text{CP}_{\calI}(t([X]))$, for each $(x_1, \dots, x_i) \in \mathfrak{F}(u, i)$
		construct the tuple $\tilde{t}$ over the schema 
		$\tilde{\sigma}_{i} = (\sigma_{i}\setminus \{[X]\}) \cup \{X_{1}, \dots, X_{i}\}$
		such that $\tilde{t}[\tilde{\sigma}_{i}\setminus \{X_1, \dots, X_{i}\}] = t[\sigma_{i}\setminus \{[X]\}]$
		and $\tilde{t}(X_{j}) = x_{j}$ for each $1 \leq j \leq i$.  
		Then, insert the tuple $\tilde{t}$ into $\tilde{R}(\tilde{\sigma}_{i})$.
		
        By Property~\ref{property:segment-tree} (3) we have that $\lvert \text{CP}_{\calI}(t([X]))\rvert = O(\log \lvert \calI \rvert) = O(\log \lvert \D \rvert)$ 
		and by Claim~\ref{claim:db-x-sigma-size} we have that $\lvert \mathfrak{F}(u, i) \rvert = O(\log^{i-1} \lvert \D \rvert)$. 
		Therefore, $\lvert \tilde{R}(\tilde{\sigma}_{i}) \rvert = 
		O(\lvert R(\sigma_{i})\rvert \cdot \log \lvert \D \rvert \cdot \log^{i-1} \lvert \D \rvert)
		= O(\lvert R(\sigma_{i})\rvert \cdot \log^{i} \lvert \D \rvert)$. 
		Furthermore, its construction time is proportional to its size.
		
		\item Let $i = k$. 
		The relation $\tilde{R}(\tilde{\sigma}_{i})$ is constructed using the following procedure: 
		for each tuple in $t \in R(\sigma_{i})$,
		for each $(x_1, \dots, x_i) \in \mathfrak{F}(\text{leaf}(t([X])), i)$
		construct the tuple $\tilde{t}$ over the schema 
		$\tilde{\sigma}_{i} = (\sigma_{i}\setminus \{[X]\}) \cup \{X_{1}, \dots, X_{i}\}$
		such that $\tilde{t}[\tilde{\sigma}_{i}\setminus \{X_{1}, \dots, X_{i}\}] = t[\sigma_{i}\setminus \{[X]\}]$
		and $\tilde{t}(X_{j}) = x_{j}$ for each $1 \leq j \leq i$.    
		Then, insert the tuple $\tilde{t}$ into $\tilde{R}(\tilde{\sigma}_{i})$.
		
		By Claim~\ref{claim:db-x-sigma-size} we have that $\lvert \mathfrak{F}(\text{leaf}(t([X])), i) \rvert = O(\log^{i-1} \lvert \D \rvert)$. 
        Therefore, $\lvert \tilde{R}(\tilde{\sigma}_{i}) \rvert  
		 = O(\lvert R(\sigma_{i})\rvert  \cdot \log^{i-1} \lvert \D \rvert)$. 
		 Furthermore, its construction time is proportional to its size.
    \end{itemize}

\end{proof}

\subsection{Proof of Lemma~\ref{lemma:correctness-one-step}}
\label{appendix:lemma:correctness-one-step}

{
	\textsc{Lemma 
	\ref{lemma:correctness-one-step}}
	Given any \query query $Q$, interval variable $[X]$ in $Q$, 
	and any database $\D$, let the \query query $\tilde{Q}_{[X]}$ and database $\tilde{\D}_{[X]}$ be constructed as per Definitions~\ref{definition:one-step-transformation} and \ref{definition:one-step-rewriting}. 
	Then, $Q(\D)$ is true if and only if $\tilde{Q}_{[X]}(\tilde{\D}_{[X]})$ is true.
}

\begin{proof}

For simplicity let $k = \lvert \calE_{[X]} \rvert$.

$\Rightarrow$:  
Assume that $Q(\D)$ is true. 
That means that there exists a set of tuples $\{t_e \in R(e) \mid e \in \calE \}$ that satisfy 
\[
	\left(\bigcap_{e \in \calE_{[X]}} t_{e}(X)\right) \neq \emptyset,
\]
and also they satisfy the rest of the join conditions of the query $Q$, 
i.e. the (intersection or equality) joins on
the variables in $\calV \setminus \{[X]\}$.

By Definition~\ref{definition:one-step-transformation}, 
for each $i \in [k]$ there exists a tuple
$\tilde{t}_{\tilde{\sigma}_{i}} \in \tilde{R}(\tilde{\sigma}_{i})$, where $\tilde{R}(\tilde{\sigma}_{i})$ is a relation over schema 
$\tilde{\sigma}_{i} = (\sigma_{i} \setminus \{[X]\})\cup \{X_{1}, \dots, X_{i}\}$ in $\tilde{\D}_{([X], \sigma)}$, 
such that $\tilde{t}_{\tilde{\sigma}_{i}}(\tilde{\sigma}_{i} \setminus \{X_{1}, \dots, X_{i}\}) = t_{\sigma_{i}}(\sigma_{i}\setminus\{[X]\})$ 
and $\tilde{t}_{\tilde{\sigma}_{i}}(X_{1}) = x_{1}, \dots, \tilde{t}_{\tilde{\sigma}_{i}}(X_{i}) = x_{i}$. 
Furthermore, by Definition~\ref{definition:one-step-transformation} for each $e \in \calE \setminus \calE_{[X]}$ we have $R(e) \in \tilde{\D}_{([X], \sigma)}$.
For each $i \in [k]$ we have $\tilde{t}_{\tilde{\sigma}_{i}}(X_{i}) =
\tilde{t}_{\tilde{\sigma}_{i+1}}(X_{i}) =
\dots = \tilde{t}_{\tilde{\sigma}_{k}}(X_{i})$,
i.e. the tuples $\tilde{t}_{\tilde{\sigma}_{1}}, \dots, \tilde{t}_{\tilde{\sigma}_{k}}$ satisfy 
the equi-join conditions for the variables $X_{1}, \dots, X_{k}$ in the query $\tilde{Q}_{([X], \sigma)}$. 
Furthermore, all the tuples $\tilde{t}_{\tilde{\sigma}_{1}}, \dots, \tilde{t}_{\tilde{\sigma}_{k}}$ and $\tilde{t}_{e}$ for each $e \in \calE\setminus \calE_{[X]}$
satisfy the rest of the join conditions of query $\tilde{Q}_{([X], \sigma)}$ i.e. the (intersection or equality) joins on
variables in $\calV \setminus \{[X]\}$. 
Hence, $\tilde{Q}_{([X], \sigma)}(\tilde{\D}_{([X], \sigma)})$ is true. Therefore, $\tilde{Q}_{[X]}(\tilde{\D}_{[X]})$ is true.

$\Leftarrow$: 
Assume that $\tilde{Q}_{[X]}(\tilde{\D}_{[X]})$ is true.  
That means that there exists a permutation $\sigma \in \pi(\calE_{[X]})$ and
there exist tuples $\tilde{t}_{\tilde{\sigma}_{i}} \in \tilde{R}(\tilde{\sigma}_{1}), \dots, \tilde{t}_{\tilde{\sigma}_{k}} \in \tilde{R}(\tilde{\sigma}_{k})$ 
and $\tilde{t}_{e} \in R(e)$ for each $e \in \calE\setminus \calE_{[X]}$
such that for each $i \in [k]$ we have
$\tilde{t}_{\tilde{\sigma}_{i}}(X_{i}) =
\tilde{t}_{\tilde{\sigma}_{i+1}}(X_{i}) =
\dots = \tilde{t}_{\tilde{\sigma}_{k}}(X_{i})$  
and also the tuples satisfy the rest of the join conditions of the query $\tilde{Q}_{[X]}$, 
i.e. the (intersection or equality) joins on variables in $\calV \setminus \{[X]\}$.

By Definition~\ref{definition:one-step-transformation}, for each $i \in [k]$ there exists a tuple $t_{\sigma_{i}} \in R(\sigma_{i})$,
where $R(\sigma_{i})$ is a relation over schema $\sigma_{i} = (\tilde{\sigma_{i}} \cup \{[X]\}) \setminus \{X_{1}, \dots, X_{k}\}$ 
in $\D$, such that $t_{\sigma_{i}}(\sigma_{i} \setminus \{[X]\}) = \tilde{t}_{\tilde{\sigma}_{i}}(\tilde{\sigma}_{i} \setminus \{X_{1}, \dots, X_{i}\})$
and $\tilde{t}_{\tilde{\sigma}_{i}}(X_{1}) \circ \dots \circ \tilde{t}_{\tilde{\sigma}_{i}}(X_{i}) \in \text{CP}_{\calI}(t_{\sigma_{i}}(X))$ 
for each $i \in [k-1]$ and $\tilde{t}_{\tilde{\sigma}_{i}}(X_{1}) \circ \dots \circ \tilde{t}_{\tilde{\sigma}_{i}}(X_{i}) = \text{leaf}(t_{\sigma_{i}}(X))$ for $i=k$. 
Furthermore, by Definition~\ref{definition:one-step-transformation} we have $R(e) \in \D$ for each $e \in \calE$.
By Lemma~\ref{lemma:from-intersections-to-equalities}, the predicate 
\[
	\left(\bigcap_{i \in [k]} t_{\sigma_{i}}(X) \right) \neq \emptyset
\]
is true. Therefore, the tuples $t_{\sigma_{1}}, \dots, t_{\sigma_{k}}$ satisfy the intersection join condition on variable $[X]$ in query $Q$.
Furthermore, all the tuples $t_{\sigma_{1}}, \dots, t_{\sigma_{k}}$ and $t_{e}$ for each $e \in \calE \setminus \calE_{[X]}$ 
satisfy the rest of the join conditions of $Q$, i.e. the (intersection or equality) joins on
variables in $\calV \setminus \{[X]\}$.  
Hence $Q(\D)$ is true.
\end{proof}

\subsection{Proof of Theorem~\ref{theorem:reduction-correctness}}
\label{appendix:theorem:theorem:reduction-correctness}

{
	\textsc{Theorem 
	\ref{theorem:reduction-correctness}}
	For any $\bcqij$ query $Q$ with hypergraph $\calH$ and any database $\D$, 
	it holds that $Q(\D)$ is true if and only if $\bigvee_{\tilde{Q}\in\tilde{\mathbf{Q}}} \tilde{Q}(\tilde{\mathbf{D}})$ is true, 
	where $(\tilde{\mathbf{H}},\tilde{\mathbf{Q}},\tilde{\mathbf{D}}) = \textsc{Reduce}(\{\calH\}, \{Q\},\D)$.
}

\begin{proof}
	Without loss of generality assume that $Q$ includes the interval variables $[X_{1}], \dots, [X_{n}]$ and that the 
	procedure $\textsc{Reduce}(\calH, Q, \D)$ iterates over them in the listed order.  
	We use a proof by induction.
	Let $P(j)$ denote the statement $Q(\D)$ if and only if $\bigvee_{\tilde{Q}\in\tilde{\mathbf{Q}}} \tilde{Q}(\tilde{\mathbf{D}})$ 
	after the $j$-th iteration of the reduction (procedure $\textsc{Reduce}(\{\calH\}, \{Q\}, \D)$, Algorithm~\ref{algorithm:reduction}).

	\textbf{Base case.} We prove that $P(1)$ is true. Note that $P(1)$ is equivalent to the statement $Q(\D)$ 
	if and only if $\bigvee_{\tilde{Q}\in\tilde{\mathbf{Q}}} \tilde{Q}(\tilde{\mathbf{D}})$ 
	after the $1$-st iteration of the reduction. 
	During the $1$-st iteration of the reduction we have $\tilde{\Q}_{0} = \{Q\}, \tilde{\Q} = \emptyset$ and $\tilde{\D} = \D$.
	Hence, after the $1$-st iteration we have 
	\begin{equation}
		\label{equation:reduction-correctness-1}
		\tilde{\Q} 
		= \bigcup_{Q \in \tilde{\Q}_{0}} \left(\bigcup_{\sigma \in \pi(\calE_{[X_{1}]})} \{\tilde{Q}_{([X_1], \sigma)}\} \right)
		= \bigcup_{Q \in \{Q\}} \left(\bigcup_{\sigma \in \pi(\calE_{[X_{1}]})} \{\tilde{Q}_{([X_1], \sigma)}\} \right)
		= \bigcup_{\sigma \in \pi(\calE_{[X_{1}]})} \{\tilde{Q}_{([X_1], \sigma)}\}
	\end{equation}
	where $\tilde{Q}_{([X_1], \sigma)}$ follows Definition~\ref{definition:query-onestepreduction}, and 
	$\tilde{\D} = \tilde{\D}_{[X_{1}]}$ where $\tilde{\D}_{[X_{1}]}$ follows Definition~\ref{definition:one-step-transformation}.
	Hence, by Equation~\eqref{equation:reduction-correctness-1} we have
	\begin{equation}
		\label{equation:reduction-correctness-2}
		\bigvee_{\tilde{Q} \in \tilde{\Q}} \tilde{Q}(\tilde{\D}) 
		\equiv
		\bigvee_{\sigma \in \pi(\calE_{[X_{1}]})} \tilde{Q}_{([X_1], \sigma)}(\tilde{\D}_{[X_1]})
		\equiv
		\tilde{Q}_{[X_{1}]}(\tilde{\D}_{[X_1]}).
	\end{equation}
	The second equivalence is due to Definition~\ref{definition:query-onestepreduction}.
	By Lemma~\ref{lemma:correctness-one-step} we have $\tilde{Q}_{[X_{1}]}(\tilde{\D}_{[X_1]}) \equiv Q(\D)$. 
	Hence, $P(1)$ is true.

	\textbf{Inductive step.} The statement $P(j)$ is equivalent to the statement $Q(\D)$ if and only if $\bigvee_{\tilde{Q}\in\tilde{\mathbf{Q}}} \tilde{Q}(\tilde{\mathbf{D}})$ 
	after the $j$-th iteration of the reduction. 
	We prove that if $P(j)$ is true then $P(j+1)$ is true for any $1 \leq j < n$.
	
	Assume $P(j)$ is true. Let $\tilde{\Q}_{0}^{\prime}$ denote the $\tilde{\Q}_{0}$, $\tilde{\Q}^{\prime}$ denote the $\tilde{\Q}$ 
	and $\tilde{\D}^{\prime}$ denote $\tilde{\D}$ during the $j$-th iteration. After the $j$-th iteration we have 
	\begin{equation}
		\tilde{\Q}^{\prime} = \bigcup_{Q^{\prime} \in \tilde{\Q}_{0}^{\prime}} 
		\left(
			\bigcup_{\sigma^{\prime} \in \pi(\calE_{[X_{j}]})} \{\tilde{Q}^{\prime}_{([X_j], \sigma^{\prime})}\}
		\right)
	\end{equation}
	where $\tilde{Q}^{\prime}_{([X_{j}], \sigma^{\prime})}$ follows Definition~\ref{definition:query-onestepreduction},
	and $\tilde{\D} = \tilde{\D}^{\prime}_{[X_{j}]}$ where $\tilde{\D}^{\prime}_{[X_{j}]}$ follows Definition~\ref{definition:one-step-transformation}.  
	During the $j+1$-th iteration we have $\tilde{\Q}_{0} = \tilde{\Q}^{\prime}$, $\tilde{\Q}^{\prime} = \emptyset$.
	Hence, after the $j+1$-th iteration we have:
	\begin{equation}
		\label{equation:reduction-correctness-3}
		\tilde{\Q} 
		= \bigcup_{Q \in \tilde{\Q}_{0}} \left(\bigcup_{\sigma \in \pi(\calE_{[X_{j+1}]})} \{\tilde{Q}_{([X_{j+1}], \sigma)}\}\right)
		= \bigcup_{Q^{\prime} \in \tilde{\Q}_{0}^{\prime}} 
		\left( 
		\bigcup_{\sigma^{\prime} \in \pi(\calE_{[X_{j}]})}	
		\left( 
			\bigcup_{\sigma \in \pi(\calE_{[X_{j+1}]})} \{\tilde{J}_{([X_{j+1}], \sigma)}\}
		\right)
		\right)
	\end{equation}
	where $\tilde{J} = \tilde{Q}^{\prime}_{([X_j], \sigma^{\prime})}$, and $\tilde{J}_{([X_{j+1}], \sigma)}$ follows Definition~\ref{definition:query-onestepreduction}, and 
	$\tilde{\D} = \tilde{\D}_{[X_{j+1}]}$ where $\tilde{\D}_{[X_{j+1}]}$ follows Definition~\ref{definition:one-step-transformation}.
	Hence, by Equation~\eqref{equation:reduction-correctness-3} we have:
	\begin{equation}
		\begin{split}
		\bigvee_{\tilde{Q} \in \tilde{\Q}} \tilde{Q}(\D)
		& \equiv
		\bigvee_{Q^{\prime} \in \tilde{\Q}_{0}^{\prime}} 
		\left( 
		\bigvee_{\sigma^{\prime} \in \pi(\calE_{[X_{j}]})}	
		\left( 
			\bigvee_{\sigma \in \pi(\calE_{[X_{j+1}]})} \tilde{J}_{([X_{j+1}], \sigma)}(\tilde{\D}_{[X_{j+1}]})
		\right)
		\right)
		\\ 
		& \equiv 
		\bigvee_{Q^{\prime} \in \tilde{\Q}_{0}^{\prime}} 
		\left( 
		\bigvee_{\sigma^{\prime} \in \pi(\calE_{[X_{j}]})}	
			\tilde{J}_{[X_{j+1}]}(\tilde{\D}_{[X_{j+1}]})
		\right).
	\end{split}
	\end{equation}
	The second equivalence is obtained by Definition~\ref{definition:query-onestepreduction}.
	Since $\tilde{J} = \tilde{Q}^{\prime}_{([X_j], \sigma^{\prime})}$, by Lemma~\ref{lemma:correctness-one-step}, we have that 
	$\tilde{J}_{[X_{j+1}]}(\tilde{\D}_{[X_{j+1}]}) \equiv  \tilde{Q}^{\prime}_{([X_{j}], \sigma^{\prime})}(\tilde{\D}_{[X_{j}]})$.
	Hence, we have:

	\begin{equation}
		\begin{split}
		\bigvee_{\tilde{Q} \in \tilde{\Q}} \tilde{Q}(\D)
		\equiv
		\bigvee_{Q^{\prime} \in \tilde{\Q}_{0}^{\prime}} 
		\left( 
		\bigvee_{\sigma^{\prime} \in \pi(\calE_{[X_{j}]})}	
			\tilde{Q}^{\prime}_{([X_{j}], \sigma^{\prime})}(\tilde{\D}_{[X_{j}]})
		\right)
		\equiv 
		\tilde{\Q}^{\prime}(\tilde{\D}^{\prime})
		\equiv
		Q(\D)
	\end{split}
	\end{equation}
	The third equivalence is due to the assumption that $P(j)$ is true.
	Therefore, the statement $P(j+1)$ also holds true.

	\textbf{Conclusion.}
	Since both the base case and the inductive step have been proved as true, 
	by induction the statement $P(n)$ is true. 
\end{proof}

\subsection{Proof of Theorem~\ref{theorem:complexity}}

{
	\textsc{Theorem 
	\ref{theorem:complexity}}
	Given any \bcqij query $Q$ with hypergraph $\calH$ and database $\D$, 
	$Q(\D)$ can be computed in time $O(|\D|^{\ijw(\calH)}\cdot\polylog |\D|)$.
}

\begin{proof}
	By Theorem~\ref{theorem:reduction-correctness} we have $Q(\D)$ 
	if and only if 
	\begin{equation}
		\label{equation:theorem-complexity-1}
		\bigvee_{\tilde{Q}\in\tilde{\mathbf{Q}}} \tilde{Q}(\tilde{\D}), 
	\end{equation}
	where $(\tilde{\mathbf{H}},\tilde{\mathbf{Q}},\tilde{\mathbf{D}}) = \textsc{Reduce}(\{\calH\}, \{Q\},\D)$.
	Therefore, the upper bound for the computation of $Q(\D)$ is given by the upper bound of 
	the query with the maximum upper bound among the queries in the disjunction of Equation~\eqref{equation:theorem-complexity-1}. 
	The query with the maximum upper bound among the queries in the disjunction of Equation~\eqref{equation:theorem-complexity-1}, is the one
	whose hypergraph has the maximum submodular width~\cite{pods/Khamis0S17}. Hence, the time complexity of $Q(\D)$
	is upper bounded by 
	\[
		Q(\lvert \D \rvert^{\max_{\tilde{\mathcal{H}} \in \tau(\mathcal{H})} 
		\subw(\tilde{\mathcal{H}})} \cdot \polylog \lvert \D \rvert),
	\]
	given that $\tau(\calH)$ is the set of hypergraphs that correspond to the queries in the disjunction 
	of Equation~\eqref{equation:theorem-complexity-1}. 
	By Definition~\ref{definition:ij-width} we have:
	\[
		O(\lvert \D \rvert^{\max_{\tilde{\mathcal{H}} \in \tau(\mathcal{H})} 
		\subw(\tilde{\mathcal{H}})} \cdot \polylog \lvert \D \rvert) = 
		O(\lvert \D \rvert^{\ijw(\mathcal{H}) } \cdot \polylog \lvert \D \rvert)
	\]
	Hence, $Q(\D)$ can be computed in time $O(\lvert \D \rvert^{\ijw(\mathcal{H}) } \cdot \polylog \lvert \D \rvert)$.
\end{proof}

\section{Missing Details from Section~\ref{section:reverse-reduction}}
\label{appendix:reverse:reduction}

\nop{
So far, we have seen how to create an algorithm for solving any $\bcqij$
query $Q$ with hypergraph $\calH = (\calV, \calE)$ over any database  $\D$. This algorithm is based on our reduction and the best known algorithms for $\bcq$ queries, such as those from $\tilde{\mathbf{Q}}$, where
$(\tilde{\mathbf{H}},\tilde{\mathbf{Q}},\tilde{\mathbf{D}}) = \textsf{REDUCE}(\{\calH\}, \{Q\},\D)$.
A natural question is how good is this algorithm? We show next, by using a backward reduction, that our
algorithm is the best possible algorithm, modulo a poly-logarithmic factor, given optimal solutions to the \bcq queries $\tilde Q\in\tilde \Q$. That is, we show that any lower bound time complexity for the
$\bcq$ queries from $\tilde{\mathbf{Q}}$ is also a lower bound time
complexity for $Q$. And since our reduction solves $Q$ in a time
complexity equal to this lower bound time
complexity~\footnote{Subject to existence of the best possible
algorithms that solve the queries from $\tilde{\mathbf{Q}}$ in the
corresponding lower bound time complexities.}, our algorithm based on
our reduction is the best possible algorithm in terms of time
complexity, modulo a poly-logarithmic factor.\\
}

{
	\textsc{Theorem~\ref{theorem:lower-bound}}
	Let $Q$ be any self-join-free $\bcqij$ query with hypergraph $\calH$. Let $\tilde{Q}$ be any $\bcq$ query whose hypergraph is in $\tau(\calH)$. For any database $\tilde\D$, let $\Omega(T(|\tilde\D|))$ be a lower bound on the time complexity for computing $\tilde{Q}$, where $T$ is a function of the size of the database $\tilde\D$.
	There cannot be an algorithm $\mathcal{A}_Q$ that computes $Q(\D)$ in time $o(T(|\D|))$ (i.e., asymptotically strictly smaller), for any database $\D$.
}

\begin{proof}
Suppose, for a contradiction, that there is such an algorithm $\mathcal{A}_Q$.
We will show that we can construct an algorithm $\mathcal{A}_{\tilde{Q}}$
based on $\mathcal{A}_Q$ that can solve $\tilde{Q}(\tilde\D)$ in time
complexity $o(T(|\tilde\D|))$ (i.e., asymptotically
strictly smaller), for any input database  $\tilde\D$ of $\tilde{Q}$.

Let $\tilde\D$ be any input database for $\bcq$ query $\tilde{Q}$.
We will base our construction on the structure of the segment tree.
Let $\text{brep}(x)$ be the binary representation of the natural number $x$.
WLOG we can assume that each value in $\tilde\D$ is a binary string of length exactly $b$ for some
constant $b$.
Let $n = |\calE| = |\tilde{\calE}|$.


Consider a slightly modified version of a perfect segment tree with
$2^{n \cdot b}$ leaves (so with height equal to $n \cdot b$)  where, for
each node $u$, we have $\text{seg}(u) := [x, y]$, where $x$ and $y$
are natural numbers such that
$\text{brep}(x) := \onestring \circ u \circ \zerostring^\ell$ and
$\text{brep}(y) := \onestring \circ u \circ \onestring^\ell$,
where $\ell := n \cdot b - |u|$ and $\zerostring^\ell$ represents the string
$\zerostring$ repeated $\ell$ times (and the same for $\onestring^\ell$).
Figure~\ref{fig:reverse-reduction-segment-tree} depicts this segment
tree for $n = 2$ and $b = 2$. 
Note that all the properties of a segment tree that are relevant for this proof hold for this version too.

\tikzset{every node/.style={rectangle, align=center, font=\footnotesize, inner sep=2pt}}

\begin{figure}
    \centering
    
\begin{tikzpicture}[->,
 level 1/.style={sibling distance=88mm},
 level 2/.style={sibling distance=44mm},
 level 3/.style={sibling distance=22mm},
  level 4/.style={sibling distance=11mm},
  level distance=17mm
]
\node [] {$\varepsilon$\\$[16, 31]$}
  	child {node [] {0\\$[16, 23]$}   
    	child {node [] {00\\$[16, 19]$}
      		child {node [] {000\\$[16, 17]$}
          		child {node[] {0000\\$[16, 16]$}} 
          		child {node[] {0001\\$[17, 17]$}}
    	    }
      		child {node [] {001\\$[18, 19]$}
          		child {node[] {0010\\$[18, 18]$}} 
          		child {node[] {0011\\$[19, 19]$}}
    	    }
    	}
    	child {node [] {01\\$[20, 23]$}
      		child {node [] {010\\$[20, 21]$}
          		child {node [] {0100\\$[20, 20]$}}
          		child {node [] {0101\\$[21, 21]$}}
    	    }
      		child {node [] {011\\$[22, 23]$}
          		child {node [] {0110\\$[22, 22]$}}
          		child {node [] {0111\\$[23, 23]$}}
    	    }
    	}
    }
  	child {node [] {1\\$[24, 31]$}
    	child {node [] {10\\$[24, 27]$}
      		child {node [] {100\\$[24, 25]$}
          		child {node [] {1000\\$[24, 24]$}}
          		child {node [] {1001\\$[25, 25]$}}
    	    }
      		child {node [] {101\\$[26, 27]$}
          		child {node [] {1010\\$[26, 26]$}}
          		child {node [] {1011\\$[27, 27]$}}
    	    }
    	}
  		child {node [] {11\\$[28 31]$}
    		child {node [] {110\\$[28, 29]$}
          		child {node [] {1100\\$[28, 28]$}}
          		child {node [] {1101\\$[29, 29]$}}
    	    }
    		child {node [] {111\\$[30, 31]$}
         		child {node [] {1110\\$[30, 30]$}}
         		child {node [] {1111\\$[31, 31]$}}
    	    }
  		}
	};
\end{tikzpicture}
    
    \caption{A slightly modified version of a perfect segment tree with $2^{n \cdot b}$ leaves, where $n=2$ and $b=2$, to be used as a tool to prove Theorem~\ref{theorem:lower-bound}.}
    \label{fig:reverse-reduction-segment-tree}
	\Description[]{}
\end{figure}
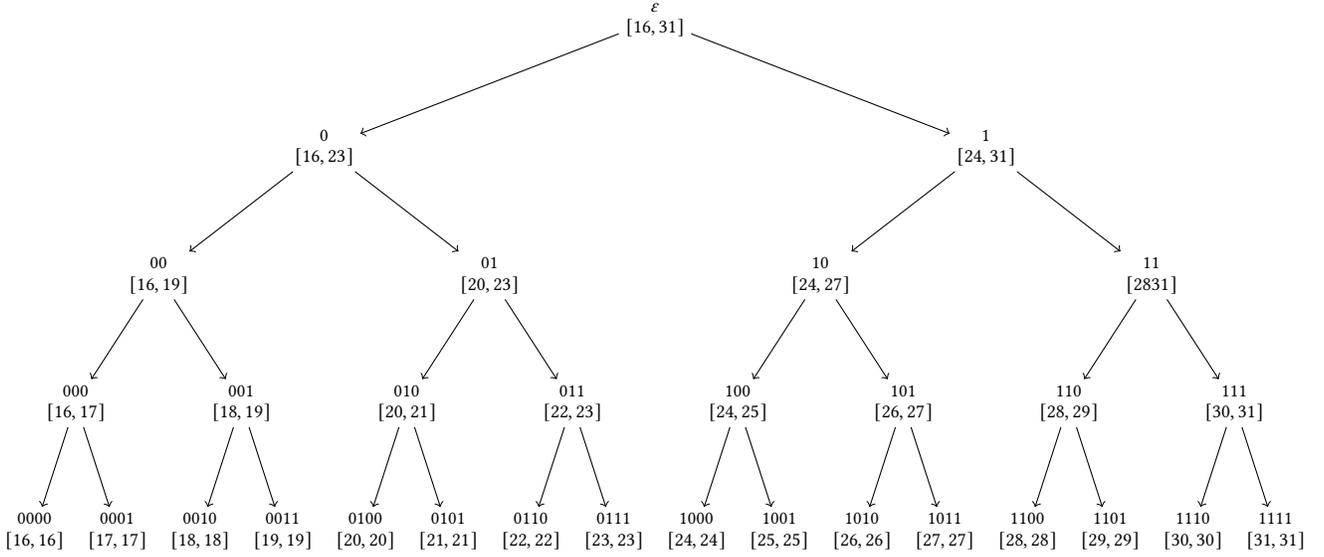

\begin{remark}
We will not construct this segment tree explicitly, since its size is $O(|\tilde\D|^n)$ and it thus cannot be constructed explicitly in the desired runtime bound. It will only be used as a theoretical tool for the proof. We choose this version of the segment tree because it enables us to easily compute the segment corresponding to any node, without having to explicitly construct the tree upfront.
\end{remark}

Similar to the one-step forward reduction from Section~\ref{section:reduction-onestep}, we define
a one-step backward reduction and then apply it repeatedly.

\begin{definition}[One-step backward database transformation]
    Given an \query $Q$, let $[X]$ be an interval variable of $Q$ and let $\sigma\in\pi(\calE_{[X]})$.
    Let $\tilde Q_{([X], \sigma)}$ be the \query resulting from the one-step query rewriting from
    Definition~\ref{definition:one-step-rewriting}.
    Let $\tilde \D_{([X], \sigma)}$ be an arbitrary database instance over the schema of $\tilde Q_{([X], \sigma)}$.
    We construct a database instance $\D$ over the schema of $Q$ as follows.
    For each tuple $\tilde t\in\tilde{R}(\tilde{\sigma_i})$, we construct a tuple $t\in R(\sigma_i)$ such that:
    \begin{itemize}
    	\item $t(\sigma_i\setminus\{[X]\}) = \tilde t (\tilde{\sigma_i}\setminus\{X_1,\ldots,X_i\})$
    	\item $t([X]) = \text{seg}(\tilde t(X_1)\circ\cdots\circ \tilde t(X_i))$
    \end{itemize}
    All relations in $\tilde{\D}_{([X],\sigma)}$ other than $\{R(\sigma_i) \suchthat i \in[k]\}$
    where $k := |\calE_{[X]}|$
    are copied directly to $\D$.
    \label{defn:reverse-one-step}
\end{definition}
Note that by the above definition, we have a bijection between tuples $\tilde t$ from $\tilde \D$
and tuples $t$ from $\D$. Let $g$ be a function that maps a tuple $\tilde t$ in $\tilde \D$
to the corresponding tuple $t$ from $\D$.
Because of this bijection, we also have $|\D| = |\tilde \D_{([X], \sigma)}|$.
\begin{claim}
    Given $Q, \D, \tilde Q_{([X], \sigma)}$ and $\tilde \D_{([X], \sigma)}$ from Definition~\ref{defn:reverse-one-step},
    $Q(\D)$ holds if and only if $\tilde Q_{([X], \sigma)}(\tilde \D_{([X], \sigma)})$ holds.
    \label{clm:reverse-one-step}
\end{claim}
First we prove both directions of the above claim:
\begin{proof}[Proof of Claim~\ref{clm:reverse-one-step}]
I) $\Leftarrow$: Assume that $\tilde Q_{([X], \sigma)}(\tilde\D_{([X], \sigma)})$ holds.
Let $\tilde \calH_{([X], \sigma)}$ be the corresponding hypergraph whose hyperedges are
$\tilde\calE_{([X], \sigma)}$.
This means that each relation
$\tilde R(\tilde e)$ in $\tilde\D_{([X], \sigma)}$
contains a tuple $\tilde t_{\tilde e}$ such that the tuples $\left(\tilde t_{\tilde e}\right)_{\tilde e \in \tilde\calE_{([X], \sigma)}}$ satisfy $\tilde Q_{([X], \sigma)}$.
Let $k := |\calE_{[X]}|$.
There must exist a tuple $(x_1, \ldots, x_k)$ where for each $i \in [k]$, we have
$\tilde t_{\tilde \sigma_i}(X_1) = x_1, \ldots, \tilde t_{\tilde \sigma_i}(X_i) = x_i$.
Therefore the binary strings $\{\tilde t_{\tilde \sigma_i}(X_1)\circ\cdots\circ \tilde t_{\tilde \sigma_i}(X_i)\suchthat i \in [k]\}$
are a prefix of one another.
By Property (1) from~\ref{property:segment-tree}, this means that the line segments
$\{\text{seg}(\tilde t_{\tilde \sigma_i}(X_1)\circ\cdots\circ \tilde t_{\tilde \sigma_i}(X_i))\suchthat i \in [k]\}$ intersect.
Therefore, the tuples $\left(t_e := g(\tilde t_{\tilde e})\right)_{\tilde e \in \tilde\calE_{([X], \sigma)}}$ satisfy $Q$.

II) $\Rightarrow$: Assume that $Q(\D)$ is true.
This means that each relation $R_e$ contains a tuple $t_e$ such that the tuples
$(t_e)_{e \in \calE}$ satisfy $Q$.
In particular, the intervals
$\{t_{\sigma_i}([X])\suchthat i \in [k]\}$ intersect. By Definition~\ref{defn:reverse-one-step},
the tuples $(\tilde t_{\tilde \sigma_i}:= g^{-1}(t_{\sigma_i})) \in \tilde R(\tilde \sigma_i)$ satisfy
$t_{\sigma_i}([X]) = \text{seg}(\tilde t_{\tilde \sigma_i}(X_1)\circ\cdots\circ \tilde t_{\tilde \sigma_i}(X_i))$ for $i \in [k]$.
By Property (1) from~\ref{property:segment-tree}, the binary strings
$\{\tilde t_{\tilde\sigma_i}(X_1)\circ\cdots\circ \tilde t_{\tilde\sigma_i}(X_i) \suchthat i \in [k]\}$ are a prefix of one another.
Moreover for each $i\in [k]$, the binary strings $\tilde t_{\tilde\sigma_i}(X_1), \ldots, \tilde t_{\tilde\sigma_i}(X_i)$
have the same length which is $b$ (by our assumption about $\tilde \D$).
Hence, for each $j \in [k]$, $\tilde t_{\tilde\sigma_j}(X_j)= \tilde t_{\tilde\sigma_{j+1}}(X_{j}) = \cdots = \tilde t_{\tilde\sigma_k}(X_{j})$.
Therefore, the tuples $\left(\tilde t_{\tilde e}=g^{-1}(t_e)\right)_{e\in\calE}$ satisfy $\tilde Q_{([X], \sigma)}$.
\end{proof}

Finally we use the above claim to finish the proof of Theorem~\ref{theorem:lower-bound}.
By repeatedly applying the above reduction on the \bcq~query $\tilde Q$ and its database instance $\tilde\D$,
we construct a database instance $\D$ such that $Q(\D)$ holds if and only if $\tilde Q(\tilde \D)$ holds.
Moreover by Definition~\ref{defn:reverse-one-step}, we have $|\D|= |\tilde \D|$.
Now we use the algorithm $\mathcal{A}_Q$ to answer $Q(\D)$ in time $o(T(|\D|))$, thus resulting in an
algorithm $\mathcal{A}_{\tilde{Q}}$ that can answer $\tilde Q(\tilde \D)$ in time $o(T(|\tilde\D|))$.
But this is a contradiction since $\tilde Q(\tilde \D)$ has a lower bound of $\Omega(T(|\tilde \D|))$.
\end{proof}

\section{Missing Details from Section~\ref{section:iota-acyclicity}}
\label{appendix:iota-acyclicity}

In this section we prove the statements of Section~\ref{section:iota-acyclicity}. 
Consider an \bcqij $Q$ as per Definition~\ref{definition:bcqij}, and let $\mathcal{H} = (\mathcal{V}, \mathcal{E})$ be the hypergraph of $Q$.
We denote the vertices of $\mathcal{H}$ by letters (e.g. $u$) and the vertices of any $\tilde{\mathcal{H}} = (\tilde{\mathcal{V}}, \tilde{\mathcal{E}})\in\tau(\mathcal{H})$ by letters with tilde (e.g. $\tilde{u}$). For each vertex $u \in \calV$, let $n_u = |\calE_u|$ be the number of hyperedges that contain $u$. Each vertex $u$ that occurs in $n_u$ hyperedges in $\mathcal{H}$ corresponds to $n_u$ vertices in $\tilde{\mathcal{H}}$ denoted by $\tilde{u}_1,\ldots,\tilde{u}_{n_u}$. Figure~\ref{fig:non-iota-example} exemplifies this notation: The hypergraph $\mathcal{H}$ has vertices $[A], [B], [C]$, with $n_A=3$, $n_B=2$, and $n_C=2$. The corresponding vertices in the hypergraph $\tilde{\mathcal{H}}$ are:  $\tilde{A}_{1}, \tilde{A}_{2}, \tilde{A}_{3}$ for $[A]$; $\tilde{B}_{1}, \tilde{B}_{2}$ for $[B]$; and $\tilde{C}_{1}, \tilde{C}_{2}$ for $[C]$.

\begin{definition}
    \label{def:help-functions-from-tilde-to-normal}
    Let $\mathcal{H = (V, E)}$ be a hypergraph and $\tilde{\mathcal{H}} = (\tilde{\mathcal{V}}, \tilde{\mathcal{E}})$ be any member of $\tau(\mathcal{H})$.
    \begin{enumerate}
    \item  Let $\nu_{\tilde{\mathcal{H}}, \mathcal{H}} : \tilde{\mathcal{V}} \rightarrow \mathcal{V}$ be the surjective function that maps each vertex $\tilde{u} \in \tilde{\mathcal{V}}$ to the corresponding vertex $u \in \mathcal{V}$. 
    \item Let $\epsilon_{\tilde{\mathcal{H}}, \mathcal{H}} : \tilde{\mathcal{E}} \rightarrow \mathcal{E}$ be the bijective function that maps each hyperedge $\tilde{e} \in \tilde{\mathcal{E}}$ to the corresponding hyperedge $e = \{u \mid \tilde{u} \in \tilde{e}\} \in \mathcal{E}$.
    \end{enumerate}
    Whenever it is clear from the context, we omit the subscript from the names of the functions.
\end{definition}

The following properties hold immediately by Definition~\ref{def:help-functions-from-tilde-to-normal} and Algorithm~\ref{algorithm:reduction}.

\begin{property}[Properties of \bcqij-to-\bcq Reduction]
    \label{property:ij-hypergraph-properties}
    Let $\mathcal{H = (V, E)}$ be a hypergraph and $\tilde{\mathcal{H}} = (\tilde{\mathcal{V}}, \tilde{\mathcal{E}})\in\tau(\mathcal{H})$.  
    \begin{enumerate}
    \item 
    \label{property:u-in-tilde-e-implies-u-in-e}
    For each hyperedge $\tilde{e} \in \tilde{\mathcal{E}}$ and each vertex $\tilde{u} \in \tilde{\mathcal{V}}$, if $\tilde{u} \in \tilde{e}$ in $\tilde{\mathcal{H}}$ then $\nu(\tilde{u}) \in \epsilon(\tilde{e})$ in $\mathcal{H}$. 

    \item 
    \label{property:u1-in-tilde-e-iff-u-in-e}
    For each hyperedge $\tilde{e} \in \tilde{\mathcal{E}}$ and each vertex $u \in \mathcal{V}$, $\tilde{u}_1 \in \tilde{e}$ in $\tilde{\mathcal{H}}$ if and only if $u \in \epsilon(\tilde{e})$ in $\mathcal{H}$. 
    
    \item 
    \label{property:uj-in-e-implies-ui-in-e}
    For any two vertices $\tilde{u}_{i}, \tilde{u}_{j} \in \tilde{\mathcal{V}}$ with $i < j$ and $\tilde{e} \in \tilde{\mathcal{E}}$, we have that if $\tilde{u}_{j} \in \tilde{e}$ then $\tilde{u}_{i} \in \tilde{e}$.
    
    \end{enumerate}
\end{property}

The first property states that any node in $\tilde{\mathcal{H}}$ is mapped back to precisely one node in ${\cal H}$ and there is a bijection between the hyperedges of these two nodes. The second property is a strengthening of the first property: there is always a node $\tilde{u}_1$ in $\tilde{\mathcal{H}}$ for every node $u$ in $\mathcal{H}$ and there is a bijection between the hyperedges of these two nodes. Finally, the third property states that whenever we have a node $\tilde{u}_{j}$ in a hyperedge in $\tilde{\mathcal{H}}$, which corresponds to a node $u$ in $\mathcal{H}$, we also have all nodes $\tilde{u}_{1},\ldots,\tilde{u}_{j-1}$ in that hyperedge.

The following lemma is an essential building block of the proofs of the main statements.

\begin{lemma}
    \label{lemma:special-Berge cycle-in-tilde-H-implies-Berge cycle-in-H}
    If $\tilde{\mathcal{H}}$ has a Berge cycle $(\tilde{e}^1, \tilde{v}^{1}, \tilde{e}^2, \tilde{v}^{2}, \dots, \tilde{e}^k, \tilde{v}^{k}, \tilde{e}^{k+1} = \tilde{e}^1)$ of length $k$ such that 
	$\nu(\tilde{v}^{1}), \dots, \nu(\tilde{v}^{k})$ are pairwise distinct vertices from $\mathcal{V}$, 
	then $\mathcal{H}$ also has a Berge cycle of length $k$.
	\begin{proof}
		We use a proof by construction. Assume that the above statement is true. Since $(\tilde{e}^1, \tilde{v}^{1}, \tilde{e}^2, \tilde{v}^{2}, \dots, \tilde{e}^k, \tilde{v}^{k}, \tilde{e}^{k+1} = \tilde{e}^1)$ is a Berge cycle, for each $1 \leq i \leq k$, 
		we have $\tilde{v}^{i} \in \tilde{e}^i$ and $\tilde{v}^{i} \in \tilde{e}^{i+1}$.
		Hence, by Property~\eqref{property:u-in-tilde-e-implies-u-in-e} of \ref{property:ij-hypergraph-properties}, for each $1 \leq i \leq k$, we get that $\nu(\tilde{v}^{i}) \in \epsilon(\tilde{e}^i)$ and $\nu(\tilde{v}^{i}) \in \epsilon(\tilde{e}^{i+1})$. 
		Since $\epsilon$ is a bijection and $\tilde{e}^1, \dots, \tilde{e}^k$ are pairwise distinct hyperedges of $\tilde{\mathcal{E}}$, we get that $\epsilon(\tilde{e}^1), \dots, \epsilon(\tilde{e}^k)$ are pairwise distinct hyperedges from $\mathcal{E}$. 
		Therefore, the sequence $(\epsilon(\tilde{e}^1), \nu(\tilde{v}^{1}), \dots, \epsilon(\tilde{e}^k), \nu(\tilde{v}^{k}), \epsilon(\tilde{e}^{k+1}) = \epsilon(\tilde{e}^1))$ is a Berge cycle of length $k$ in $\mathcal{H}$.
	\end{proof}
\end{lemma}

\subsection{Proof of Theorem~\ref{theorem:iota-acyclicity}}

{
	\textsc{Theorem 
	\ref{theorem:iota-acyclicity}}
	A hypergraph $\mathcal{H = (V, E)}$ is $\iota$-acyclic if and only if $\mathcal{H}$ has no Berge cycle of length strictly greater than two.
}

\begin{proof}
		$\Rightarrow$: Assume for a contradiction that $\mathcal{H}$ has a Berge cycle of length strictly greater than two, or equivalently at least three. 
		Hence, there exist a cyclic sequence $(e^1, v^{1}, e^2, v^{2}, \dots,$ $e^k, v^{k}, e^{k+1} = e^{1})$ 
		such that $k \geq 3$, $v^{1}, \dots, v^{k}$ are pairwise distinct vertices from $\mathcal{V}$, $e^1, \dots, e^k$ 
		are pairwise distinct hyperedges from $\mathcal{E}$, and for each $1 \leq i \leq k$, we have $v^{i} \in e^i$ and $v^{i} \in e^{i+1}$. 

		By our construction in Algorithm~\ref{algorithm:reduction}, there exists a hypergraph $\tilde{\mathcal{H}} = (\tilde{\mathcal{V}}, \tilde{\mathcal{E}})\in\tau(\mathcal{H})$ such that for each $1 \leq i \leq k$ we have:
		\[
			\{\tilde{v}^{i}_{1}, \dots, \tilde{v}^{i}_{n_{v^{i}} - 1} \} \subseteq \epsilon^{-1}(e^i) 
			\text{ and } \tilde{v}^{i}_{n_{v^{i}}} \notin \epsilon^{-1}(e^i)
		\] 
		and
		\[
			\{\tilde{v}^{i}_{1}, \dots, \tilde{v}^{i}_{n_{v^{i}}} \} \subseteq \epsilon^{-1}(e^{i+1}).
		\]
		Since $k \geq 3$, the hypergraph $\tilde{\mathcal{H}}$ has the following three properties:
		\begin{enumerate}
			\item For each $1 \leq i \leq k$ the vertex $\tilde{v}^{i}_{n_{v^{i}} - 1}$ belongs to precisely two hyperedges from $\tilde{\mathcal{E}}$. 
			These hyperedges are $\epsilon^{-1}(e^i)$ and $\epsilon^{-1}(e^{i+1})$;
			
			\item For each $1 \leq i, j \leq k$ with $|i - j| \leq 2$, the hyperedge $\epsilon^{-1}(e^i)$ cannot be contained in hyperedge $\epsilon^{-1}(e^j)$ because, 
			by (1) the vertex $\tilde{v}^{i}_{n_{v^{i}} - 1}$ belongs to $\epsilon^{-1}(e^i)$ but cannot belong to $\epsilon^{-1}(e^j)$;
			
			\item For each $1 \leq i \leq k$, the hyperedges $\epsilon^{-1}(e^i)$ and $\epsilon^{-1}(e^{i+1})$ cannot be subset of each other because by (1) we have:
				\begin{itemize}
					\item $\tilde{v}^{i-1}_{n_{v^{i-1}} - 1}$ belongs to $\epsilon^{-1}(e^i)$ but cannot belong to $\epsilon^{-1}(e^{i+1})$, and
					\item $\tilde{v}^{i+1}_{n_{v^{i+1}} - 1}$ belongs to $\epsilon^{-1}(e^{i+1})$ but cannot belong to $\epsilon^{-1}(e^i)$.
				\end{itemize}
		\end{enumerate}
		Let $\tilde{\mathcal{V}}^{\prime} = \{\tilde{v}^{i}_{n_{v^i} - 1} \mid 1 \leq i \leq k\}$ and $\tilde{\mathcal{E}}^{\prime} = \{\epsilon^{-1}(e^i) \mid 1 \leq i \leq k\}$. 
		Note that $\tilde{\calV}^{\prime} \subseteq \tilde{\calV}$ and $\tilde{\calE}^{\prime} \subseteq \tilde{\calE}$.
		Therefore, no matter what other steps are taken during the runtime of the GYO reduction on $\tilde{\mathcal{H}}$, by (1), no vertex from $\tilde{\mathcal{V}}^{\prime}$ will become candidate for removal, 
		and by (2) and (3), no hyperedge from $\tilde{\mathcal{E}}^{\prime}$ will become candidate for removal. 
		Hence, $\tilde{\mathcal{H}}$ cannot be GYO reducible to the empty hypergraph. 
		In other words, the hypergraph $\tilde{\calH}$ includes the cycle 
		\[
			\{\{\tilde{v}^{i}_{n_{v^{i}}-1}, \tilde{v}^{i+1}_{n_{v^{i+1}}-1}\} \mid 1 \leq i < k\} \cup \{\{\tilde{v}^{k}_{n_{v^{k}}-1}, \tilde{v}^{1}_{n_{v^{1}}-1}\}\},
		\]
		where, by the above property (1), the edge $\{\tilde{v}^{i}_{n_{v^{i}}-1}, \tilde{v}^{i+1}_{n_{v^{i+1}}-1}\}$ for each $1 \leq i < k$ is included in 
		precisely one hyperedge from $\tilde{\calE}$, that is $\epsilon^{-1}(e^{i+1})$, and $\{\tilde{v}^{k}_{n_{v^{k}}-1}, \tilde{v}^{1}_{n_{v^{1}}-1}\}$ is included in 
		precisely one hyperedge from $\tilde{\calE}$, that is $\epsilon^{-1}(e^{1}) (= \epsilon^{-1}(e^{k+1}))$. 
		Thus, $\tilde{\mathcal{H}}$ is not $\alpha$-acyclic. 
		Therefore, by Definition~\ref{definition:iota-acyclicity} of $\iota$-acyclicity, $\mathcal{H}$ is not $\iota$-acyclic. 
		This contradicts the initial assumption.
		
		$\Leftarrow$: Assume for a contradiction that $\mathcal{H}$ is not $\iota$-acyclic. 
		Hence, by Definition~\ref{definition:iota-acyclicity}, there exists $\tilde{\mathcal{H}} = (\tilde{\mathcal{V}}, \tilde{\mathcal{E}}) \in \tau(\mathcal{H})$ that is not $\alpha$-acyclic. 
		Therefore, by Definition~\ref{definition:alpha-conformal-cycle-free}, $\tilde{\mathcal{H}}$ is not conformal or not cycle-free. 
		Next, we prove that each of the two cases leads to a contradiction.
				 
		 {\bf Case 1.} The hypergraph $\tilde{\mathcal{H}}$ is not conformal. 
		Therefore, there exists a subset $S \subseteq \tilde{\mathcal{V}}$ with $|S| \geq 3$ such that 
		\[	
			\mathcal{M}(\tilde{\mathcal{E}}[S]) = \{S \setminus \{x\} \mid x \in S\}.
		\] 
		According to Definition~\ref{definition:induced-set} of an \textit{induced set}, we have $\tilde{\calE}[S] = \{e \cap S \mid e \in \tilde{\calE}\} \setminus \emptyset$.
		According to Definition~\ref{definition:minimisation} of the \textit{minimisation of a family of sets}, we have $\mathcal{M}(\tilde{\mathcal{E}}[S]) = \{e \in \tilde{\calE}[S] \mid \nexists f \in \tilde{\calE}[S], e \subset f\}$.
		Let $\tilde{x}, \tilde{y}, \tilde{z} \in S$ be distinct vertices from $S$. 
		Let $\tilde{e}_{\tilde{x}} = S \setminus \{\tilde{x}\}$, $\tilde{e}_{\tilde{y}} = S \setminus \{\tilde{y}\}$, and $\tilde{e}_{\tilde{z}} = S \setminus \{\tilde{z}\}$ be distinct hyperedges from $\mathcal{M}(\tilde{\mathcal{E}}[S])$. 
		Since $S \subseteq \tilde{\mathcal{V}}$, we also have that $\tilde{x}, \tilde{y}, \tilde{z} \in \tilde{\mathcal{V}}$.
				
		We now need the following claim at this point in the proof; its own proof is given at the end of this section.
		
		\begin{claim}
			\label{claim:pairwise-distinct-vertices-1}
		The vertices $\nu(\tilde{x})$, $\nu(\tilde{y})$, and $\nu(\tilde{z})$ are pairwise distinct vertices of $\mathcal{V}$.
		\end{claim}

		By Definition~\ref{definition:minimisation} of the minimization of a family of sets, we have $\mathcal{M}(\tilde{\mathcal{E}}[S]) \subseteq \tilde{\mathcal{E}}[S]$, hence, $e_{\tilde{x}}$, $e_{\tilde{y}}$, and $e_{\tilde{z}}$ belong also to $\tilde{\mathcal{E}}[S]$. By Definition~\ref{definition:induced-set} of the induced set, 
		this means that there exist three distinct hyperedges $\tilde{c}_{\tilde{x}}, \tilde{c}_{\tilde{y}}, \tilde{c}_{\tilde{z}} \in \tilde{\mathcal{E}}$ such that $\tilde{e}_{\tilde{x}} \subseteq \tilde{c}_{\tilde{x}}$, $e_{\tilde{y}} \subseteq \tilde{c}_{\tilde{y}}$, and $e_{\tilde{z}} \subseteq \tilde{c}_{\tilde{z}}$. 
		Therefore, the sequence $(\tilde{c}_{\tilde{x}}, \tilde{z}, \tilde{c}_{\tilde{y}}, \tilde{x}, \tilde{c}_{\tilde{z}}, \tilde{y}, \tilde{e}_{\tilde{x}})$ is a Berge cycle of length 3 in $\tilde{\mathcal{H}}$ where, by Claim~\ref{claim:pairwise-distinct-vertices-1}, $\nu(\tilde{x}), \nu(\tilde{y}), \nu(\tilde{z})$ are pairwise distinct vertices of $\mathcal{H}$. 
		Therefore, by Lemma~\ref{lemma:special-Berge cycle-in-tilde-H-implies-Berge cycle-in-H}, $\mathcal{H}$ has also a Berge cycle of length 3. This statement contradicts the initial assumption that $\mathcal{H}$ is $\iota$-acyclic.
		 		 
		{\bf Case 2.} The hypergraph $\tilde{\mathcal{H}}$ is non-cycle-free. Hence, there exist $S = \{\tilde{v}^{1}, \dots, \tilde{v}^{k}\} \subseteq \tilde{\mathcal{V}}$ where $k \geq 3$ of pairwise distinct vertices such that 
		\[ 
			\mathcal{M}(\tilde{\mathcal{E}}[S]) = \{\{\tilde{v}^{i}, \tilde{v}^{i+1}\} \mid 1 \leq i < k\} \cup \{\{\tilde{v}^{k}, \tilde{v}^{1}\}\}.
		\]
		Let $\tilde{e}^{i+1} := \{\tilde{v}^{i}, \tilde{v}^{i+1}\}$ for each $1 \leq i < k$, and $\tilde{e}^1 := \{\tilde{v}^{k}, \tilde{v}^{1}\}$.
			
		We now need the following claim at this point in the proof; its own proof is given at the end of this section.
				
		\begin{claim}
			\label{claim:pairwise-distinct-vertices-2}
			The vertices $\nu(\tilde{v}^{1}), \dots, \nu(\tilde{v}^{k})$ are pairwise distinct vertices of $\mathcal{V}$.
		\end{claim}

		By Definition~\ref{definition:minimisation} of the minimization of a family of sets we have $\mathcal{M}(\tilde{\mathcal{E}}[S]) \subseteq \tilde{\mathcal{E}}[S]$, this means $\tilde{e}^i \in \tilde{\mathcal{E}}[S]$ for each $1 \leq i \leq k$. 
		By Definition~\ref{definition:induced-set} of the induced set, there exist $k$ pairwise distinct hyperedges $\tilde{c}^1, \dots, \tilde{c}^k \in \tilde{\mathcal{E}}$ such that $\tilde{e}^i \subseteq \tilde{c}^i$ for each $1 \leq i \leq k$. 
		Since $\tilde{c}^1, \dots, \tilde{c}^k$ are distinct hyperedges in $\tilde{\mathcal{E}}$ and $\tilde{v}^{1}, \dots, \tilde{v}^{k}$ are distinct vertices in $\tilde{\calV}$, the sequence $(\tilde{c}^1, \tilde{v}^{1}, \tilde{c}^2, \tilde{c}^{2}, \dots, \tilde{c}^k, \tilde{v}^{k}, \tilde{c}^1)$ 
		is a Berge cycle of length $k \geq 3$ in $\tilde{\calH}$. Moreover, by Claim~\ref{claim:pairwise-distinct-vertices-2}, $\nu(\tilde{v}^{1}), \dots, \nu(\tilde{v}^{k})$ are pairwise distinct vertices of $\mathcal{H}$. 
		Therefore, by Lemma~\ref{lemma:special-Berge cycle-in-tilde-H-implies-Berge cycle-in-H}, $\mathcal{H}$ has a Berge-cycle
		of length $k\geq 3$. This statement contradicts the initial assumption that $\mathcal{H}$ is $\iota$-acyclic.
				
		Finally, we give the proofs of Claims~\ref{claim:pairwise-distinct-vertices-1} and \ref{claim:pairwise-distinct-vertices-2}.

		\textsc{Proof of Claim~\ref{claim:pairwise-distinct-vertices-1}}. Assume for contradiction that there are two distinct vertices $\tilde{u}, \tilde{v} \in \{\tilde{x}, \tilde{y}, \tilde{z}\}$ such that $\nu(\tilde{u}) = \nu(\tilde{v})$. 
			By Property~\ref{property:ij-hypergraph-properties} \eqref{property:uj-in-e-implies-ui-in-e}, this means that any hyperedge $\tilde{e} \in \tilde{\mathcal{E}}$ that contains $\tilde{u}$ contains also vertex $\tilde{v}$ (or vice versa). 
			Note that, since $\mathcal{M}(\tilde{\mathcal{E}}[S]) \subseteq \tilde{\mathcal{E}}[S]$ and since $\tilde{u}, \tilde{v} \in S$, the property of the previous statement holds also for the hyperedges of $\mathcal{M}(\tilde{\mathcal{E}}[S])$.
			This violates the condition that $\mathcal{M}(\tilde{\mathcal{E}}[S]) = \{S \setminus \{x\} \mid x \in S\}$ since in this case the hyperedge $S \setminus \{\tilde{v}\}$ (which contains vertex $\tilde{u}$) would actually need to include vertex $\tilde{v}$ as well. 
			The reverse case is analogous due to symmetry. Contradiction.

		\textsc{Proof of Claim~\ref{claim:pairwise-distinct-vertices-2}}.
			Assume for contradiction that there are $1 \leq i < j \leq k$ such that $\nu(\tilde{v}^{i}) = \nu(\tilde{v}^{j})$. 
			By Property~\ref{property:ij-hypergraph-properties} \eqref{property:uj-in-e-implies-ui-in-e}, this means that any hyperedge $\tilde{e} \in \tilde{\mathcal{E}}$ that contains vertex $\tilde{v}^{i}$ also contains vertex $\tilde{v}^{j}$ (or vice versa). 
			Since $\mathcal{M}(\tilde{\mathcal{E}}[S]) \subseteq \tilde{\mathcal{E}}[S]$ and since $\tilde{v}^{i}, \tilde{v}^{j} \in S$, we get that the hyperedges from $\mathcal{M}(\tilde{\calE}[S])$ satisfy this property too. That is, any hyperedge $\tilde{e} \in \mathcal{M}(\tilde{\calE}[S])$ that contains vertex $\tilde{v}^i$ also contains vertex $\tilde{v}^j$ (or vice versa).
			This violates the condition that $\mathcal{M}(\tilde{\mathcal{E}}[S]) = \{\tilde{e}^1, \dots, \tilde{e}^k\}$ since, in this case, the hyperedge $\tilde{e}^i = \{\tilde{v}^{i-1}, \tilde{v}^{i}\}$ (in case $\tilde{v}^{i-1} \neq \tilde{v}^{j}$) 
			or the hyperedge $\tilde{e}^{i+1} = \{\tilde{v}^{i}, \tilde{v}^{i+1}\}$ (in case $\tilde{v}^{i+1} \neq \tilde{v}^{j}$) would also include the vertex $\tilde{v}^{j}$, creating a chord in the cycle\footnote{We define $\tilde{v}^0 := \tilde{v}^k$, $\tilde{e}^0 := \tilde{e}^k$, $\tilde{v}^{k+1} := \tilde{v}^1$, and $\tilde{e}^{k+1} := \tilde{e}^1$ (i.e., the sequence is cyclic).}. 
			The reverse case is analogous due to symmetry. Contradiction. 
\end{proof}

\subsection{Proof of Corollary~\ref{corollary:iota-between-berge-and-gamma}}
{
	\textsc{Corollary 
	\ref{corollary:iota-between-berge-and-gamma}}
	The class of $\iota$-acyclic hypergraphs is a strict superset of the class of Berge-acyclic hypergraphs and it is a strict subset of the class of $\gamma$-acyclic hypergraphs.
}
	
\begin{proof}
	The statement that $\iota$-acyclicity strictly includes Berge-acyclicity follows immediately from  Theorem~\ref{theorem:iota-acyclicity}, since $\iota$-acyclicity allows for Berge cycles of length up to two. 

	We next prove the statement that $\iota$-acyclicity is strictly included in  $\gamma$-acyclicity.
	Assume, for a contradiction, that $\mathcal{H}$ is not $\gamma$-acyclic. 
	Then, by Definition~\ref{definition:gamma-acyclicity}, either $\mathcal{H}$ is non-cycle-free or there exist three distinct vertices $x, y, z \in \mathcal{V}$ such that $\{\{x, y\}, \{y, z\}, \{x, y, z\}\} \subseteq \mathcal{E}[\{x, y, z\}]$. 
	Since a Berge cycle consists of at least $3$ distinct vertices and $3$ distinct hyperedges, 
	it means that in both cases the hypergraph $\mathcal{H}$ contains a Berge cycle of length at least $3$. Contradiction. 
	
	To check the strictness of the inclusion, consider the following hypergraph: $\{\{x, y, z\}, \{x, y, z\}, \{x, y, z\}\}$. 
	This hypergraph is: 
	\begin{enumerate}
		\item not $\iota$-acyclic, since it contains the Berge cycle 1-$x$-2-$y$-3-$z$-1 of length 3, where we denote the three hyperedges by 1, 2, and 3;	
		\item $\gamma$-acyclic, since it is cycle-free and there are no three distinct vertices that satisfy the condition from above.
	\end{enumerate}
\end{proof}

A further immediate corollary is the following (we nevertheless give its proof).
\begin{corollary}
	\label{corollary:iota-subset-of-alpha}
	Let $\mathcal{H}$ be a hypergraph. 
	If $\mathcal{H}$ is $\iota$-acyclic, then $\mathcal{H}$ is $\alpha$-acyclic. 	
\end{corollary}
   
\begin{proof} 
	We prove by construction that $\mathcal{H}$ has a join tree, and hence, by Definition~\ref{definition:join-tree}, $\mathcal{H}$ is $\alpha$-acyclic.
		
	Assume that $\mathcal{H}$ is $\iota$-acyclic. By Definition~\ref{definition:iota-acyclicity} this means that all members of $\tau(\mathcal{H})$ are $\alpha$-acyclic. 
	Let $\tilde{\mathcal{H}} = (\tilde{\mathcal{V}}, \tilde{\mathcal{E}})$ be a member of $\tau(\mathcal{H})$.
	Since all members of $\tau(\mathcal{H})$ are $\alpha$-acyclic, then $\tilde{\mathcal{H}}$ is $\alpha$-acyclic, 
	and hence, it has a join tree $(\tilde{\mathcal{T}}, \tilde{\chi})$ where $\tilde{\mathcal{T}}$ is a tree and $\tilde{\chi}$ is a bijection $\tilde{\chi}: V(\tilde{\mathcal{T}}) \rightarrow \tilde{\mathcal{E}}$ 
	such that the connectivity property holds (see Definition~\ref{definition:join-tree}). 
		
	It is possible to construct a join tree $(\mathcal{T}, \chi)$ for $\mathcal{H} = (\mathcal{V}, \mathcal{E})$ as follows: assign $\mathcal{T} := \tilde{\mathcal{T}}$ and for each node $t \in \mathcal{T}$ assign $\chi(t) := \epsilon(\tilde{\chi}(t))$.  
	Next, we show that $(\mathcal{T}, \chi)$ is a valid join tree, that is (1) $\chi$ is a bijection of the from $\chi:V(\mathcal{T})\rightarrow\mathcal{E}$ and (2) connectivity property holds (see Definition~\ref{definition:join-tree}).
		
	\begin{enumerate}
		\item $\chi$ is the composition of the bijections $\tilde{\chi}: V(\tilde{\mathcal{T}}) \rightarrow \tilde{\mathcal{E}}$ and $\epsilon: \tilde{\mathcal{E}} \rightarrow \mathcal{E}$. 
		Hence, it is a bijection of the form $\chi: V(\tilde{\mathcal{T}}) \rightarrow \mathcal{E}$. 
		Since $\mathcal{T}= \tilde{\mathcal{T}}$, we get that $\chi$ is a bijection of the form $\chi: V(\mathcal{T}) \rightarrow \mathcal{E}$.  
		\item Let $v \in \mathcal{V}$ be any vertex of $\mathcal{H}$. 
		By Property~\eqref{property:u1-in-tilde-e-iff-u-in-e} of \ref{property:ij-hypergraph-properties}, for each hyperedge $\tilde{e} \in \tilde{\mathcal{E}}$ we have $v_{1} \in \tilde{e}$ if and only if $v \in \epsilon(\tilde{e})$. 
		Since $\mathcal{T} = \tilde{\mathcal{T}}$ and $\chi(t) = \epsilon(\tilde{\chi}(t))$ for each node $t \in \mathcal{T}$, the set of nodes $\{t \in \tilde{\mathcal{T}} \mid v_{1} \in \tilde{\chi}(t)\}$ is equal to the set of nodes $\{t \in \mathcal{T} \mid v \in \epsilon(\tilde{\chi}(t)) = \chi(t)\}$. 
		Since $(\tilde{\mathcal{T}}, \tilde{\chi})$ is a join tree, the former set of nodes is a non-empty connected subtree of $\tilde{\mathcal{T}}$ (by Definition~\ref{definition:join-tree}). 
		Therefore, the latter set of nodes is also a non-empty connected subtree of $\mathcal{T}$.
	\end{enumerate}
\end{proof}

Figure~\ref{figure:acyclicity-classes} shows the relationship of $\iota$-acyclicity with the notions of Berge-acyclicity, $\gamma$-acyclicity 
and $\alpha$-acyclicity, that is discussed in Corollaries~\ref{corollary:iota-between-berge-and-gamma} and~\ref{corollary:iota-subset-of-alpha}.

\subsection{Proof of Theorem~\ref{theorem:hardness-of-non-iota-acyclic}}
{
	\textsc{Theorem 
	\ref{theorem:hardness-of-non-iota-acyclic} [Iota Acyclicity Dichotomy]}
    Let $Q$  be any~\bcqij query with hypergraph $\mathcal{H}$ and let $D$ be any database.
    \begin{itemize}
    	\item If $\mathcal{H}$ is $\iota$-acyclic, then $Q$ can be computed in time $O(\lvert \D \rvert\cdot \polylog \lvert \D \rvert)$.
    	\item If $\mathcal{H}$ is not $\iota$-acyclic, then there is no algorithm that can compute $Q$ in time $O(\lvert \D \rvert^{4/3-\epsilon})$ for $\epsilon>0$, unless the \textsc{3SUM} conjecture fails.   
    \end{itemize}
}

\begin{proof}
The linear-time complexity in case $\mathcal{H}$ is $\iota$-acyclic follows immediately: Since each hypergraph in $\tau({\cal H})$ is $\alpha$-acyclic, its corresponding \bcq query can be computed in linear time using Yannakakis's algorithm~\cite{vldb/Yannakakis81}. Furthermore, the size of $\tau({\cal H})$ is independent of the size of input database $\D$.

We next prove the hardness in case $\mathcal{H}$ is not $\iota$-acyclic. Suppose, for a contradiction, that there exists an algorithm $\mathcal{A}_Q$ that can solve $Q(\D)$ in time $O(|\D|^{4/3 - \epsilon})$, for some $\epsilon > 0$. 
Since $\mathcal{H}$ is not $\iota$-acyclic, by Definition \ref{definition:iota-acyclicity}, $\mathcal{H}$ has a Berge cycle $(e^1, v^{1}, e^2, v^{2}$, $\dots, e^k, v^{k}, e^{k+1} = e^1)$ of length $k \geq 3$. 
This means that $v^{1}, \dots, v^{k}$ are pairwise distinct vertices from $\mathcal{V}$, $e^1, \dots, e^k$ are pairwise distinct hyperedges from $\mathcal{E}$, and for each $1 \leq i \leq k$, $v^{i} \in e^i$ and $v^{i} \in e^{i+1}$. 
Let us denote the relations corresponding to the hyperedges $e^1, \dots, e^k$ by $R_1, \dots, R_k \in \D$, respectively.
Assume, without loss of generality, that for each $1 \leq i \leq k$, the first two variables in the relation schema $R_i$ are $v^{i-1}$ and $v^i$ (we define $v^0 := v^k$ since the sequence is cyclic).
Let 
\[
Q^{\prime} := S_1(X_k, X_1) \land S_2(X_1, X_2) \land \dots \land S_k(X_{k-1}, X_k)
\] 
be the $k$-cycle \bcq query, i.e., the $k$-cycle Boolean conjunctive query with equality joins. 
We will show that we can construct an algorithm $\mathcal{A}_{Q^{\prime}}$ based on $\mathcal{A}_Q$ that can solve $Q'(\D')$ in time $O(|\D'|^{4/3 - \epsilon})$, for any input database $\D'=(S_1,\ldots,S_k)$.

We construct the following input database $\D$ for the \bcqij query $Q$:

\begin{itemize}
\item For each $1 \leq i \leq k$ and for each tuple $(a, b) \in S_i$ from $\D^{\prime}$, include in the relation $R_i$ from $\D$ the tuple $([a, a], [b, b], (-\infty, +\infty)$, $\dots, (-\infty, +\infty))$. 
That is, the tuple where the value of $v^{i-1}$ is the point interval $[a, a]$, the value of $v^i$ is the point interval $[b, b]$, and the value of each other variable from $R_i$ is the interval $(-\infty, +\infty)$;
\item Each relation $R$ other than $R_1, \dots, R_k$ from $\D$ consists of \textit{exactly one} tuple: $((-\infty, +\infty)$, $\dots, (-\infty, +\infty))$. 
That is, the tuple where the value of each variable from $R$ is the interval $(-\infty, +\infty)$.
\end{itemize}
	
By construction, $|\D| = O(|\D^{\prime}|)$. Moreover, since the interval $(-\infty, +\infty)$ joins with any other interval and since the intervals $[a, a]$ and $[b, b]$ join if and only if $a=b$, we have that
the set of satisfying assignments of the variables $X_1, \dots, X_k$ for the \bcq query $Q^{\prime}(\D')$ is in bijection with the set of satisfying assignments of the variables in $\calV$ for the \bcqij query $Q(\D)$. That is, the satisfying assignment $(X_1, \dots, X_k) = (x_1, \dots, x_k)$ for $Q'(\D')$ maps to the satisfying assignment for $Q(\D)$ which, for each $1 \leq i \leq k$, sets the value of attribute $v^i$ to $[x_i, x_i]$, and which sets the value of all other variables to $(-\infty, +\infty)$.
		
Therefore, on any input database  $\D^{\prime}$, the answer of $Q^{\prime}(\D^{\prime})$ is equal to the answer of $Q(\D)$. Hence, on the input database  $\D^{\prime}$, the algorithm $\mathcal{A}_{Q'}$ first constructs the database $\D$ in time $O(\lvert \D^{\prime}\rvert)$ and then calls the algorithm $\mathcal{A}_Q$ on input $\D$ (which runs in time $O(|\D|^{4/3 - \epsilon}) = O(|\D'|^{4/3 - \epsilon}|)$, since $|\D| = O(|\D'|)$), and returns its answer.
		
Thus, the algorithm $\mathcal{A}_{Q^{\prime}}$ solves $Q'(\D')$ in time $O(|\D'|^{4/3 - \epsilon})$, for any input database $\D^{\prime}$. However, the $k$-cycle query $Q^{\prime}$ is not $\alpha$-acyclic and cannot be computed in time $O(\lvert \D^{\prime} \rvert^{4/3-\epsilon})$ for $\epsilon >0$~\cite{tods/KhamisNRR16}, unless the widely-held \textsc{3SUM} conjecture fails~\cite{DBLP:conf/stoc/Patrascu10}. Contradiction.
\end{proof}

\subsection{Examples}
\label{appendix:examples-iota}

\begin{figure}
	\centering 
	\begin{subfigure}{0.3\textwidth}
		\includegraphics[width=\linewidth]{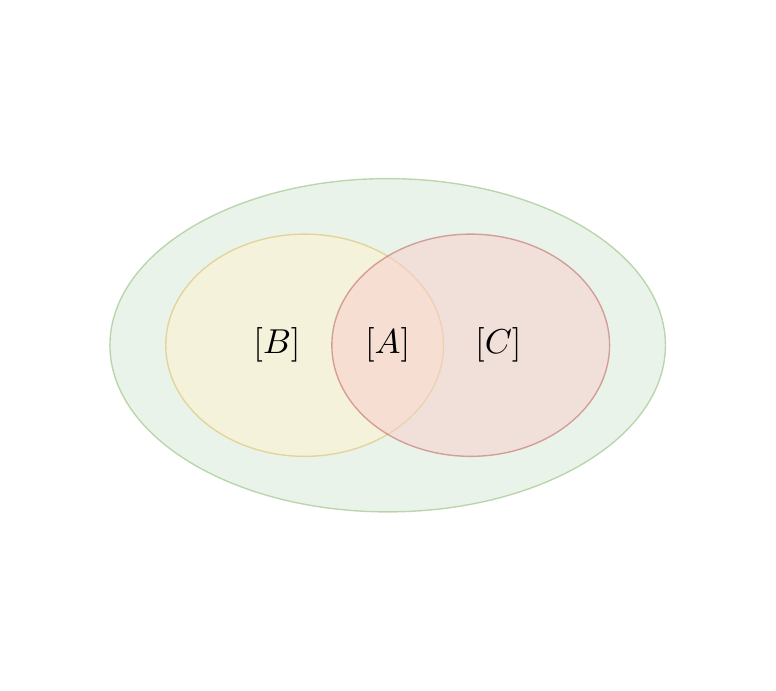}
		\caption{$\mathcal{H} = (\calV, \calE)$}
		\label{fig:non-iota-acyclic}
	\end{subfigure}\hfil 
	\begin{subfigure}{0.3\textwidth}
		\includegraphics[width=\linewidth]{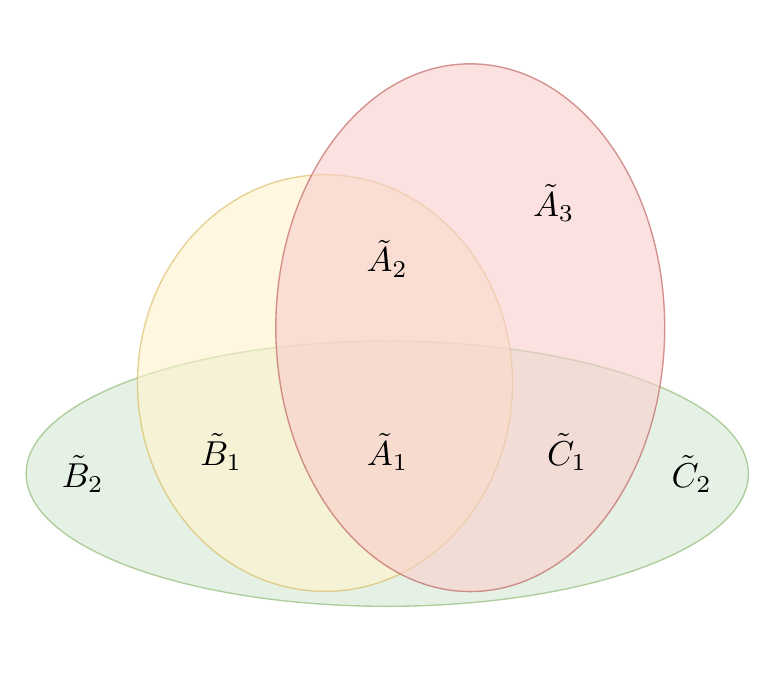}
		\caption{$\tilde{\mathcal{H}} = (\tilde{\calV}, \tilde{\calE}) \in \text{IJ}(\mathcal{H})$}
		\label{fig:non-iota-acyclic-1}
	\end{subfigure} 
	\caption{
		(a) Hypergraph $\mathcal{H}$ ($\alpha$-acyclic but not $\gamma$-acyclic). (b) Hypergraph $\tilde{\mathcal{H}}$ from $\tau(\mathcal{H})$ (not $\alpha$-acyclic). 
		There is a correspondence between each node $u$ in $\mathcal{H}$ and nodes $\tilde{u}_1,\ldots,\tilde{u}_{n_u}$ in $\tilde{\mathcal{H}}$, where $n_u$ is the number of edges containing $u$ in $\mathcal{H}$. For instance, vertex $[A]$ occurs in three hyperedges in $\mathcal{H}$ so there are three corresponding nodes $\tilde{A}_1$, $\tilde{A}_2$, and $\tilde{A}_3$ in $\tilde{\mathcal{H}}$.
			} 
	\label{fig:non-iota-example}
	\Description[]{}
\end{figure}

According to Corollary~\ref{corollary:iota-between-berge-and-gamma}, the class of $\iota$-acyclic hypergraphs is a strict subset of the class of $\gamma$-hypergraphs. 
Figure~\ref{fig:iota-example-hypergraphs} depicts six $\alpha$-acyclic hypergraphs. 
The hypergraph of Figure~\ref{fig:non-iota-hyp-3} is $\alpha$-acyclic but not $\gamma$-acyclic.
The hypergraphs of Figures~\ref{fig:non-iota-hyp-1}-~\ref{fig:non-iota-hyp-2} are $\gamma$-acyclic but not $\iota$-acyclic. 
The hypergraphs of Figures~\ref{fig:iota-hyp-1}-~\ref{fig:iota-hyp-3} are $\iota$-acyclic. 

Below we explain why these hypergraphs belong or do not belong to the class of $\iota$-acyclic hypergraphs. 
We argue by using the characterisation in Theorem~\ref{theorem:iota-acyclicity}. 
Furthermore, we analyse the complexity of the corresponding queries. 
To simplify our complexity analysis, we reduce the number of the EJ queries by dropping {\em singleton variables},
i.e. variables that occur in only one atom in an EJ. Such variables do not 
change the fractional hypertree and submodular widths of an EJ and do not affect the
overall time complexity~\cite{FAQ:PODS:2016,pods/Khamis0S17}. 

\subsubsection{Query in Figure~\ref{fig:non-iota-hyp-1}}

\begin{equation*}
	Q_{1} := \textcolor{c82b366}{R}([A], [B], [C]) \wedge 
			 \textcolor{cb85450}{S}([A], [B], [C]) \wedge 
			 \textcolor{cd6b656}{T}([A], [B], [C]).
\end{equation*}
The hypergraph has a Berge cycle of length $3$: $\textcolor{c82b366}{R}-[A]-\textcolor{cb85450}{S}-[B]-\textcolor{cd6b656}{T}-[C]-\textcolor{c82b366}{R}$.
Applying the reduction produces $3! \cdot 3! \cdot 3! = 216$  
equality join queries. Dropping singleton variables and collapsing
EJ queries that become identical afterwards reduces the total number of EJ queries to
$27$. We further simplify our analysis by grouping the different IJ queries into isomorphic classes. 
Then, we take a representative query from each isomorphism class and 
present its complexity (the complexity is the same for all queries in the same
isomorphism class). We derive the following $3$ isomorphic classes.

\paragraph*{Class 1: $\fhtw = 1.0, \subw = 1.0$}
\begin{equation*}
    \tilde Q_{1}^{(1)} :=
    \tilde R(A_1, B_1, C_1) \wedge
	\tilde S(A_1, B_1, C_1, A_2, B_2, C_2) \wedge
	\tilde T(A_1, B_1, C_1, A_2, B_2, C_2) 
\end{equation*}

The above query has a $\fhtw$ of 1.0 obtained through a tree decomposition consisting of the following bag:
$\{A_1, B_1, C_1, A_2, B_2, C_2\}$ containing relations $\tilde R, \tilde S$ and $\tilde T$
and has a fractional edge cover number $\rho^*$ of 1.0 obtained by assigning the following coefficients: $[0.0, 1.0, 0.0]$. 
	
\paragraph*{Class 2: $\fhtw = 1.0, \subw = 1.0$}
\begin{equation*}
    \tilde Q_{1}^{(2)} :=
    \tilde R(A_1, B_1, C_1, A_2) \wedge
	\tilde S(A_1, B_1, C_1, B_2, C_2) \wedge
	\tilde T(A_1, B_1, C_1, A_2, B_2, C_2) 
\end{equation*}

The above query has a $\fhtw$ of 1.0 obtained through a tree decomposition consisting of the following bag:
$\{A_1, B_1, C_1, A_2, B_2, C_2\}$ containing relations $\tilde R, \tilde S$ and $\tilde T$
and has a fractional edge cover number $\rho^*$ of 1.0 obtained by assigning the following coefficients: $[0.0, 0.0, 1.0]$. 

\paragraph*{Class 3: $\fhtw = 1.5, \subw = 1.5$}

\begin{equation*}
    \tilde Q_{1}^{(3)} :=
    \tilde R(A_1, B_1, C_1, A_2, B_2) \wedge
	\tilde S(A_1, B_1, C_1, A_2, C_2) \wedge
	\tilde T(A_1, B_1, C_1, B_2, C_2) 
\end{equation*}

The above query has a $\fhtw$ of 1.5 obtained through a tree decomposition consisting of the following bag:
$\{A_1, B_1, C_1, A_2, B_2, C_2\}$ containing relations $\tilde R, \tilde S$ and $\tilde T$
and has a fractional edge cover number $\rho^*$ of 1.5 obtained by assigning the following coefficients: $[0.5, 0.5, 0.5]$. 

We have $\ijw(Q_{1}) = 3/2$. Therefore, our approach has takes time $O(N^{3/2} \cdot \polylog N)$.

\subsubsection{Query in Figure~\ref{fig:non-iota-hyp-2}}	
	
	\begin{equation*}
		Q_{2} := \textcolor{c82b366}{R}([A], [B], [C]) \wedge 
				 \textcolor{cb85450}{S}([A], [B], [C]) \wedge
				 \textcolor{cd6b656}{T}([A], [B]).
	\end{equation*}

The hypergraph has a Berge cycle of length $3$: $\textcolor{c82b366}{R}-[A]-\textcolor{cd6b656}{T}-[B]-\textcolor{cb85450}{S}-[C]-\textcolor{c82b366}{R}$.
Applying the reduction produces $3! \cdot 3! \cdot 2! = 72$  
equality join queries. Dropping singleton variables and collapsing
EJ queries that become identical afterwards reduces the total number of EJ queries to
$9$. We further simplify our analysis by grouping the different IJ queries into isomorphic classes. 
Then, we take a representative query from each isomorphism class and 
present its complexity (the complexity is the same for all queries in the same
isomorphism class). We the following $3$ isomorphic classes.
	
\paragraph*{Class 1: $\fhtw = 1.0, \subw = 1.0$}
\begin{equation*}
    \tilde Q_{2}^{(1)} :=
    \tilde R(A_1, B_1, C_1) \wedge
	\tilde S(A_1, B_1, C_1, A_2, B_2) \wedge
	\tilde T(A_1, B_1, A_2, B_2) 
\end{equation*}

The above query has a $\fhtw$ of 1.0 obtained through a tree decomposition consisting of the following bag:
$\{A_1, B_1, C_1, A_2, B_2\}$ containing relations $\tilde R, \tilde S$ and $\tilde T$
and has a fractional edge cover number $\rho^*$ of 1.0 obtained by assigning the following coefficients: $[0.0, 1.0, 0.0]$. 
	
\paragraph*{Class 2: $\fhtw = 1.5, \subw = 1.5$}
\begin{equation*}
    \tilde Q_{2}^{(2)} :=
    \tilde R(A_1, B_1, C_1, A_2) \wedge
	\tilde S(A_1, B_1, C_1, B_2) \wedge
	\tilde T(A_1, B_1, A_2, B_2) 
\end{equation*}

The above query has a $\fhtw$ of 1.5 obtained through a tree decomposition consisting of the following bag:
$\{A_1, B_1, C_1, A_2, B_2\}$ containing relations $\tilde R, \tilde S$ and $\tilde T$
and has a fractional edge cover number $\rho^*$ of 1.5 obtained by assigning the following coefficients: $[0.5, 0.5, 0.5]$. 

\paragraph*{Class 3: $\fhtw = 1.0, \subw = 1.0$}

\begin{equation*}
    \tilde Q_{2}^{(3)} :=
    \tilde R(A_1, B_1, C_1, A_2, B_2) \wedge
	\tilde S(A_1, B_1, C_1, A_2, B_2) \wedge
	\tilde T(A_1, B_1) 
\end{equation*}

The above query has a $\fhtw$ of 1.0 obtained through a tree decomposition consisting of the following bag:
$\{A_1, B_1, C_1, A_2, B_2\}$ containing relations $\tilde R, \tilde S$ and $\tilde T$
and has a fractional edge cover number $\rho^*$ of 1.0 obtained by assigning the following coefficients: $[1.0, 0.0, 0.0]$. 

We have $\ijw(Q_{2}) = 3/2$. Therefore, our approach takes time $O(N^{3/2} \cdot \polylog N)$.

\subsubsection{Query in Figure~\ref{fig:non-iota-hyp-3}}	
	\begin{equation*}
		Q_{3} := \textcolor{c82b366}{R}([A], [B], [C]) \wedge 
				 \textcolor{cb85450}{S}([B], [C]) \wedge 
				 \textcolor{cd6b656}{T}([A], [B]).
	\end{equation*}
The hypergraph has a Berge cycle of length $3$: $\textcolor{c82b366}{R}-[A]-\textcolor{cd6b656}{T}-[B]-\textcolor{cb85450}{S}-[C]-\textcolor{c82b366}{R}$.
Applying the reduction produces $2! \cdot 3! \cdot 2! = 24$ 
equality join queries. Dropping singleton variables and collapsing
EJ queries that become identical afterwards reduces the total number of EJ queries to
$3$. In the following we analyse each of the three cases separately.

\paragraph*{Case 1: $\fhtw = 1.5, \subw = 1.5$}
\begin{equation*}
    \tilde Q_{3}^{(1)} :=
    \tilde R(A_1, B_1, C_1) \wedge
	\tilde S(B_1, C_1, B_2) \wedge
	\tilde T(A_1, B_1, B_2) 
\end{equation*}

The above query has a $\fhtw$ of 1.5 obtained through a tree decomposition consisting of the following bag:
$\{A_1, B_1, C_1, B_2\}$ containing relations $\tilde R, \tilde S$ and $\tilde T$
and has a fractional edge cover number $\rho^*$ of 1.5 obtained by assigning the following coefficients: $[0.5, 0.5, 0.5]$. 
	
\paragraph*{Case 2: $\fhtw = 1.0, \subw = 1.0$}
\begin{equation*}
    \tilde Q_{3}^{(2)} :=
    \tilde R(A_1, B_1, C_1, B_2) \wedge
	\tilde S(B_1, C_1, B_2) \wedge
	\tilde T(A_1, B_1) 
\end{equation*}

The above query has a $\fhtw$ of 1.0 obtained through a tree decomposition consisting of the following bag:
$\{A_1, B_1, C_1, B_2\}$ containing relations $\tilde R, \tilde S$ and $\tilde T$
and has a fractional edge cover number $\rho^*$ of 1.0 obtained by assigning the following coefficients: $[1.0, 0.0, 0.0]$. 

\paragraph*{Case 3: $\fhtw = 1.0, \subw = 1.0$}
\begin{equation*}
    \tilde Q_{3}^{(3)} :=
    \tilde R(A_1, B_1, C_1, B_2) \wedge
	\tilde S(B_1, C_1) \wedge
	\tilde T(A_1, B_1, B_2) 
\end{equation*}

The above query has a $\fhtw$ of 1.0 obtained through a tree decomposition consisting of the following bag:
$\{A_1, B_1, C_1, B_2\}$ containing relations $\tilde R, \tilde S$ and $\tilde T$
and has a fractional edge cover number $\rho^*$ of 1.0 obtained by assigning the following coefficients: $[1.0, 0.0, 0.0]$. 

We have $\ijw(Q_{3}) = 3/2$. Therefore, our approach takes time $O(N^{3/2} \cdot \polylog N)$.

\subsubsection{Query in Figure~\ref{fig:iota-hyp-1}}	
	
	\begin{equation*}
		Q_{4} := \textcolor{c82b366}{R}([A], [B], [C]) \wedge 
				 \textcolor{cb85450}{S}([A], [B], [C]) \wedge 
				 \textcolor{cd6b656}{T}([A]).
	\end{equation*}
	The hypergraph has Berge cycles of length 2 but no Berge cycle of length $\geq 3$. 
	This is because it has three distinct nodes but one of them, namely $[A]$, only belongs to one edge so it cannot be part of a cycle.
	Applying the reduction produces $3! \cdot 2! \cdot 1! = 12$ queries.
	All the queries in the reduction are $\alpha$-acyclic, hence our approach takes time $O(N \cdot \polylog N)$.

\subsubsection{Query in Figure~\ref{fig:iota-hyp-2}}	
	
\begin{equation*}
		Q_{5} := \textcolor{c82b366}{R}([A], [B]) \wedge 
				 \textcolor{cb85450}{S}([A], [C]) \wedge 
				 \textcolor{cd6b656}{T}([C], [D]) \wedge 
				 \textcolor{cd6b656}{T}([C], [E]).
	\end{equation*}
	The hypergraph has no Berge cycle.
	Applying the reduction produces $2! \cdot 1! \cdot 3! \cdot 1! \cdot 1! = 12$ queries.
	All the queries in the reduction are $\alpha$-acyclic, hence our approach takes time $O(N \cdot \polylog N)$.

\subsubsection{Query in Figure~\ref{fig:iota-hyp-3}}	

	\begin{equation*}
		Q_{6} := \textcolor{c82b366}{R}([A], [B], [C]) \wedge 
				 \textcolor{cb85450}{S}([A], [B]).
	\end{equation*}
	The hypergraph has one Berge cycle of length 2 but no Berge cycle of length $\geq 3$.
	Applying the reduction produces $2! \cdot 2! \cdot 1! = 4$ queries.
	All the queries in the reduction are $\alpha$-acyclic, hence our approach takes time $O(N \cdot \polylog N)$.

\begin{figure}
\centering 
\begin{subfigure}{0.27\textwidth}
  \includegraphics[width=\linewidth]{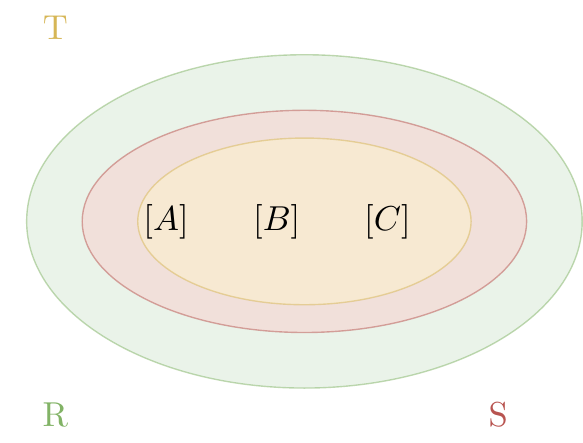}
  \caption{}
  \label{fig:non-iota-hyp-1}
\end{subfigure}\hfil 
\begin{subfigure}{0.27\textwidth}
  \includegraphics[width=\linewidth]{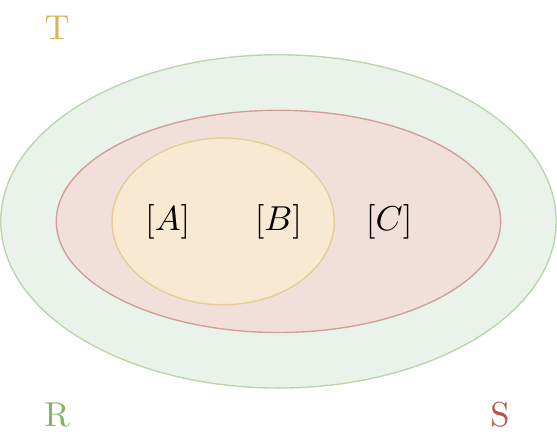}
  \caption{}
  \label{fig:non-iota-hyp-2}
\end{subfigure}\hfil 
\begin{subfigure}{0.27\textwidth}
  \includegraphics[width=\linewidth]{hyp-3.pdf}
  \caption{}
  \label{fig:non-iota-hyp-3}
\end{subfigure}
\begin{subfigure}{0.27\textwidth}
  \includegraphics[width=\linewidth]{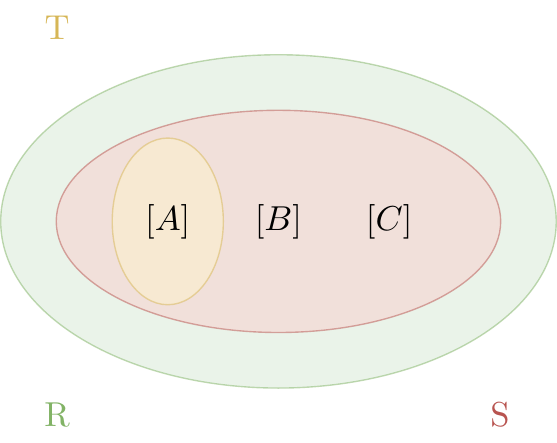}
  \caption{}
  \label{fig:iota-hyp-1}
\end{subfigure}\hfil 
\begin{subfigure}{0.3\textwidth}
  \includegraphics[width=\linewidth]{hyp-5.pdf}
  \caption{}
  \label{fig:iota-hyp-2}
\end{subfigure}\hfil
\begin{subfigure}{0.27\textwidth}
  \includegraphics[width=\linewidth]{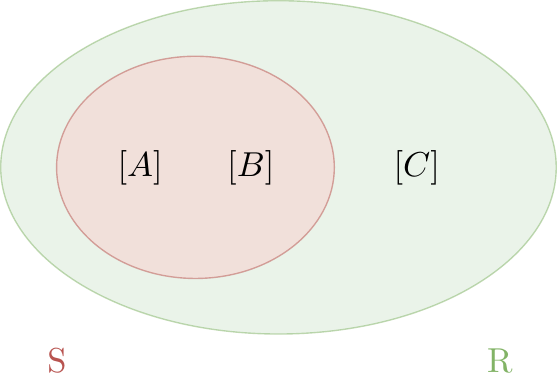}
  \caption{}
  \label{fig:iota-hyp-3}
\end{subfigure}
\caption{Example hypergraphs. Hypergraphs~\ref{fig:non-iota-hyp-1} -~\ref{fig:non-iota-hyp-3}
are $\alpha$-acyclic but not $\iota$-acyclic. Hypergraphs~\ref{fig:iota-hyp-1} -~\ref{fig:iota-hyp-3}
are $\iota$-acyclic.} 
\label{fig:iota-example-hypergraphs}
\Description[]{}
\end{figure}

\section{Complexity Analysis for Three Cyclic \bcqij Queries: Our Approach versus FAQ-AI}
\label{appendix:examples}

We next exemplify the upper bound on the time complexity obtained by our approach versus FAQ-AI~\cite{FAQAI:TODS:2020} for three cyclic queries with intersection joins: the triangle query, the Loomis Whitney query with four variables, and the 4-clique query.
See Table~\ref{tab:examples} for a summary of the comparison.

\begin{table}[th]
    \begin{tabular}{|c|c|c|}
        \hline
        \bcqij Query &
        $\faqai$ approach~\cite{FAQAI:TODS:2020} &
        Our approach\\
        \hline
        Triangle query &
        $O(N^2 \log^3 N)$ &
        $O(N^{3/2} \log^3 N$)\\
        Loomis-Whitney query 4 &
        $O(N^2 \log^k N)$, for $k \geq 9$ &
        $O(N^{5/3} \log^8 N)$ \\
        4-clique &
        $O(N^3 \log^k N)$, for $k \geq 5$ &
        $O(N^2 \log^8 N)$\\
        \hline
    \end{tabular}
    \caption{Comparison of the runtimes of our approach versus the $\faqai$ approach
    on three intersection join queries. See Appendix~\ref{appendix:examples} for more details.}
    \label{tab:examples}
\end{table}

\subsection{The triangle intersection join query}
\label{app:ex:triangle}

Consider the triangle intersection join query from Section~\ref{sec:example}:
\[
	Q_{\triangle} = R([A],[B]) \wedge S([B], [C]) \wedge T([A], [C])
\]
While we showed in Section~\ref{sec:example} that our approach solves this query in time $O(N^{3/2}\log^3 N)$,
we show here that the $\faqai$ approach~\cite{FAQAI:TODS:2020} needs time $O(N^2\log^3 N)$ for this query.

To apply the $\faqai$ approach to query $Q_\triangle$ above, we have to express it as
a query with inequality joins. In particular, each interval variable (say $[A]$) is
going to be replaced by two scalar variables ($A.l$ and $A.r$) representing the beginning and end
of $[A]$, i.e. $[A] = [A.l, A.r]$.
Let $A.l^{(R)}$ and $A.r^{(R)}$ be the beginning and end of interval $[A]$ in relation $R$.
Similarly, we define $A.l^{(T)}$ and $A.r^{(T)}$.
For the intervals $[A.l^{(R)}, A.r^{(R)}]$ and $[A.l^{(T)}, A.r^{(T)}]$
to overlap, the following condition must hold:
\begin{eqnarray}
    \left(A.l^{(T)} \leq A.l^{(R)} \leq A.r^{(T)}\right) \vee
    \left(A.l^{(R)} \leq A.l^{(T)} \leq A.r^{(R)}\right)\label{eq:triangle:faqai:intersect}
\end{eqnarray}
For each variable $X\in \{A, B, C\}$, let $F(X)$ denote the set of relations containing $X$, i.e.
\begin{eqnarray*}
    F(A) &:=& \{T, R\},\\
    F(B) &:=& \{R, S\},\\
    F(C) &:=& \{S, T\}.
\end{eqnarray*}
Note that~\eqref{eq:triangle:faqai:intersect} can be written equivalently as:
\begin{eqnarray}
    \bigvee_{\substack{V \in F(A)\\W \in F(A)-\{V\}}} A.l^{(W)} \leq A.l^{(V)} \leq A.r^{(W)}
    \label{eq:triangle:faqai:intersect2}
\end{eqnarray}
The same applies to the other two interval variables $[B]$ and $[C]$.
After distributing disjunctions over conjunctions, $Q_\triangle$ can be written as follows:
\begin{eqnarray}
    Q_\triangle =
    \bigvee_{\substack{
        (V_A, V_B, V_C)\in F(A)\times F(B)\times F(C)\\
        W_A \in F(A) -\{V_A\}\\
        W_B \in F(B) -\{V_B\}\\
        W_C \in F(C) -\{V_C\}
    }}
    &R\left(A.l^{(R)}, A.r^{(R)}, B.l^{(R)}, B.r^{(R)}\right) \wedge
     S\left(B.l^{(S)}, B.r^{(S)}, C.l^{(S)}, C.r^{(S)}\right) \wedge
     T\left(A.l^{(T)}, A.r^{(T)}, C.l^{(T)}, C.r^{(T)}\right) \wedge\nonumber\\
    &\left(A.l^{(W_A)} \leq A.l^{(V_A)} \leq A.r^{(W_A)}\right)\wedge
     \left(B.l^{(W_B)} \leq B.l^{(V_B)} \leq B.r^{(W_B)}\right)\wedge
     \left(C.l^{(W_C)} \leq C.l^{(V_C)} \leq C.r^{(W_C)}\right)\nonumber\\
    \label{eq:triangle:faqai}
\end{eqnarray}
For each $(V_A, V_B, V_C)\in F(A)\times F(B)\times F(C)$,
the inner conjunction in~\eqref{eq:triangle:faqai} is an $\faqai$ query~\cite{FAQAI:TODS:2020}.
While solving each such query, it is possible to relax the definition of tree decompositions
thus extending the set of valid tree decompositions and potentially reducing the
fractional hypertree and submodular widths, ultimately resulting in the {\em relaxed} versions of these widths
$\fhtw_\ell$ and $\subw_\ell$ respectively~\cite{FAQAI:TODS:2020}.
In particular, in a relaxed tree decomposition, we no longer require each inequality to have
its variables contained in one bag of the tree. Instead, it suffices to have its variables
contained in two adjacent bags in the tree.

Fix an arbitrary $(V_A, V_B, V_C)\in F(A)\times F(B)\times F(C)$ and let $\bar Q$
be the resulting $\faqai$ query corresponding to the inner conjunction in~\eqref{eq:triangle:faqai}.
Note that for every pair of the relations $R, S$ and $T$,
the query $\bar Q$ contains at least one inequality between two variables from that pair.
Hence if we distribute the relations $R, S$ and $T$ among three or more bags,
there will be an inequality between two non-adjacent bags thus violating the condition for a
relaxed tree decomposition. Therefore, every relaxed tree decomposition of $\bar Q$ must have at most
two bags where each one the relations $R, S$ and $T$ falls within one bag.
Consequently, there will be one bag with (at least) two relations.
Noting that the variables of relations $R, S$ and $T$ are pairwise disjoint,
this implies that $\fhtw_\ell(\bar Q) \geq 2$.
To minimize $\fhtw_\ell(\bar Q)$, an optimal tree decomposition would have two bags with
two relations in one bag and the third relation in the other, thus resulting in $\fhtw_\ell(\bar Q) = 2$.

The relaxed submodular width $\subw_\ell(\bar Q)$ is not any better in this case. In particular,
consider the following function $\bar h:2^{\vars(\bar Q)}\to \R^+$:
\begin{equation}
    \bar h(X) := \frac{|X|}{4}, \quad\forall X \subseteq \vars(\bar Q).
    \label{eq:triangle:faqai:barh}
\end{equation}
Recall notation from Section~\ref{subsection:widths} and~\cite{FAQAI:TODS:2020, jacm/Marx13}.
The above $\bar h$ is a modular function hence it is submodular. Since it is also monotone,
$\bar h$ is a polymatroid, i.e. $\bar h \in \Gamma_{\vars(\bar Q)}$ where
$\Gamma_{\vars(\bar Q)}$ denotes the set of polymatroids over the variables $\vars(\bar Q)$~\cite{FAQAI:TODS:2020}.
Moreover for each finite input relation $E\in \{R, S, T\}$,
we have $\bar h(E) = 1$ since each one of these relations has four variables, i.e. $|E| = 4$.
Therefore $\bar h$ is edge dominated, i.e. $\bar h \in \ed(\bar Q)$ where $\ed(\bar Q)$
denotes the set of edge dominated functions $h:2^{\vars(\bar Q)}\to\R^+$.
Recall the definition of $\subw_\ell(Q)$ for an $\faqai$ query $Q$ from~\cite{FAQAI:TODS:2020} where $\td_\ell(Q)$ denotes
the set of relaxed tree decompositions of $Q$:
\begin{equation}
\subw_\ell(Q) :=
    \max_{h \in \ed(Q) \cap \Gamma_{\vars(Q)}}
    \min_{(\calT, \chi) \in \td_\ell(Q)}
    \max_{t \in V(\calT)} h(\chi(t)).
    \label{eq:defn:subw:ell}
\end{equation}
Based on the above definition and by choosing $\bar h \in \ed(\bar Q) \cap \Gamma_{\vars(\bar Q)}$, we have
\[
    \subw_\ell(\bar Q) \geq
    \min_{(\calT, \chi) \in \td_\ell(\bar Q)}
    \max_{t \in V(\calT)} \bar h(\chi(t)).
\]
However for each relaxed tree decomposition $(\calT, \chi)\in \td_\ell(\bar Q)$,
we argued before that there must exist some bag $t^*\in V(\calT)$ containing at least
two of the input relations $\{R, S, T\}$ hence at least 8 distinct variables,
meaning that $|\chi(t^*)| \geq 8$.
From~\eqref{eq:triangle:faqai:barh}, we have $\bar h(\chi(t^*))\geq 2$ which implies that
$\subw_\ell(\bar Q) \geq 2$.
And since $\subw_\ell(Q) \leq \fhtw_\ell(Q)$ for any query $Q$ according to~\cite{FAQAI:TODS:2020},
we have
\begin{equation}
    \subw_\ell(\bar Q) = \fhtw_\ell(\bar Q) = 2.
\end{equation}

Finally according to Theorem 3.5 in~\cite{FAQAI:TODS:2020}, the time complexity in $\faqai$
involves an extra factor of $(\log N)^{\max(k-1,1)}$ where $k$ is the number of inequalities
that involve variables from two adjacent bags (i.e. that are not contained in a single bag) in an optimal relaxed tree decomposition.
In query~\eqref{eq:triangle:faqai}, when constructing any optimal relaxed tree decomposition
involving two relations in one bag (say $R$ and $S$) and the third relation in another bag,
there will be exactly $4$ inequalities involving variables from both bags. Hence $k = 4$
and the overall time complexity of $\faqai$ for $Q_\triangle$ is $O(N^2 \log^3 N)$.

\subsection{The Loomis-Whitney intersection join query with 4 variables}
\label{app:ex:LW4}

The Loomis-Whitney intersection join query with 4-variables (LW4) is as follows:
\begin{equation}
    Q_{\mathrm{LW4}} = R([A], [B], [C]) \wedge S([B], [C], [D]) \wedge T([C], [D], [A])
        \wedge U([D], [A], [B]),
    \label{eq:LW4}
\end{equation}
where each one of the variables $[A], [B], [C]$ and $[D]$ above is an interval variable.
The $\faqai$ approach~\cite{FAQAI:TODS:2020} cannot solve this query in time
better than $O(N^2 \log^9 N)$.
However, the reduction from this work can be used to solve this query in time
$O(N^{5/3} \log^8 N)$. Below we apply both the $\faqai$ approach and the one from this
work.



\subsubsection{The $\faqai$ approach~\cite{FAQAI:TODS:2020} takes time $O(N^2 \log^k N)$ for some $k\geq 9$}

Similar to Section~\ref{app:ex:triangle}, to apply the $\faqai$ approach to query~\eqref{eq:LW4} above, we formulate it as a query with inequality joins.
Specifically we replace each interval variable $[A]$ with two scalar variables $A.l$ and $A.r$ representing the beginning and end
of interval $[A]$. Furthermore, we use $A.l^{(R)}$ and $A.r^{(R)}$ to refer to the beginning
and end of interval $[A]$ in relation $R$, and similarly
we use $A.l^{(T)}, A.r^{(T)}, A.l^{(U)}$ and $A.r^{(U)}$ to refer to corresponding interval
boundaries in relations $T$ and $U$.
The three intervals $[A.l^{(R)}, A.r^{(R)}], [A.l^{(T)}, A.r^{(T)}]$ and $[A.l^{(U)}, A.r^{(U)}]$
overlap if and only if the following condition is met:
\begin{eqnarray}
    \left(A.l^{(T)} \leq A.l^{(R)} \leq A.r^{(T)}\right) &\wedge&
    \left(A.l^{(U)} \leq A.l^{(R)} \leq A.r^{(U)}\right) \quad\vee\nonumber\\
    \left(A.l^{(R)} \leq A.l^{(T)} \leq A.r^{(R)}\right) &\wedge&
    \left(A.l^{(U)} \leq A.l^{(T)} \leq A.r^{(U)}\right) \quad\vee\nonumber\\
    \left(A.l^{(R)} \leq A.l^{(U)} \leq A.r^{(R)}\right) &\wedge&
    \left(A.l^{(T)} \leq A.l^{(U)} \leq A.r^{(T)}\right) \label{eq:lw4:faqai:intersect}
\end{eqnarray}
Given a variable $X\in \{A, B, C, D\}$, let $F(X)$ denote the set of relations containing $X$, i.e.
\begin{eqnarray*}
    F(A) &:=& \{T, U, R\},\\
    F(B) &:=& \{U, R, S\},\\
    F(C) &:=& \{R, S, T\},\\
    F(D) &:=& \{S, T, U\}.
\end{eqnarray*}
Condition~\eqref{eq:lw4:faqai:intersect} can be formulated as follows:
\begin{eqnarray}
    \bigvee_{V \in F(A)} \bigwedge_{W \in F(A)-\{V\}} A.l^{(W)} \leq A.l^{(V)} \leq A.r^{(W)}
    \label{eq:lw4:faqai:intersect2}
\end{eqnarray}

The same applies to the other three interval variables $[B], [C]$ and $[D]$.
By distributing disjunctions over conjunctions, we rewrite query~\eqref{eq:LW4} as follows:
\begin{eqnarray}
    Q_{\mathrm{LW4}} =
    \bigvee_{(V_A, V_B, V_C, V_D)\in F(A)\times F(B)\times F(C)\times F(D)}
    &R\left(A.l^{(R)}, A.r^{(R)}, B.l^{(R)}, B.r^{(R)}, C.l^{(R)}, C.r^{(R)}\right) \wedge\nonumber\\
    &S\left(B.l^{(S)}, B.r^{(S)}, C.l^{(S)}, C.r^{(S)}, D.l^{(S)}, D.r^{(S)}\right) \wedge\nonumber\\
    &T\left(C.l^{(T)}, C.r^{(T)}, D.l^{(T)}, D.r^{(T)}, A.l^{(T)}, A.r^{(T)}\right) \wedge\nonumber\\
    &U\left(D.l^{(U)}, D.r^{(U)}, A.l^{(U)}, A.r^{(U)}, B.l^{(U)}, B.r^{(U)}\right) \wedge\nonumber\\
    &\displaystyle{\bigwedge_{W_A \in F(A)-\{V_A\}}} A.l^{(W_A)} \leq A.l^{(V_A)} \leq A.r^{(W_A)}\wedge\nonumber\\
    &\displaystyle{\bigwedge_{W_B \in F(B)-\{V_B\}}} B.l^{(W_B)} \leq B.l^{(V_B)} \leq B.r^{(W_B)}\wedge\nonumber\\
    &\displaystyle{\bigwedge_{W_C \in F(C)-\{V_C\}}} C.l^{(W_C)} \leq C.l^{(V_C)} \leq C.r^{(W_C)}\wedge\nonumber\\
    &\displaystyle{\bigwedge_{W_D \in F(D)-\{V_D\}}} D.l^{(W_D)} \leq D.l^{(V_D)} \leq D.r^{(W_D)}
    \label{eq:LW4:faqai}
\end{eqnarray}

For each $(V_A, V_B, V_C, V_D)\in F(A)\times F(B)\times F(C)\times F(D)$,
the inner conjunction in~\eqref{eq:LW4:faqai} is an $\faqai$ query~\cite{FAQAI:TODS:2020}.
While solving each such query, we can use {\em relaxed} tree decompositions~\cite{FAQAI:TODS:2020},
in a similar way to what we did in Section~\ref{app:ex:triangle}.

The following is one $\faqai$ query from~\eqref{eq:LW4:faqai} obtained by choosing
$V_A = U, V_B = S, V_C = T, V_D = T$:
\begin{eqnarray*}
    \bar Q =
    &R\left(A.l^{(R)}, A.r^{(R)}, B.l^{(R)}, B.r^{(R)}, C.l^{(R)}, C.r^{(R)}\right) \wedge
    S\left(B.l^{(S)}, B.r^{(S)}, C.l^{(S)}, C.r^{(S)}, D.l^{(S)}, D.r^{(S)}\right) \wedge\\
    &T\left(C.l^{(T)}, C.r^{(T)}, D.l^{(T)}, D.r^{(T)}, A.l^{(T)}, A.r^{(T)}\right) \wedge
    U\left(D.l^{(U)}, D.r^{(U)}, A.l^{(U)}, A.r^{(U)}, B.l^{(U)}, B.r^{(U)}\right) \wedge\\
    &\left(A.l^{(R)} \leq A.l^{(U)} \leq A.r^{(R)}\right) \wedge
    \left(A.l^{(T)} \leq A.l^{(U)} \leq A.r^{(T)}\right)\wedge\\
    &\left(B.l^{(R)} \leq B.l^{(S)} \leq B.r^{(R)}\right) \wedge
    \left(B.l^{(U)} \leq B.l^{(S)} \leq B.r^{(U)}\right)\wedge\\
    &\left(C.l^{(R)} \leq C.l^{(T)} \leq C.r^{(R)}\right) \wedge
    \left(C.l^{(S)} \leq C.l^{(T)} \leq C.r^{(S)}\right)\wedge\\
    &\left(D.l^{(S)} \leq D.l^{(T)} \leq D.r^{(S)}\right) \wedge
    \left(D.l^{(U)} \leq D.l^{(T)} \leq D.r^{(U)}\right)
\end{eqnarray*}
Note that in the above query $\bar Q$,
there exists at least one inequality between two variables from every pair of the relations $R, S, T$ and $U$.
Therefore if we were to divide the relations $R, S, T$ and $U$ among three or more bags,
there will be at least one inequality between two non-adjacent bags thus violating the definition of a
relaxed tree decomposition. As a result, in every relaxed tree decomposition of $\bar Q$, there
can be at most
two bags where each one the relations $R, S, T$ and $U$ falls within one bag.
Consequently, there will be at least one bag with (at least) two relations.
Noting that the variables of relations $R, S, T$ and $U$ are pairwise disjoint,
this implies that $\fhtw_\ell(\bar Q) \geq 2$.
To minimize $\fhtw_\ell(\bar Q)$, an optimal tree decomposition would have two bags with exactly
two relations in each bag, thus resulting in $\fhtw_\ell(\bar Q) = 2$.

In this case, the relaxed submodular width $\subw_\ell(\bar Q)$ is identical $\fhtw_\ell(\bar Q)$.
To show this, we use the following function $\bar h:2^{\vars(\bar Q)}\to \R^+$ in
a similar way to what we did in Section~\ref{app:ex:triangle}:
\begin{equation}
    \bar h(X) := \frac{|X|}{6}, \quad\forall X \subseteq \vars(\bar Q).
    \label{eq:faqai:barh}
\end{equation}
(Recall notation from Section~\ref{subsection:widths} and~\cite{FAQAI:TODS:2020, jacm/Marx13}.)
The above $\bar h$ is modular hence submodular. Because it is also monotone,
$\bar h$ is a polymatroid, i.e. $\bar h \in \Gamma_{\vars(\bar Q)}$ where
$\Gamma_{\vars(\bar Q)}$ denotes the set of polymatroids over the variables $\vars(\bar Q)$~\cite{FAQAI:TODS:2020}.
Moreover for each finite input relation $E\in \{R, S, T, U\}$,
we have $\bar h(E) = 1$ since each one of these relations has six variables, i.e. $|E| = 6$.
Therefore $\bar h$ is edge dominated, i.e. $\bar h \in \ed(\bar Q)$.
Based on the definition of $\subw_\ell$ from~\eqref{eq:defn:subw:ell} and by choosing $\bar h \in \ed(\bar Q) \cap \Gamma_{\vars(\bar Q)}$, we have
\[
    \subw_\ell(\bar Q) \geq
    \min_{(\calT, \chi) \in \td_\ell(\bar Q)}
    \max_{t \in V(\calT)} \bar h(\chi(t)),
\]
where $\td_\ell(\bar Q)$ denotes the set of relaxed tree decompositions of $\bar Q$.
However for each relaxed tree decomposition $(\calT, \chi)\in \td_\ell(\bar Q)$,
we argued before that there must exist some bag $t^*\in V(\calT)$ containing at least
two of the input relations $\{R, S, T, U\}$ hence at least 12 distinct variables,
meaning that $|\chi(t^*)| \geq 12$.
From~\eqref{eq:faqai:barh}, we have $\bar h(\chi(t^*))\geq 2$ which implies that
$\subw_\ell(\bar Q) \geq 2$.
And since $\subw_\ell(Q) \leq \fhtw_\ell(Q)$ for any query $Q$ according to~\cite{FAQAI:TODS:2020},
we have
\begin{equation}
    \subw_\ell(\bar Q) = \fhtw_\ell(\bar Q) = 2.
\end{equation}
As mentioned in Section~\ref{app:ex:triangle}, the runtime complexity in $\faqai$ involves
an extra factor of $(\log N)^{\max(k-1,1)}$ where $k$ is the number of inequalities
involving variables from two adjacent bags in an optimal relaxed tree decomposition.
In query $\bar Q$ above, the minimum value of $k$ over all optimal relaxed tree decompositions
is 10. Hence the $\faqai$ time complexity for $\bar Q$ is $O(N^2 \log^9 N)$.

Finally note that for every other choice of
$(V_A, V_B, V_C, V_D)\in F(A)\times F(B)\times F(C)\times F(D)$, the resulting $\faqai$
 query $\tilde Q$ corresponding to the inner conjunction in~\eqref{eq:LW4:faqai} must satisfy
 \begin{eqnarray*}
    \fhtw_\ell(\tilde Q) \leq \fhtw_\ell(\bar Q),\\
    \subw_\ell(\tilde Q) \leq \subw_\ell(\bar Q).
 \end{eqnarray*}
 This is because $\td_\ell(\tilde Q) \supseteq \td_\ell(\bar Q)$ since $\bar Q$ already contains
 at least one equality involving every pair of the relations $R, S, T$ and $U$.

\subsubsection{Our approach takes time $O(N^{5/3} \log^8 N)$.}

Applying the reduction from this work to query~\eqref{eq:LW4} produces a large number of equality
join queries. We can reduce the number of these EJs by dropping {\em singleton variables},
that is variables that occur in only one atom in an EJ. Such variables don't
change the fractional hypertree and submodular widths of an EJ and don't affect the
overall time complexity~\cite{FAQ:PODS:2016,pods/Khamis0S17}. Dropping singleton variables and collapsing
EJ queries that become identical afterwards reduces the total number of EJ queries down to
81, which is still big.

Luckily many of these 81 queries are isomorphic to one another. Aided by a computer program
to analyze them, it turns out that they can be grouped into only 6 isomorphism classes.
Below we take a representative query from each isomorphism class and compute its
fractional hypertree and submodular widths (which are the same for all queries in the same
isomorphism class).

\paragraph*{Class 1: $\fhtw = 2, \subw = 1.5$}
Queries in the first class are isomorphic to the following. (We use $\tilde R, \tilde S,
\ldots$ to denote relations corresponding to $R, S, \ldots$ {\em after} the reduction.)
\begin{equation}
    Q_{\mathrm{LW4}}^{(1)} :=
    \tilde R(A_1, B_1, C_1, B_2, C_2) \wedge
	\tilde S(B_1, C_1, D_1, C_2, D_2) \wedge
	\tilde T(C_1, D_1, A_1, D_2, A_2) \wedge
	\tilde U(D_1, A_1, B_1, A_2, B_2)
    \label{eq:LW4:C1}
\end{equation}
The $\fhtw$ of the above query is 2, which is higher than our final target of $5/3$ that is
needed to achieve the runtime of $O(N^{5/3} \log^8 N)$.
Therefore we skip how to compute $\fhtw$ for this query.

Luckily, the $\subw$ turns out to be 1.5, just like the 4-cycle query~\cite{pods/Khamis0S17, AYZ97}.
And in fact, there is a corresponding algorithm to answer this query in time
$O(N^{1.5} \log N)$, which is very similar in nature to the algorithm for solving the
4-cycle query in the same time complexity~\cite{pods/Khamis0S17, AYZ97}. We will skip showing the
computation of the $\subw$ itself and directly show the corresponding algorithm
solving~\eqref{eq:LW4:C1} in the desired time.

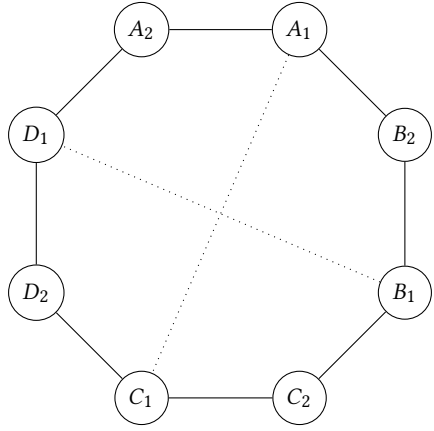
\begin{figure}
\begin{tikzpicture}[scale=0.7, every node/.style={scale = 1}]
    \pgfmathsetmacro{\U}{1.5}
    \pgfmathsetmacro{\V}{3.5}
    \node[circle,draw] at (\U, \V) (A1) {$A_1$};
    \node[circle,draw] at (\V, \U) (B2) {$B_2$};
    \node[circle,draw] at (\V, -\U) (B1) {$B_1$};
    \node[circle,draw] at (\U, -\V) (C2) {$C_2$};
    \node[circle,draw] at (-\U, -\V) (C1) {$C_1$};
    \node[circle,draw] at (-\V, -\U) (D2) {$D_2$};
    \node[circle,draw] at (-\V, \U) (D1) {$D_1$};
    \node[circle,draw] at (-\U, \V) (A2) {$A_2$};
    \path[] (A1) edge (B2);
    \path[] (B2) edge (B1);
    \path[] (B1) edge (C2);
    \path[] (C2) edge (C1);
    \path[] (C1) edge (D2);
    \path[] (D2) edge (D1);
    \path[] (D1) edge (A2);
    \path[] (A2) edge (A1);
    \path[dotted] (A1) edge (C1);
    \path[dotted] (B1) edge (D1);
\end{tikzpicture}
\caption{A visual representation of query $Q_{\mathrm{LW4}}^{(1)}$ from ~\eqref{eq:LW4:C1}. Note that
each input relation in $Q_{\mathrm{LW4}}^{(1)}$ spans 5 consecutive variables on the above cycle.}
\Description[A visual representation of query $Q_{\mathrm{LW4}}^{(1)}$ from ~\eqref{eq:LW4:C1}.]{A visual representation of query $Q_{\mathrm{LW4}}^{(1)}$ from ~\eqref{eq:LW4:C1}. Note that
each input relation in $Q_{\mathrm{LW4}}^{(1)}$ spans 5 consecutive variables on the above cycle.}
\label{fig:LW4-cycle}
\end{figure}

In order to mimic the algorithm for a 4-cycle~\cite{pods/Khamis0S17, AYZ97},
it is helpful to imagine the 8 variables $A_1, A_2, B_1, B_2, C_1, C_2, D_1$ and $D_2$
arranged on a cycle as shown in Figure~\ref{fig:LW4-cycle}. Note that each one of
the 4 input relations $\tilde R, \tilde S, \tilde T$ and $\tilde U$ of $Q_{\mathrm{LW4}}^{(1)}$ spans 5 consecutive variables on the above
cycle.

Let $N$ be the maximum relation size among relations $\tilde R, \tilde S, \tilde T$ and $\tilde U$.
We partition the relation $\tilde R(A_1, B_1, C_1, B_2, C_2)$ based on the {\em degree} of
$(a_1, b_1, b_2)$, i.e. based on the number of different $(c_1, c_2)$ pairs for every given
triple $(a_1, b_1, b_2)$:
\begin{eqnarray}
    \deg_{\tilde R}(a_1, b_1, b_2) &:=& |\{(c_1, c_2) \suchthat (a_1, b_1, c_1, b_2, c_2)\in \tilde R\}|,\\
    \tilde R_h &:=& \{(a_1, b_1, b_2)\suchthat \deg_R(a_1, b_1, b_2) \geq \sqrt{N}\},\label{eq:LW4:Rh}\\
    \tilde R_\ell &:=& \{(a_1, b_1, c_1, b_2, c_2) \in \tilde R \suchthat \deg_{\tilde R}(a_1, b_1, b_2) < \sqrt{N}\}.
\end{eqnarray}
$\tilde R_h$ and $\tilde R_\ell$ above are meant to be the ``heavy'' and ``light'' parts of $\tilde R$ respectively
in the same sense as in~\cite{pods/Khamis0S17, AYZ97}. Note that \eqref{eq:LW4:Rh} implies
that $|\tilde R_h| \leq \sqrt{N}$. Similarly we partition relation $\tilde T(C_1, D_1, A_1, D_2, A_2)$
based on the degree of $(c_1, d_1, d_2)$ into $\tilde T_h(C_1, D_1, D_2)$ and $\tilde T_\ell(C_1, D_1, A_1, D_2, A_2)$.

In order to evaluate $Q_{\mathrm{LW4}}^{(1)}$, we will divide its output tuples
 into three parts and use a different evaluation strategy
to evaluate each part. In particular, each output tuple $(a_1, b_1, c_1, d_1, a_2, b_2, c_2, d_2) \in Q_{\mathrm{LW4}}^{(1)}$
belongs to exactly one of the following three categories:
\begin{itemize}
    \item Category (1): $(a_1, b_1, b_2) \in \tilde R_h$.
    \item Category (2): $(a_1, b_1, b_2) \not\in \tilde R_h$ and $(c_1, d_1, d_2) \in \tilde T_h$.
    \item Category (3): $(a_1, b_1, b_2) \not\in \tilde R_h$ and $(c_1, d_1, d_2) \not\in \tilde T_h$.
\end{itemize}

In order to set up our evaluation strategies for each one of the above three categories,
we compute the following helper relations:
\begin{eqnarray*}
    W_{11}(A_1, B_2, B_1, C_2, C_1, D_2, D_1) &:=& \tilde R_h(A_1, B_1, B_2) \Join \tilde S(B_1, C_1, D_1, C_2, D_2),\\
    W_{12}(C_1, D_2, D_1, A_2, A_1, B_2, B_1) &:=& \tilde R_h(A_1, B_1, B_2) \Join \tilde T(C_1, D_1, A_1, D_2, A_2),\\
    W_{21}(C_1, D_2, D_1, A_2, A_1, B_2, B_1) &:=& \tilde T_h(C_1, D_1, D_2) \Join \tilde U(D_1, A_1, B_1, A_2, B_2),\\
    W_{22}(A_1, B_2, B_1, C_2, C_1, D_2, D_1) &:=& \tilde T_h(C_1, D_1, D_2) \Join \tilde R_\ell(A_1, B_1, C_1, B_2, C_2),\\
    W_{31}(D_1, A_2, A_1, B_2, B_1, C_2, C_1) &:=& \tilde U(D_1, A_1, B_1, A_2, B_2) \Join \tilde R_\ell(A_1, B_1, C_1, B_2, C_2),\\
    W_{32}(B_1, C_2, C_1, D_2, D_1, A_2, A_1) &:=& \tilde S(B_1, C_1, D_1, C_2, D_2) \Join \tilde T_\ell(C_1, D_1, A_1, D_2, A_2).
\end{eqnarray*}

Each one of the above six relations can be straightforwardly shown to have size upper bounded by
$N^{1.5}$ based on the definitions of $\tilde R_h, \tilde R_\ell, \tilde T_h$ and $\tilde T_\ell$.
Finally we show how to compute output tuples belonging to each one of the three categories
above in the desired runtime of $O(N^{1.5} \log N)$:

\begin{itemize}
    \item Category (1): $(a_1, b_1, b_2)\in \tilde R_h$. We can produce output tuples
    $(a_1, b_1, c_1, d_1, a_2, b_2, c_2, d_2)$ belonging to this category
    by running Yannakakis algorithm~\cite{vldb/Yannakakis81} over a tree decomposition whose bags
    are $W'_{11}$ and $W'_{12}$ defined below:
    \begin{eqnarray*}
        W'_{11}(A_1, B_2, B_1, C_2, C_1, D_2, D_1) &:=&
        W_{11}(A_1, B_2, B_1, C_2, C_1, D_2, D_1) \Join \tilde R(A_1, B_1, C_1, B_2, C_2),\\
        W'_{12}(C_1, D_2, D_1, A_2, A_1, B_2, B_1) &:=&
        W_{12}(C_1, D_2, D_1, A_2, A_1, B_2, B_1) \Join \tilde U(D_1, A_1, B_1, A_2, B_2).
    \end{eqnarray*}

    \item Category (2): $(a_1, b_1, b_2)\not\in \tilde R_h$ and $(c_1, d_1, d_2) \in \tilde T_h$. We produce
    these output tuples using a tree decomposition whose bags are $W_{21}'$ and $W_{22}'$ defined below:
    \begin{eqnarray*}
        W'_{21}(C_1, D_2, D_1, A_2, A_1, B_2, B_1) &:=&
        W_{21}(C_1, D_2, D_1, A_2, A_1, B_2, B_1) \Join \tilde T(C_1, D_1, A_1, D_2, A_2),\\
        W'_{22}(A_1, B_2, B_1, C_2, C_1, D_2, D_1) &:=&
        W_{22}(A_1, B_2, B_1, C_2, C_1, D_2, D_1) \Join \tilde S(B_1, C_1, D_1, C_2, D_2).
    \end{eqnarray*}

    \item Category (3): $(a_1, b_1, b_2)\not\in \tilde R_h$ and $(c_1, d_1, d_2) \not\in \tilde T_h$. For this
    category, we use a tree decomposition whose bags are $W_{31}$ and $W_{32}$.
\end{itemize}
In each one of the three cases above, the runtime is $O(N^{1.5} \log N)$.

\paragraph*{Class 2: $\fhtw = \subw = 5/3$}

Queries in this class are isomorphic to the following:
\begin{equation}
    Q_{\mathrm{LW4}}^{(2)} :=
    \tilde R(A_1, B_1, C_1, A_2) \wedge
	\tilde S(B_1, C_1, D_1, B_2, C_2) \wedge
	\tilde T(C_1, D_1, A_1, C_2, D_2) \wedge
	\tilde U(D_1, A_1, B_1, D_2, A_2, B_2)
\end{equation}
The above query has $\fhtw = 5/3$. In particular, it accepts a tree decomposition consisting
of the following two bags:
\begin{itemize}
    \item Bag $\{A_1, B_1, C_1, D_1, A_2, B_2, D_2\}$ containing relations $\tilde R$ and $\tilde U$.
    It has a fractional edge cover number $\rho^*$ of $5/3$,
    which is obtained by assigning the following coefficients to relations $[\tilde R, \tilde S, \tilde T, \tilde U]$ in order: $[1/3, 1/3, 1/3, 2/3]$.
    In particular, it can be computed by solving the following query using a worst-case optimal join
    algorithm~\cite{jacm/NgoPRR18,Veldhuizen14}:
    \begin{eqnarray*}
        W_1(A_1, B_1, C_1, D_1, A_2, B_2, D_2) :=&
        \tilde R(A_1, B_1, C_1, A_2) \Join
        \pi_{\{B_1, C_1, D_1, B_2\}}\tilde S(B_1, C_1, D_1, B_2, C_2) \Join\\
	    &\pi_{\{C_1, D_1, A_1, D_2\}}\tilde T(C_1, D_1, A_1, C_2, D_2) \Join
	    \tilde U(D_1, A_1, B_1, D_2, A_2, B_2)
    \end{eqnarray*}
    \item Bag $\{A_1, B_1, C_1, D_1, B_2, C_2, D_2\}$ containing relations $\tilde S$ and $\tilde T$.
    It has a $\rho^*$ value of 1.5,
    obtained by assigning the following coefficients to relations $[\tilde R, \tilde S, \tilde T, \tilde U]$ in order:
    [0, 1/2, 1/2, 1/2].
    It can be computed by solving the following query\footnote{We could drop relation $\tilde R$ from query
    $W_2$ without increasing its time complexity beyond $N^{1.5}$ because $\tilde R$ has a coefficient of 0 in the
    optimal fractional edge cover.}:
    \begin{eqnarray*}
        W_2(A_1, B_1, C_1, D_1, B_2, C_2, D_2) :=&
        \pi_{\{A_1, B_1, C_1\}}\tilde R(A_1, B_1, C_1, A_2) \Join
        \tilde S(B_1, C_1, D_1, B_2, C_2) \Join\\
        &\tilde T(C_1, D_1, A_1, C_2, D_2) \Join
        \pi_{\{D_1, A_1, B_1, D_2, B_2\}}\tilde U(D_1, A_1, B_1, D_2, A_2, B_2)
    \end{eqnarray*}
\end{itemize}

The $\subw$ of this query class is also 5/3, hence we skip its computation.
This query class {\em is} the bottleneck of our final bound of 5/3 for the LW4 intersection join query from~\eqref{eq:LW4}.

\paragraph*{Class 3: $\fhtw = \subw = 1.5$}
\begin{equation}
    Q_{\mathrm{LW4}}^{(3)} :=
    \tilde R(A_1, B_1, C_1) \wedge
	\tilde S(B_1, C_1, D_1, B_2, C_2) \wedge
	\tilde T(C_1, D_1, A_1, C_2, D_2, A_2) \wedge
	\tilde U(D_1, A_1, B_1, D_2, A_2, B_2)
\end{equation}

The above query has a fractional edge cover number $\rho^*$ of 1.5 obtained by assigning the coefficients [0, 1/2, 1/2, 1/2]. The $\fhtw$ and $\subw$ are also 1.5.

\paragraph*{Class 4: $\fhtw = \subw = 1.5$}
\begin{equation}
    Q_{\mathrm{LW4}}^{(4)} :=
    \tilde R(A_1, B_1, C_1, B_2) \wedge
	\tilde S(B_1, C_1, D_1, C_2) \wedge
	\tilde T(C_1, D_1, A_1, C_2, D_2, A_2) \wedge
	\tilde U(D_1, A_1, B_1, D_2, A_2, B_2)
\end{equation}

The above query has a $\fhtw$ of 1.5 obtained through a tree decomposition consisting of the following two bags:
\begin{itemize}
    \item Bag $\{A_1, B_1, C_1, D_1, A_2, B_2, D_2\}$ containing relations $\tilde R$ and $\tilde U$
    and has a fractional edge cover number
    $\rho^*$ of 1.5 obtained by assigning the following coefficients: $[1/2, 0, 1/2, 1/2]$.
    \item Bag $\{A_1, B_1, C_1, D_1, A_2, C_2, D_2\}$ containing relations $\tilde S$ and $\tilde T$ and
    has a fractional edge cover number
    $\rho^*$ of 1.5 obtained by assigning the following coefficients: $[0, 1/2, 1/2, 1/2]$.
\end{itemize}
$\subw$ is also 1.5.

\paragraph*{Class 5: $\fhtw = \subw = 1.5$}
\begin{equation}
    Q_{\mathrm{LW4}}^{(5)} :=
    \tilde R(A_1, B_1, C_1, A_2, B_2) \wedge
	\tilde S(B_1, C_1, D_1, C_2) \wedge
	\tilde T(C_1, D_1, A_1, C_2, D_2) \wedge
	\tilde U(D_1, A_1, B_1, D_2, A_2, B_2)
\end{equation}

The above query has a $\fhtw$ of 1.5 obtained through a tree decomposition consisting of the following two bags:
\begin{itemize}
    \item Bag $\{A_1, B_1, C_1, D_1, A_2, B_2, D_2\}$ containing relations
    $\tilde R$ and $\tilde U$ and has a fractional edge cover number
    $\rho^*$ of 1.5 obtained by assigning the coefficients [1/2, 0, 1/2, 1/2].
    \item Bag $\{A_1, B_1, C_1, D_1, C_2, D_2\}$ containing relations $\tilde S$ and $\tilde T$ and
    has a fractional edge cover number $\rho^*$
    of 1.5 obtained by assigning the coefficients [0, 1/2, 1/2, 1/2].
\end{itemize}
$\subw$ is also 1.5.

\paragraph*{Class 6: $\fhtw = \subw = 1.5$}
\begin{equation}
    Q_{\mathrm{LW4}}^{(6)} :=
	\tilde R(A_1, B_1, C_1, B_2, C_2) \wedge
	\tilde S(B_1, C_1, D_1, B_2, C_2) \wedge
	\tilde T(C_1, D_1, A_1, D_2, A_2) \wedge
	\tilde U(D_1, A_1, B_1, D_2, A_2)
\end{equation}

The above query has $\fhtw$ of 1.5 using a tree decomposition of the following two bags:
\begin{itemize}
    \item Bag $\{A_1, B_1, C_1, D_1, B_2, C_2\}$ containing $\tilde R$ and $\tilde S$
    and has a $\rho^*$ of 1.5 using [1/2, 1/2, 1/2, 0].
    \item Bag $\{A_1, B_1, C_1, D_1, A_2, D_2\}$ containing $\tilde T$ and $\tilde U$
    and has a $\rho^*$ of 1.5 using [1/2, 0, 1/2, 1/2].
\end{itemize}
$\subw$ is also 1.5.

Finally note that query $Q_{\mathrm{LW4}}$ from~\eqref{eq:LW4} involves four interval variables
$[A], [B], [C]$ and $[D]$, each of which appears in three relations. Hence under our reduction,
each one of the four variables contributes an extra factor of $\log^2 N$ to the time complexity
thus resulting in an overall bound of $O(N^{5/3} \log^8 N)$.

\subsection{The 4-clique intersection join query}
The 4-clique intersection join query looks as follows:

\begin{equation}
    Q_{\mathrm{4clique}} :=
	R([A], [B]) \wedge
    S([A], [C]) \wedge
    T([A], [D]) \wedge
    U([B], [C]) \wedge
    V([B], [D]) \wedge
    W([C], [D])
    \label{eq:4clique}
\end{equation}
We show in this section that our approach can solve this query in time $O(N^2 \log^8 N)$
while the $\faqai$ approach takes time $O(N^3 \log^k N)$ for some constant $k \geq 5$.

\subsubsection{The $\faqai$ approach takes time $O(N^3 \log^k N)$ for $k \geq 5$}
To apply the $\faqai$ approach on $Q_{\mathrm{4clique}}$, we follow the script of
Sections~\ref{app:ex:triangle} and~\ref{app:ex:LW4}.
In particular, we start by defining $F(X)$ to be the set of relations containing variable $X$:
\begin{eqnarray*}
    F(A) &:=& \{R, S, T\},\\
    F(B) &:=& \{R, U, V\},\\
    F(C) &:=& \{S, U, W\},\\
    F(D) &:=& \{T, V, W\}.
\end{eqnarray*}

Now query $Q_{\mathrm{4clique}}$ can be written as follows:
\begin{eqnarray}
    Q_{\mathrm{4clique}} =
    \bigvee_{(X_A, X_B, X_C, X_D)\in F(A)\times F(B)\times F(C)\times F(D)}
    &R\left(A.l^{(R)}, A.r^{(R)}, B.l^{(R)}, B.r^{(R)}\right) \wedge
     S\left(A.l^{(S)}, A.r^{(S)}, C.l^{(S)}, C.r^{(S)}\right) \wedge\nonumber\\
    &T\left(A.l^{(T)}, A.r^{(T)}, D.l^{(T)}, D.r^{(T)}\right) \wedge
     U\left(B.l^{(U)}, B.r^{(U)}, C.l^{(U)}, C.r^{(U)}\right) \wedge\nonumber\\
    &V\left(B.l^{(V)}, B.r^{(V)}, D.l^{(V)}, D.r^{(V)}\right) \wedge
     W\left(C.l^{(W)}, C.r^{(W)}, D.l^{(W)}, D.r^{(W)}\right) \wedge\nonumber\\
    &\displaystyle{\bigwedge_{Y_A \in F(A)-\{X_A\}}} A.l^{(Y_A)} \leq A.l^{(X_A)} \leq A.r^{(Y_A)}\wedge\nonumber\\
    &\displaystyle{\bigwedge_{Y_B \in F(B)-\{X_B\}}} B.l^{(Y_B)} \leq B.l^{(X_B)} \leq B.r^{(Y_B)}\wedge\nonumber\\
    &\displaystyle{\bigwedge_{Y_C \in F(C)-\{X_C\}}} C.l^{(Y_C)} \leq C.l^{(X_C)} \leq C.r^{(Y_C)}\wedge\nonumber\\
    &\displaystyle{\bigwedge_{Y_D \in F(D)-\{X_D\}}} D.l^{(Y_D)} \leq D.l^{(X_D)} \leq D.r^{(Y_D)}
    \label{eq:4C:faqai}
\end{eqnarray}

The disjunction in~\eqref{eq:4C:faqai} above contains $(3^4 = 81)$ disjuncts each of which is an
$\faqai$ query. Consider the specific disjunct that corresponds to
$X_A = R, X_B = U, X_C = S, X_D = T$. Let's call it $\bar Q$:
\begin{eqnarray}
    \bar Q :=
    &R\left(A.l^{(R)}, A.r^{(R)}, B.l^{(R)}, B.r^{(R)}\right) \wedge
     S\left(A.l^{(S)}, A.r^{(S)}, C.l^{(S)}, C.r^{(S)}\right) \wedge
     T\left(A.l^{(T)}, A.r^{(T)}, D.l^{(T)}, D.r^{(T)}\right) \wedge\nonumber\\
    &U\left(B.l^{(U)}, B.r^{(U)}, C.l^{(U)}, C.r^{(U)}\right) \wedge
     V\left(B.l^{(V)}, B.r^{(V)}, D.l^{(V)}, D.r^{(V)}\right) \wedge
     W\left(C.l^{(W)}, C.r^{(W)}, D.l^{(W)}, D.r^{(W)}\right) \wedge\nonumber\\
    &\left(A.l^{(S)}\leq A.l^{(R)}\leq A.r^{(S)}\right)\wedge\left(A.l^{(T)}\leq A.l^{(R)}\leq A.r^{(T)}\right)\wedge\nonumber\\
    &\left(B.l^{(R)}\leq B.l^{(U)}\leq B.r^{(R)}\right)\wedge\left(B.l^{(V)}\leq B.l^{(U)}\leq B.r^{(V)}\right)\wedge\nonumber\\
    &\left(C.l^{(U)}\leq C.l^{(S)}\leq C.r^{(U)}\right)\wedge\left(C.l^{(W)}\leq C.l^{(S)}\leq C.r^{(W)}\right)\wedge\nonumber\\
    &\left(D.l^{(V)}\leq D.l^{(T)}\leq D.r^{(V)}\right)\wedge\left(D.l^{(W)}\leq D.l^{(T)}\leq D.r^{(W)}\right)
    \label{eq:4C:faqai:barQ}
\end{eqnarray}
Let $\calI$ be the set of relation pairs that are connected by some inequality in query $\bar Q$:
\begin{eqnarray*}
    \calI :=\{
    \{R, S\},
    \{R, T\},
    \{R, U\},
    \{S, U\},
    \{S, W\},
    \{T, V\},
    \{T, W\},
    \{U, V\}\}
\end{eqnarray*}
Similar to Sections~\ref{app:ex:triangle} and~\ref{app:ex:LW4}, if $\subw_\ell(\bar Q) < 3$,
then there must exist a relaxed tree decomposition of $\bar Q$ where each bag contains at most two
relations:
Otherwise, we could have used the edge dominated polymatroid $\bar h$ below to show that
$\subw_\ell(\bar Q) \geq 3$, in the exact same way we did in the previous two sections:
\begin{equation}
    \bar h(X) := \frac{|X|}{4}, \quad\forall X \subseteq \vars(\bar Q).
\end{equation}
Consider all partitions of relations $\{R, S, T, U, V, W\}$ of $\bar Q$
into bags where each bag contains at most two relations.
We can show that in every one of these partitions,
the query $\bar Q$ contains inequalities that form a cycle among the bags.
Hence no matter how we try to arrange the bags into a tree to form a tree decomposition,
there will be at least one inequality between two non-adjacent bags in the tree
thus violating the condition for a relaxed tree decomposition.
In particular, there are 76 such partitions. We enumerate all of them using a computer
program and check that each partition contains a cycle of inequalities among the bags.
This proves that any relaxed tree decomposition must contain a bag with at least 3 relations,
hence $\subw_\ell(\bar Q) \geq 3$. Table~\ref{tab:faqai:clique4} lists all partitions
of relations $\{R, S, T, U, V, W\}$ of $\bar Q$ into 3 bags with exactly two relations in each.
Basic combinatorics show\footnote{There are 6! ways to partition 6 items
into an (ordered) tuple of 3 (ordered) tuples of size 2 each,
i.e.~$((X_1, X_2), (X_3, X_4), (X_5, X_6))$. Consequently there are $6!/(2^3) = 90$ ways
to partition 6 items into an (ordered) tuple of 3 (unordered) sets of size 2 each,
i.e.~$(\{X_1, X_2\}, \{X_3, X_4\}, \{X_5, X_6\})$.
Finally, there are $90/(3!) = 15$ possible ways to partition 6 items into an (unordered) set of
3 (unordered) sets of size 2 each, i.e.~$\{\{X_1, X_2\}, \{X_3, X_4\}, \{X_5, X_6\}\}$.}
that there are 15 such partitions (out of the 76 partitions in total that we need to consider).
Table~\ref{tab:faqai:clique4} shows a triangle of inequalities for every one of these 15 partitions.

To prove that $\fhtw_\ell(\bar Q) \leq 3$, we can use a relaxed tree decomposition of two (adjacent)
bags, each of which contains 3 relations. No matter what inequalities are there in $\bar Q$,
every inequality is covered by (the only) two adjacent bags, hence this is a valid relaxed
tree decomposition.
Since $\subw_\ell(Q)\leq \fhtw_\ell(Q)$ for any query $Q$ according to~\cite{FAQAI:TODS:2020},
this proves that
\[
    3 \leq \subw_\ell(\bar Q) \leq \fhtw_\ell(\bar Q) \leq 3,
\]
where all the inequalities above are equalities.

\begin{table*}
\begin{tabular}{|c|c|}
    \hline
    Possible partitions of $\{R, S, T, U, V, W\}$&
    3 edges in $\calI$ connecting\\
    into 3 parts of size 2 each&
    every 2 parts in the partition\\
    \hline
    \{\{R, W\}, \{S, U\}, \{T, V\}\}& \{R, S\}, \{R, T\}, \{U, V\}\\
    \{\{R, U\}, \{S, T\}, \{V, W\}\}& \{R, S\}, \{U, V\}, \{S, W\}\\
    \{\{R, S\}, \{T, V\}, \{U, W\}\}& \{R, T\}, \{R, U\}, \{T, W\}\\
    \{\{R, V\}, \{S, T\}, \{U, W\}\}& \{R, S\}, \{R, U\}, \{S, U\}\\
    \{\{R, W\}, \{S, V\}, \{T, U\}\}& \{R, S\}, \{R, T\}, \{S, U\}\\
    \{\{R, T\}, \{S, V\}, \{U, W\}\}& \{R, S\}, \{R, U\}, \{S, U\}\\
    \{\{R, U\}, \{S, V\}, \{T, W\}\}& \{R, S\}, \{R, T\}, \{S, W\}\\
    \{\{R, V\}, \{S, W\}, \{T, U\}\}& \{R, S\}, \{R, T\}, \{S, U\}\\
    \{\{R, V\}, \{S, U\}, \{T, W\}\}& \{R, S\}, \{R, T\}, \{S, W\}\\
    \{\{R, T\}, \{S, W\}, \{U, V\}\}& \{R, S\}, \{R, U\}, \{S, U\}\\
    \{\{R, U\}, \{S, W\}, \{T, V\}\}& \{R, S\}, \{R, T\}, \{T, W\}\\
    \{\{R, S\}, \{T, U\}, \{V, W\}\}& \{R, T\}, \{S, W\}, \{T, V\}\\
    \{\{R, T\}, \{S, U\}, \{V, W\}\}& \{R, S\}, \{T, V\}, \{S, W\}\\
    \{\{R, S\}, \{T, W\}, \{U, V\}\}& \{R, T\}, \{R, U\}, \{T, V\}\\
    \{\{R, W\}, \{S, T\}, \{U, V\}\}& \{R, S\}, \{R, U\}, \{S, U\}\\\hline
\end{tabular}
\caption{A proof that $\bar Q$ from~\eqref{eq:4C:faqai:barQ}
does not admit a relaxed tree decomposition with exactly two relations
in each bag. The left column shows all possible ways to partition the relations of $\bar Q$
into 3 bags of size 2. The right column shows 3 inequalities in $\bar Q$ connecting every pair of these 3 bags,
thus violating the definition of a relaxed tree decomposition.}
\label{tab:faqai:clique4}
\end{table*}

Finally in query $\bar Q$ from~\eqref{eq:4C:faqai:barQ}, for any optimal relaxed tree decomposition,
the minimum number $k$ of inequalities between two adjacent bags that can be achieved is $k=6$.
The $\faqai$ approach~\cite{FAQAI:TODS:2020} incurs an extra factor of $(\log N)^{\max(k-1, 1)}$
which corresponds to $\log^5 N$ for $\bar Q$. Other queries corresponding to different choices
of $(X_A, X_B, X_C, X_D)$ in~\eqref{eq:4C:faqai} might have bigger $k$-values.
Hence the overall $\faqai$ time complexity for $Q_{\mathrm{4clique}}$ is $O(N^3 \log^k N)$ for some $k \geq 5$.

\subsubsection{Our approach takes time $O(N^2 \log^8 N)$}

Using the reduction presented in this work, we can solve this query in time $O(N^2 \log^8 N)$.
In particular, the reduction produces a number of equality join queries.
Similar to what we did in Section~\ref{app:ex:LW4}, we drop singleton variables from these queries.
Consequently, we reduce their number down to 81 queries and we group them into the following 6 isomorphism classes. (We refer to relations resulting from
$R, S, \ldots$ after the reduction as $\tilde R, \tilde S, \ldots$)

\paragraph*{Class 1: $\fhtw = \subw = 2$}
\begin{equation}
    Q_1 :=
    \tilde R(A_1, B_1) \wedge
    \tilde S(A_1, C_1, A_2) \wedge
    \tilde T(A_1, D_1, A_2) \wedge
    \tilde U(B_1, C_1, B_2, C_2) \wedge
    \tilde V(B_1, D_1, B_2, D_2) \wedge
    \tilde W(C_1, D_1, C_2, D_2)
\end{equation}

The above query accepts a tree decomposition of two bags, witnessing that $\fhtw$ is at most 2
(and it can be shown to be exactly 2):
\begin{itemize}
    \item $\{A_1, A_2, B_1, B_2, C_1, C_2, D_1\}$ has $\rho^*$ of 2.0 achieved by the following edge cover of
    the relations $[\tilde R, \tilde S, \tilde T, \tilde U, \tilde V, \tilde W]$ in order: [0.0, 0.0, 1.0, 1.0, 0.0, 0.0].
    \item $\{B_1, B_2, C_1, C_2, D_1, D_2\}$ has $\rho^*$ of 1.5 achieved by the edge cover
    [0.0, 0.0, 0.0, 0.5, 0.5, 0.5].
\end{itemize}

\paragraph*{Class 2: $\fhtw = \subw = 2$}
\begin{equation}
    Q_2 :=
	\tilde R(A_1, B_1, B_2) \wedge
    \tilde S(A_1, C_1, A_2) \wedge
    \tilde T(A_1, D_1, A_2) \wedge
    \tilde U(B_1, C_1, C_2) \wedge
    \tilde V(B_1, D_1, B_2, D_2) \wedge
    \tilde W(C_1, D_1, C_2, D_2)
\end{equation}

$Q_2$ accepts a tree decomposition with two bags:
\begin{itemize}
    \item $\{A_1, A_2, B_1, B_2, C_1, D_1, D_2\}$ with fractional edge cover
    [0.0, 1.0, 0.0, 0.0, 1.0, 0.0].
    \item $\{B_1, C_1, C_2, D_1, D_2\}$ with edge cover [0.0, 0.0, 0.0, 0.5, 0.5, 0.5].
\end{itemize}

\paragraph*{Class 3: $\fhtw = \subw = 2$}
\begin{equation}
    Q_3 :=
	\tilde R(A_1, B_1, A_2, B_2) \wedge
    \tilde S(A_1, C_1) \wedge
    \tilde T(A_1, D_1, A_2) \wedge
    \tilde U(B_1, C_1, C_2) \wedge
    \tilde V(B_1, D_1, B_2, D_2) \wedge
    \tilde W(C_1, D_1, C_2, D_2)
\end{equation}

$Q_3$ accepts a tree decomposition with a single bag using the edge cover:
[1.0, 0.0, 0.0, 0.0, 0.0, 1.0].

\paragraph*{Class 4: $\fhtw = \subw = 2$}
\begin{equation}
    Q_4 :=
    \tilde R(A_1, B_1, A_2, B_2) \wedge
    \tilde S(A_1, C_1, A_2) \wedge
    \tilde T(A_1, D_1) \wedge
    \tilde U(B_1, C_1, C_2) \wedge
    \tilde V(B_1, D_1, B_2, D_2) \wedge
    \tilde W(C_1, D_1, C_2, D_2)
\end{equation}

$Q_4$ accepts a tree decomposition with a single bag using the edge cover:
[1.0, 0.0, 0.0, 0.0, 0.0, 1.0].

\paragraph*{Class 5: $\fhtw = \subw = 2$}
\begin{equation}
    Q_5 :=
	\tilde R(A_1, B_1, A_2, B_2) \wedge
    \tilde S(A_1, C_1, A_2, C_2) \wedge
    \tilde T(A_1, D_1) \wedge
    \tilde U(B_1, C_1) \wedge
    \tilde V(B_1, D_1, B_2, D_2) \wedge
    \tilde W(C_1, D_1, C_2, D_2)
\end{equation}

$Q_5$ accepts a tree decomposition with a single bag using the edge cover:
[1.0, 0.0, 0.0, 0.0, 0.0, 1.0].

\paragraph*{Class 6: $\fhtw = \subw = 2$}
\begin{equation}
    Q_6 :=
	\tilde R(A_1, B_1, A_2, B_2) \wedge
    \tilde S(A_1, C_1, C_2) \wedge
    \tilde T(A_1, D_1, A_2) \wedge
    \tilde U(B_1, C_1, B_2) \wedge
    \tilde V(B_1, D_1, D_2) \wedge
    \tilde W(C_1, D_1, C_2, D_2)
\end{equation}

$Q_6$ accepts a tree decomposition with a single bag using the edge cover:
[1.0, 0.0, 0.0, 0.0, 0.0, 1.0].

Finally similar to query $Q_{\mathrm{LW4}}$ from the previous section,
query $Q_{\mathrm{4clique}}$ involves four interval variables
each of which appears in exactly three relations.
Therefore according to our reduction,
each variable contributes an extra factor of $\log^2 N$ to the time complexity.
The overall runtime is $O(N^2 \log^8 N)$.

\section{Improved Intersection Predicate Rewriting to Yield Disjoint Conjuncts}
\label{appendix:disjoint}

In this section, we discuss an alternative rewriting of the intersection predicate that ensures disjointness of the conjuncts in the resulting disjunction. This is essential to ensure that our forward reduction yields disjoint conjunctive queries with equality joins in the output disjunction, as required for efficient enumeration and aggregate computation.

\subsection{Distinct Left Endpoints}
\label{appendix:subsection:distinct-endpoints}

We show how to ensure that any two intervals from any two different relations have distinct left endpoints, without affecting query evaluation.


Let $n = |\calE|$ and denote the relations from the database by $R_1, \dots, R_n$. Since the database $\mathbf{D}$ is finite, there exists a sufficiently small real number $\epsilon > 0$ such that $n * \epsilon$ is strictly smaller than the distance 
between the two distinct endpoints of any two intervals in the data. Now, for each $i\in[n]$, any interval $x$ from the relation $R_i$ is replaced 
with the interval $[x.l + i*\epsilon, x.r + n * \epsilon]$. This is a valid interval since $i \leq n$.
It holds that, after this data modification, the intersection joins behave exactly the same as before. 
To see this, let $x$ be an interval from $R_i$ and let $y$ be an interval from $R_j$, with $i \neq j$. If $y.l < x.l$ then $y.l + j*\epsilon < x.l+ i*\epsilon$, since $j \leq n$ and $n * \epsilon < |x.l - y.l|$ 
(by our choice of $\epsilon$). If $y.l = x.l$ then $y.l + j*\epsilon < x.l + i*\epsilon$, or $x.l + i * \epsilon < y.l + j * \epsilon$, 
depending on whether $j < i$ or $j > i$. If $x.l \leq y.r$ then it holds that $x.l + i*\epsilon \leq y.r + n*\epsilon$, since $i \leq n$. 
Also, if $y.r < x.l$ then it holds that $y.r + n*\epsilon < x.l + i*\epsilon$, since $n*\epsilon < |x.l - y.r|$ (by the choice of $\epsilon$).
    
Moreover, it is trivial to recover the original interval: just subtract from the endpoints of the interval the values $i*\epsilon$ and $n*\epsilon$, respectively. 
Therefore, by applying this data transformation, the answer of $Q(\mathbf{D})$ remains exactly the same.

\subsection{Rewriting of the Intersection Predicate with Disjoint Conjuncts}
\label{appenndix:subsection:disjoint-rewriting}

We rewrite the intersection predicate from Section~\ref{section:intersection-predicate-rewriting} in a way such that, given a set of intervals
$S = \{x_{1}, \dots, x_{k}\}$ there exists precisely one $\sigma \in \pi(S)$ that satisfies the
equivalence of Lemma~\ref{lemma:intersection-predicate-ordered}. 

As mentioned in Section~\ref{section:intersection-predicate-rewriting}, 
by Property~\ref{property:unique-tuple-of-nodes}, we still have that, 
for a fixed permutation $\sigma \in \pi(S)$, 
there can be at most one tuple $(v_{j})_{j \in [k-1]}$ that satisfies the conjunction. 
However, even if the intervals from $S$ have distinct left endpoints, 
the predicate of Lemma~\ref{lemma:intersection-predicate-ordered} may be satisfied by multiple permutations. 
To see this, suppose that $\sigma$ and $(v_{j})_{j \in [k-1]}$ satisfy the predicate. 
If there is $j \in [k-1]$ such that $v_j = v_{j+1}$, then the permutation $\sigma'$, 
obtained from $\sigma$ by swapping $\sigma_j$ and $\sigma_{j+1}$, 
together with the same tuple $(v_{j})_{j \in [k-1]}$ satisfy the predicate as well. 
Hence, we will allow for $v_j = v_{j+1}$ only in those permutations $\sigma$ that have $\sigma_j < \sigma_{j+1}$.

Given a segment tree $\mathfrak{T}_{\calI}$, let $\text{sanc}(u)$ denote the set of 
        strict ancestors of $u$ i.e. $u$ is not included in the set.
We next define the following set of tuples that is needed to formulate our alternative rewriting.

\begin{definition}[Ordered Tuples Set]
\label{def:ordered-tuples-set}
Given a permutation $\sigma \in \pi(\{1, \dots, k\})$ and a point $p \in \mathbb{R}$, 
define $\text{OT}_\mathcal{I}(\sigma, p)$ to be the set of all tuples 
$(v_1, \dots, v_k) \in (V(\mathfrak{T}_{\calI}))^{k}$ such that:
\begin{itemize}
    \item $v_k := \text{leaf}(p)$,
    \item $v_{k-1} \in \text{anc}(v_k)$, and
    \item for each $1 < j < k$: 
    \begin{itemize}
        \item $v_{j-1} \in \text{anc}(v_j)$ if $\sigma_{j-1} < \sigma_j$, and
        \item $v_{j-1} \in \text{sanc}(v_j)$ otherwise.
    \end{itemize}
\end{itemize}
That is, for each tuple $(v_{1}, \dots, v_{k})$, the pairs $(v_1, \sigma_1), \dots, (v_{k-1}, \sigma_{k-1})$ 
form a strictly increasing sequence with respect to $\prec$, where $(v_x, \sigma_x) \prec (v_y, \sigma_y)$ 
if and only if $v_x \in \text{sanc}(v_y)$, or $v_x = v_y$ and $\sigma_x < \sigma_y$.
\end{definition}

\begin{lemma}[Disjoint Intersection Predicate 1]
\label{lemma:disjoint-intersection-predicate-1}
For any set of intervals $S = \{x_{1}, \dots, x_{k}\} \subseteq \mathcal{I}$, we have:
\begin{equation}
\label{eq:interection-predicate-final}
    \left( 
        \bigcap_{i \in [k]} x_i
    \right)    
    \neq \emptyset 
    \equiv \bigvee_{\sigma \in \pi(\{1, \dots, k\})} \left( \bigvee_{(v_1, \dots, v_k) \in \text{OT}_\mathcal{I}(\sigma, x_{\sigma_{k}}.l)} \left( \bigwedge_{j \in [k-1]} v_j \in \text{CP}_\mathcal{I}(x_{\sigma_j}) \right) \right)
\end{equation}
Moreover, if the intervals from $S$ have distinct left endpoints, then the right hand side predicate of Equation~\eqref{eq:interection-predicate-final} can be satisfied by at most one permutation $\sigma \in \pi(\{1, \dots, k\})$ 
and one tuple $(v_1, \dots, v_k) \in \text{OT}_\mathcal{I}(\sigma, x_{\sigma_k})$.

\begin{proof}
    I) "$\Leftarrow$": Assume that the permutation $\sigma \in \pi(\{1, \dots, k\})$ and the tuple $(v_1, \dots, v_k) \in \text{OT}_\mathcal{I}(\sigma, x_{\sigma_{k}}.l)$ 
    satisfy the right hand side predicate. This means that, for each $1 \leq j < k$, the node $v_j$ is an ancestor of $\text{leaf}(x_{\sigma_{k}}.l)$ such that $v_j \in \text{CP}_\mathcal{I}(x_{\sigma_{j}})$. 
    This means that $x_{\sigma_{k}}.l \in x_{\sigma_{j}}$, for each $1 \leq j < k$. 
    Hence, the intervals in $S$ intersect and the left endpoint of the intersection is precisely the point $x_{\sigma_k}.l$.
    
    II) "$\Rightarrow$": Assume that the intervals in $S$ intersect. 
    Since all the intervals have distinct left endpoints, 
    the expression $i := \text{arg max}_{1 \leq i \leq k} \, x_i.l$ has only one solution, 
    and hence, the interval $x_i$ is the only interval satisfying that the point $x_i.l$ 
    is contained in all other intervals in $S$. Therefore, all the permutations whose 
    last element is not $i$ cannot satisfy the right hand side predicate. 
    Moreover, by Property~\ref{property:unique-tuple-of-nodes}, 
    for each $1 \leq j \leq k$ with $j \neq i$, there is exactly one segment tree node 
    $u_j \in \text{CP}_\mathcal{I}(x_j)$ such that $u_j \in \text{anc}(\text{leaf}(x_i.l))$. 
    Define $u_i := \text{leaf}(x_i)$. Hence, only the permutations $\sigma \in \pi(\{1, \dots, k\})$ 
    such that $\sigma_k = i$ and $(u_{\sigma_1}, \dots, u_{\sigma_k}) \in \text{OT}_\mathcal{I}(\sigma, x_i.l)$ 
    can satisfy the right hand side predicate. But, by Definition~\ref{def:ordered-tuples-set} of $\text{OT}_\mathcal{I}(\sigma, x_i.l)$, 
    there is exactly one such permutation: the one that satisfies either $u_{\sigma_{j-1}} \in \text{sanc}(u_{\sigma_j})$ 
    or both $u_{\sigma_{j-1}} = u_{\sigma_j}$ and $\sigma_{j-1} < \sigma_j$, for each $1 < j < k$. Therefore, the predicate 
    from Equation~\eqref{eq:interection-predicate-final} is satisfied by exactly one permutation and one tuple of segment tree nodes.
\end{proof}    
\end{lemma}

\end{document}